\documentclass[reqno,12pt,a4paper]{amsart}

\usepackage{mathrsfs}
\usepackage{amssymb}
\usepackage{amsfonts}
\usepackage{latexsym}
\usepackage{amsthm}

\usepackage[shortlabels]{enumitem}
\usepackage{titletoc}
\usepackage{pict2e}

\numberwithin{equation}{section}

\usepackage{graphicx}
\usepackage{tikz}
\begin{document}

%%%%%%%%%% Some definitions %%%%%%%%%%

%%%%%%%% Equations, theorems %%%%%%%%%
\renewcommand{\theequation}{\arabic{section}.\arabic{equation}}
\theoremstyle{plain}
\newtheorem{theorem}{\bf Theorem}[section]
\newtheorem{lemma}[theorem]{\bf Lemma}
\newtheorem{corollary}[theorem]{\bf Corollary}
\newtheorem{proposition}[theorem]{\bf Proposition}
\newtheorem{definition}[theorem]{\bf Definition}
\newtheorem*{definition*}{\bf Definition}
\newtheorem*{example}{\bf Example}
\theoremstyle{remark}
\newtheorem*{remark}{\bf Remark}

%%%%% Alphabet %%%%%
\def\a{\alpha}  \def\cA{{\mathcal A}}     \def\bA{{\bf A}}  \def\mA{{\mathscr A}}
\def\b{\beta}   \def\cB{{\mathcal B}}     \def\bB{{\bf B}}  \def\mB{{\mathscr B}}
\def\g{\gamma}  \def\cC{{\mathcal C}}     \def\bC{{\bf C}}  \def\mC{{\mathscr C}}
\def\G{\Gamma}  \def\cD{{\mathcal D}}     \def\bD{{\bf D}}  \def\mD{{\mathscr D}}
\def\d{\delta}  \def\cE{{\mathcal E}}     \def\bE{{\bf E}}  \def\mE{{\mathscr E}}
\def\D{\Delta}  \def\cF{{\mathcal F}}     \def\bF{{\bf F}}  \def\mF{{\mathscr F}}
\def\c{\chi}    \def\cG{{\mathcal G}}     \def\bG{{\bf G}}  \def\mG{{\mathscr G}}
\def\z{\zeta}   \def\cH{{\mathcal H}}     \def\bH{{\bf H}}  \def\mH{{\mathscr H}}
\def\e{\eta}    \def\cI{{\mathcal I}}     \def\bI{{\bf I}}  \def\mI{{\mathscr I}}
\def\p{\psi}    \def\cJ{{\mathcal J}}     \def\bJ{{\bf J}}  \def\mJ{{\mathscr J}}
\def\vT{\Theta} \def\cK{{\mathcal K}}     \def\bK{{\bf K}}  \def\mK{{\mathscr K}}
\def\k{\kappa}  \def\cL{{\mathcal L}}     \def\bL{{\bf L}}  \def\mL{{\mathscr L}}
\def\l{\lambda} \def\cM{{\mathcal M}}     \def\bM{{\bf M}}  \def\mM{{\mathscr M}}
\def\L{\Lambda} \def\cN{{\mathcal N}}     \def\bN{{\bf N}}  \def\mN{{\mathscr N}}
\def\m{\mu}     \def\cO{{\mathcal O}}     \def\bO{{\bf O}}  \def\mO{{\mathscr O}}
\def\n{\nu}     \def\cP{{\mathcal P}}     \def\bP{{\bf P}}  \def\mP{{\mathscr P}}
\def\r{\varrho} \def\cQ{{\mathcal Q}}     \def\bQ{{\bf Q}}  \def\mQ{{\mathscr Q}}
\def\s{\sigma}  \def\cR{{\mathcal R}}     \def\bR{{\bf R}}  \def\mR{{\mathscr R}}
\def\S{\Sigma}  \def\cS{{\mathcal S}}     \def\bS{{\bf S}}  \def\mS{{\mathscr S}}
\def\t{\tau}    \def\cT{{\mathcal T}}     \def\bT{{\bf T}}  \def\mT{{\mathscr T}}
\def\f{\phi}    \def\cU{{\mathcal U}}     \def\bU{{\bf U}}  \def\mU{{\mathscr U}}
\def\F{\Phi}    \def\cV{{\mathcal V}}     \def\bV{{\bf V}}  \def\mV{{\mathscr V}}
\def\P{\Psi}    \def\cW{{\mathcal W}}     \def\bW{{\bf W}}  \def\mW{{\mathscr W}}
\def\o{\omega}  \def\cX{{\mathcal X}}     \def\bX{{\bf X}}  \def\mX{{\mathscr X}}
\def\x{\xi}     \def\cY{{\mathcal Y}}     \def\bY{{\bf Y}}  \def\mY{{\mathscr Y}}
\def\X{\Xi}     \def\cZ{{\mathcal Z}}     \def\bZ{{\bf Z}}  \def\mZ{{\mathscr Z}}
\def\O{\Omega}

\newcommand{\mc}{\mathscr {c}}

\newcommand{\gA}{\mathfrak{A}}          \newcommand{\ga}{\mathfrak{a}}
\newcommand{\gB}{\mathfrak{B}}          \newcommand{\gb}{\mathfrak{b}}
\newcommand{\gC}{\mathfrak{C}}          \newcommand{\gc}{\mathfrak{c}}
\newcommand{\gD}{\mathfrak{D}}          \newcommand{\gd}{\mathfrak{d}}
\newcommand{\gE}{\mathfrak{E}}
\newcommand{\gF}{\mathfrak{F}}           \newcommand{\gf}{\mathfrak{f}}
\newcommand{\gG}{\mathfrak{G}}           %\newcommand{\gg}{\mathfrak{g}}
\newcommand{\gH}{\mathfrak{H}}           \newcommand{\gh}{\mathfrak{h}}
\newcommand{\gI}{\mathfrak{I}}           \newcommand{\gi}{\mathfrak{i}}
\newcommand{\gJ}{\mathfrak{J}}           \newcommand{\gj}{\mathfrak{j}}
\newcommand{\gK}{\mathfrak{K}}            \newcommand{\gk}{\mathfrak{k}}
\newcommand{\gL}{\mathfrak{L}}            \newcommand{\gl}{\mathfrak{l}}
\newcommand{\gM}{\mathfrak{M}}            \newcommand{\gm}{\mathfrak{m}}
\newcommand{\gN}{\mathfrak{N}}            \newcommand{\gn}{\mathfrak{n}}
\newcommand{\gO}{\mathfrak{O}}
\newcommand{\gP}{\mathfrak{P}}             \newcommand{\gp}{\mathfrak{p}}
\newcommand{\gQ}{\mathfrak{Q}}             \newcommand{\gq}{\mathfrak{q}}
\newcommand{\gR}{\mathfrak{R}}             \newcommand{\gr}{\mathfrak{r}}
\newcommand{\gS}{\mathfrak{S}}              \newcommand{\gs}{\mathfrak{s}}
\newcommand{\gT}{\mathfrak{T}}             \newcommand{\gt}{\mathfrak{t}}
\newcommand{\gU}{\mathfrak{U}}             \newcommand{\gu}{\mathfrak{u}}
\newcommand{\gV}{\mathfrak{V}}             \newcommand{\gv}{\mathfrak{v}}
\newcommand{\gW}{\mathfrak{W}}             \newcommand{\gw}{\mathfrak{w}}
\newcommand{\gX}{\mathfrak{X}}               \newcommand{\gx}{\mathfrak{x}}
\newcommand{\gY}{\mathfrak{Y}}              \newcommand{\gy}{\mathfrak{y}}
\newcommand{\gZ}{\mathfrak{Z}}             \newcommand{\gz}{\mathfrak{z}}

\def\ve{\varepsilon}   \def\vt{\vartheta}    \def\vp{\varphi}    \def\vk{\varkappa}

\def\A{{\mathbb A}} \def\B{{\mathbb B}} \def\C{{\mathbb C}}
\def\dD{{\mathbb D}} \def\E{{\mathbb E}} \def\dF{{\mathbb F}} \def\dG{{\mathbb G}} \def\H{{\mathbb H}}\def\I{{\mathbb I}} \def\J{{\mathbb J}} \def\K{{\mathbb K}} \def\dL{{\mathbb L}}\def\M{{\mathbb M}} \def\N{{\mathbb N}} \def\O{{\mathbb O}} \def\dP{{\mathbb P}} \def\R{{\mathbb R}} \def\dQ{{\mathbb Q}}
\def\S{{\mathbb S}} \def\T{{\mathbb T}} \def\U{{\mathbb U}} \def\V{{\mathbb V}}\def\W{{\mathbb W}} \def\X{{\mathbb X}} \def\Y{{\mathbb Y}} \def\Z{{\mathbb Z}}

\newcommand{\1}{\mathbbm 1}
\newcommand{\dd}{\, \mathrm d}

%%%%%%%%%%%%%%%%%%%%%%%

%%%%% Arrows %%%%%

\def\la{\leftarrow}              \def\ra{\rightarrow}            \def\Ra{\Rightarrow}
\def\ua{\uparrow}                \def\da{\downarrow}
\def\lra{\leftrightarrow}        \def\Lra{\Leftrightarrow}

%%%%% Typography %%%%%

\def\lt{\biggl}                  \def\rt{\biggr}
\def\ol{\overline}               \def\wt{\widetilde}
\def\no{\noindent}

%%%%% Math signs %%%%%

\let\ge\geqslant                 \let\le\leqslant
\def\lan{\langle}                \def\ran{\rangle}
\def\/{\over}                    \def\iy{\infty}
\def\sm{\setminus}               \def\es{\emptyset}
\def\ss{\subset}                 \def\ts{\times}
\def\pa{\partial}                \def\os{\oplus}
\def\om{\ominus}                 \def\ev{\equiv}
\def\iint{\int\!\!\!\int}        \def\iintt{\mathop{\int\!\!\int\!\!\dots\!\!\int}\limits}
\def\el2{\ell^{\,2}}             \def\1{1\!\!1}
\def\sh{\sharp}
\def\wh{\widehat}
%%%%% Math operations %%%%%

\def\all{\mathop{\mathrm{all}}\nolimits}
\def\where{\mathop{\mathrm{where}}\nolimits}
\def\as{\mathop{\mathrm{as}}\nolimits}
\def\Area{\mathop{\mathrm{Area}}\nolimits}
\def\arg{\mathop{\mathrm{arg}}\nolimits}
\def\adj{\mathop{\mathrm{adj}}\nolimits}
\def\const{\mathop{\mathrm{const}}\nolimits}
\def\det{\mathop{\mathrm{det}}\nolimits}
\def\diag{\mathop{\mathrm{diag}}\nolimits}
\def\diam{\mathop{\mathrm{diam}}\nolimits}
\def\dim{\mathop{\mathrm{dim}}\nolimits}
\def\dist{\mathop{\mathrm{dist}}\nolimits}
\def\Im{\mathop{\mathrm{Im}}\nolimits}
\def\Iso{\mathop{\mathrm{Iso}}\nolimits}
\def\Ker{\mathop{\mathrm{Ker}}\nolimits}
\def\Lip{\mathop{\mathrm{Lip}}\nolimits}
\def\rank{\mathop{\mathrm{rank}}\limits}
\def\Ran{\mathop{\mathrm{Ran}}\nolimits}
\def\Re{\mathop{\mathrm{Re}}\nolimits}
\def\Res{\mathop{\mathrm{Res}}\nolimits}
\def\res{\mathop{\mathrm{res}}\limits}
\def\sign{\mathop{\mathrm{sign}}\nolimits}
\def\supp{\mathop{\mathrm{supp}}\nolimits}
\def\Tr{\mathop{\mathrm{Tr}}\nolimits}
\def\AC{\mathop{\rm AC}\nolimits}
\def\pr{\mathop{\mathrm{pr}}\nolimits}
\def\esssup{\displaystyle \mathrm{ess} \sup}
\def\mes{\mathop{\mathrm{mes}}\nolimits}
\def\BBox{\hspace{1mm}\vrule height6pt width5.5pt depth0pt \hspace{6pt}}

%%%%%%%%%%%%% specialities %%%%%%%%%%%%%%

\newcommand\nh[2]{\widehat{#1}\vphantom{#1}^{(#2)}}
% {{\mathop{#1}\limits^\wedge}\vphantom{#1}^{(#2)}}
\def\dia{\diamond}

\def\Oplus{\bigoplus\nolimits}

%%%%%%%%%%% End of definitions %%%%%%%%%%

%\DeclareMathOperator{\Ric}    {Ric}
%\newcommand{\dd}    {\, \mathrm d}    % not optimal: no \, if at beginning

%%%%% OLD OLD OLD

\def\qqq{\qquad}
\def\qq{\quad}
\let\ge\geqslant
\let\le\leqslant
\let\geq\geqslant
\let\leq\leqslant

\newcommand{\ca}{\begin{cases}}
\newcommand{\ac}{\end{cases}}
\newcommand{\ma}{\begin{pmatrix}}
\newcommand{\am}{\end{pmatrix}}
\renewcommand{\[}{\begin{equation}}
\renewcommand{\]}{\end{equation}}
\def\bu{\bullet}

\title[{Periodic Dirac operator with dislocation}]
{Periodic Dirac operator with dislocation}

\date{\today}

\author[Evgeny Korotyaev]{Evgeny Korotyaev}
\address{Department of Analysis,  Saint Petersburg State University,   Universitetskaya nab. 7/9, St.
Petersburg, 199034, Russia, \ korotyaev@gmail.com, \ e.korotyaev@spbu.ru}
\author[Dmitrii Mokeev]{Dmitrii Mokeev}
\address{Saint Petersburg State University,   Universitetskaya nab. 7/9, St.
Petersburg, 199034, Russia, \ mokeev.ds@yandex.ru}


\begin{abstract}
    We consider a Dirac operator with a dislocation potential on the
    real line. The dislocation potential is a fixed periodic potential
    on the negative half-line and the same potential but shifted by real
    parameter $t$ on the positive half-line. Its spectrum has an
    absolutely continuous part (the union of bands separated by gaps)
    plus at most two eigenvalues in each non-empty gap. Its resolvent
    admits a meromorphic continuation onto a two-sheeted Riemann
    surface. We prove that it has only two simple poles on each open
    gap: on the first sheet (an eigenvalue) or on the second sheet (a
    resonance). These poles are called states and there are no other
    poles. We prove: 1) each state is a continuous function of $t$, and
    we obtain its local asymptotic; 2) for each $t$ states in the gap are
    distinct; 3) in general, a state is non-monotone function of $t$ but it
    can be monotone for specific potentials; 4) we construct
    examples of operators, which have: a) one eigenvalue and one resonance
    in any finite number of gaps; b) two eigenvalues or two resonances
    in any finite number of gaps; c) two static virtual states
    in one gap.
\end{abstract}

\maketitle

\section{Introduction} \label{p0}

    \subsection{Dirac operator}

    We consider a Dirac operator $H_t$ with a dislocation potential
    $V_t(x)$ on $L^2(\R, \C^2)$ given by
        $$
        \begin{aligned}
            H_t f = J f' + V_t f,\ \ f = \ma f_1 \\ f_2 \am,\ \ J = \ma 0 & 1 \\ -1 & 0 \am,
        \end{aligned}
        $$
    and
    $$
        V_t(x) = V(x) \c_- (x) + V(x + t) \c_+ (x),\ \ x \in \R.
    $$
    where $t \in \R$ is a dislocation parameter.
    Here $\c_{\pm}$ are characteristic functions of $\R_{\pm}$ and $V$ is a $2 \ts 2$ matrix-valued
    1-periodic function, which belongs to the real Hilbert space $\cP$ defined by
    $$
        \cP = \rt\{ \, \left. V= \ma q_1 & q_2 \\ q_2 & -q_1 \am \, \right| \, q_1,q_2
        \in L_{real}^2(\T) \, \rt\},\qq         \T = \R/\Z,
    $$
    equipped with the norm $\left\| V \right\|_{\cP}^2 ={1\/2} \int_{\T} \Tr V^2(x)dx$. Below we
    show that the operator $H_t$ is self-adjoint and its spectrum consists of an absolutely
    continuous part $\s_{ac}(H_t) = \s(H_0) = \cup_{n \in \Z} \s_n$ plus at most two eigenvalues
    in each non-empty gap $\g_n$, where the bands $\s_n$ and gaps $\g_n$ are given by
    $$
        \s_n = [\a_{n-1}^+,\a_n^-],\qq \g_n = (\a_{n}^-,\a_n^+),\qq \text{satisfying}\qq
        \a_{n-1}^+ < \a_n^- \leq \a_n^+\qq \forall n \in \Z.
    $$
    The sequence $\a_n^\pm$, $n \in \Z$, is the spectrum of the equation $J y' + Vy = \l y$
    with the condition of 2-periodicity, $y(x+2)=y(x)$ $(x\in \R)$.
    If some  gap degenerates, i.e. $\g_n= \es $, then the corresponding bands $\s_{n}$
    and $\s_{n+1}$ touch. This happens when $\a_n^-=\a_n^+$ is a double eigenvalue
    of the 2-periodic problem. Generally, the eigenfunctions corresponding to the eigenvalues
    $\a_{2n}^{\pm}$ are 1-periodic, those for $\a_{2n+1}^{\pm}$ are anti-periodic, i.e.
    $y(x+1)=-y(x)$ ($x\in\R$).

    For the Dirac operator $H_t$ we introduce the two-sheeted Riemann surface $\L$ obtained by joining
    the upper and lower rims of two copies of the cut plane $\C \sm \s(H_0)$ in the usual
    (crosswise) way (see Fig.~1). We denote the $n$-th gap on the first, physical,
    sheet $\L_1$ by $\g_n^{(1)}$ and its counterpart on the second, nonphysical,
    sheet $\L_2$ by $\g_n^{(2)}$, and we set a circle gap
    $\g_n^c=\ol \g_n^{(1)} \cup \ol \g_n^{(2)}$ (see Fig.~2).

    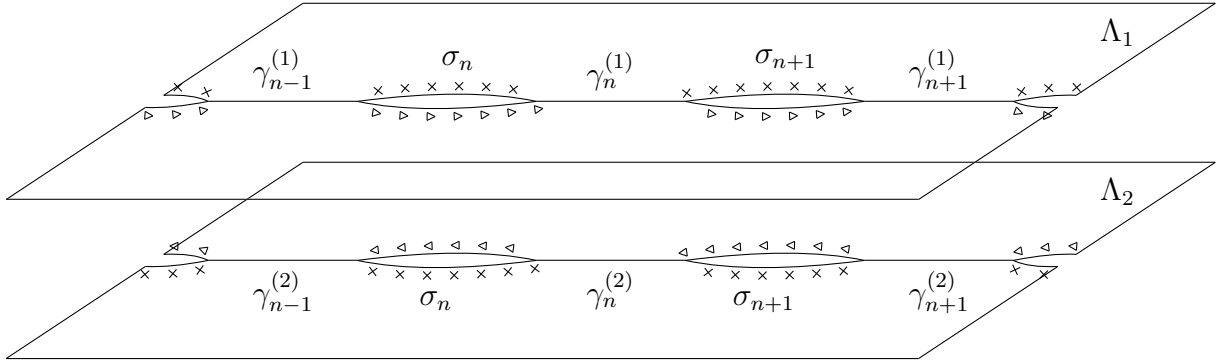
\begin{figure}[h!]
        \begin{center}
            \begin{picture}(452.405,140.078)(0.297988,-5.79967)
              \color[rgb]{1.000000,1.000000,1.000000}
              \qbezier(75.869347,37.287704)(70.168396,40.058857)(59.197762,39.565815)
              \qbezier(59.197762,39.565815)(111.267685,74.277496)(111.267685,74.277496)
              \put(111.267685,74.277496){\line(1,0){341.435608}}
              \qbezier(452.703278,74.277496)(400.633759,39.565815)(400.633759,39.565815)
              \qbezier(400.633759,39.565815)(391.151581,40.058857)(377.135956,37.287704)
              \qbezier(377.135956,37.287704)(382.833740,34.516678)(393.807251,35.009693)
              \qbezier(393.807251,35.009693)(341.737701,0.297988)(341.737701,0.297988)
              \put(341.737701,0.297988){\line(-1,0){341.439728}}
              \qbezier(0.297988,0.297988)(52.367546,35.009693)(52.367546,35.009693)
              \qbezier(52.367546,35.009693)(61.851196,34.516678)(75.869347,37.287704)
              \color[rgb]{0.000000,0.000000,0.000000}
              \qbezier(75.869347,37.287704)(70.168396,40.058857)(59.197762,39.565815)
              \qbezier(59.197762,39.565815)(111.267685,74.277496)(111.267685,74.277496)
              \put(111.267685,74.277496){\line(1,0){341.435608}}
              \qbezier(452.703278,74.277496)(400.633759,39.565815)(400.633759,39.565815)
              \qbezier(400.633759,39.565815)(391.151581,40.058857)(377.135956,37.287704)
              \qbezier(377.135956,37.287704)(382.833740,34.516678)(393.807251,35.009693)
              \qbezier(393.807251,35.009693)(341.737701,0.297988)(341.737701,0.297988)
              \put(341.737701,0.297988){\line(-1,0){341.439728}}
              \qbezier(0.297988,0.297988)(52.367546,35.009693)(52.367546,35.009693)
              \qbezier(52.367546,35.009693)(61.851196,34.516678)(75.869347,37.287704)
              \put(75.869347,37.287704){\line(1,0){55.790535}}
              \put(198.604980,37.287704){\line(1,0){55.790436}}
              \put(321.345520,37.287704){\line(1,0){55.790436}}
              \put(92.455414,21.138973){$\g_{n-1}^{(2)}$}
              \put(217.445358,21.138973){$\g_{n}^{(2)}$}
              \put(337.929474,21.138973){$\g_{n+1}^{(2)}$}
              \qbezier(198.604980,37.287704)(173.292908,42.727982)(131.659882,37.287704)
              \qbezier(131.659882,37.287704)(156.971954,31.847580)(198.604980,37.287704)
              \qbezier(321.345520,37.287704)(296.033203,42.727982)(254.395416,37.287704)
              \qbezier(254.395416,37.287704)(279.712952,31.847580)(321.345520,37.287704)
              \put(409.677551,60.077446){$\L_2$}
              \qbezier(400.735138,41.059696)(397.845337,42.761356)(397.845337,42.761356)
              \qbezier(397.845337,42.761356)(400.938904,44.055187)(400.938904,44.055187)
              \qbezier(400.938904,44.055187)(400.735138,41.059696)(400.735138,41.059696)
              \qbezier(390.513947,40.796917)(387.388275,42.020287)(387.388275,42.020287)
              \qbezier(387.388275,42.020287)(390.238922,43.784679)(390.238922,43.784679)
              \qbezier(390.238922,43.784679)(390.513947,40.796917)(390.513947,40.796917)
              \qbezier(380.480438,39.448177)(377.264465,40.408768)(377.264465,40.408768)
              \qbezier(377.264465,40.408768)(379.962524,42.404522)(379.962524,42.404522)
              \qbezier(379.962524,42.404522)(380.480438,39.448177)(380.480438,39.448177)
              \qbezier(315.358154,39.969627)(312.616943,41.906563)(312.616943,41.906563)
              \qbezier(312.616943,41.906563)(315.812866,42.933800)(315.812866,42.933800)
              \qbezier(315.812866,42.933800)(315.358154,39.969627)(315.358154,39.969627)
              \qbezier(305.296570,41.063610)(302.406799,42.773098)(302.406799,42.773098)
              \qbezier(302.406799,42.773098)(305.508362,44.055187)(305.508362,44.055187)
              \qbezier(305.508362,44.055187)(305.296570,41.063610)(305.296570,41.063610)
              \qbezier(295.219940,41.483078)(292.243805,43.027950)(292.243805,43.027950)
              \qbezier(292.243805,43.027950)(295.267120,44.482582)(295.267120,44.482582)
              \qbezier(295.267120,44.482582)(295.219940,41.483078)(295.219940,41.483078)
              \qbezier(285.155334,41.408600)(282.112946,42.824089)(282.112946,42.824089)
              \qbezier(282.112946,42.824089)(285.069000,44.408104)(285.069000,44.408104)
              \qbezier(285.069000,44.408104)(285.155334,41.408600)(285.155334,41.408600)
              \qbezier(275.109802,40.930214)(272.016235,42.228062)(272.016235,42.228062)
              \qbezier(272.016235,42.228062)(274.910034,43.921890)(274.910034,43.921890)
              \qbezier(274.910034,43.921890)(275.109802,40.930214)(275.109802,40.930214)
              \qbezier(265.092377,40.087265)(261.955658,41.275303)(261.955658,41.275303)
              \qbezier(261.955658,41.275303)(264.790253,43.071110)(264.790253,43.071110)
              \qbezier(264.790253,43.071110)(265.092377,40.087265)(265.092377,40.087265)
              \qbezier(255.113098,38.899223)(251.937225,39.981369)(251.937225,39.981369)
              \qbezier(251.937225,39.981369)(254.705566,41.871231)(254.705566,41.871231)
              \qbezier(254.705566,41.871231)(255.113098,38.899223)(255.113098,38.899223)
              \qbezier(189.865341,40.353859)(187.077911,42.220234)(187.077911,42.220234)
              \qbezier(187.077911,42.220234)(190.241745,43.329876)(190.241745,43.329876)
              \qbezier(190.241745,43.329876)(189.865341,40.353859)(189.865341,40.353859)
              \qbezier(179.800720,41.236057)(176.883835,42.894653)(176.883835,42.894653)
              \qbezier(176.883835,42.894653)(179.961334,44.231647)(179.961334,44.231647)
              \qbezier(179.961334,44.231647)(179.800720,41.236057)(179.800720,41.236057)
              \qbezier(169.728088,41.506668)(166.727890,43.012291)(166.727890,43.012291)
              \qbezier(166.727890,43.012291)(169.732117,44.506073)(169.732117,44.506073)
              \qbezier(169.732117,44.506073)(169.728088,41.506668)(169.728088,41.506668)
              \qbezier(159.666519,41.314548)(156.608078,42.694702)(156.608078,42.694702)
              \qbezier(156.608078,42.694702)(159.549072,44.313953)(159.549072,44.313953)
              \qbezier(159.549072,44.313953)(159.666519,41.314548)(159.666519,41.314548)
              \qbezier(149.629013,40.734180)(146.524414,42.000713)(146.524414,42.000713)
              \qbezier(146.524414,42.000713)(149.398148,43.725857)(149.398148,43.725857)
              \qbezier(149.398148,43.725857)(149.629013,40.734180)(149.629013,40.734180)
              \qbezier(139.619614,39.797081)(136.470840,40.953800)(136.470840,40.953800)
              \qbezier(136.470840,40.953800)(139.290390,42.777012)(139.290390,42.777012)
              \qbezier(139.290390,42.777012)(139.619614,39.797081)(139.619614,39.797081)
              \qbezier(74.720551,39.420673)(72.407211,41.851658)(72.407211,41.851658)
              \qbezier(72.407211,41.851658)(75.739960,42.243721)(75.739960,42.243721)
              \qbezier(75.739960,42.243721)(74.720551,39.420673)(74.720551,39.420673)
              \qbezier(64.432114,41.118420)(61.491428,42.726021)(61.491428,42.726021)
              \qbezier(61.491428,42.726021)(64.545845,44.117920)(64.545845,44.117920)
              \qbezier(64.545845,44.117920)(64.432114,41.118420)(64.432114,41.118420)
              \qbezier(50.771683,33.617794)(53.559498,30.418343)(53.559498,30.418343)
              \qbezier(50.567818,30.622202)(53.763359,33.413933)(53.763359,33.413933)
              \qbezier(60.997383,33.645199)(64.255661,30.928045)(64.255661,30.928045)
              \qbezier(61.267994,30.657534)(63.985146,33.915810)(63.985146,33.915810)
              \qbezier(71.046730,34.868568)(74.520607,32.429756)(74.520607,32.429756)
              \qbezier(71.564262,31.912222)(74.003075,35.386101)(74.003075,35.386101)
              \qbezier(136.164719,34.833233)(138.674088,31.414263)(138.674088,31.414263)
              \qbezier(135.710007,31.869061)(139.128769,34.378437)(139.128769,34.378437)
              \qbezier(146.210236,33.617794)(148.994644,30.414427)(148.994644,30.414427)
              \qbezier(145.998444,30.626118)(149.205414,33.406105)(149.205414,33.406105)
              \qbezier(156.282867,33.115917)(159.235901,30.069338)(159.235901,30.069338)
              \qbezier(156.235687,30.116415)(159.283081,33.068844)(159.283081,33.068844)
              \qbezier(166.348480,33.123749)(169.433990,30.210566)(169.433990,30.210566)
              \qbezier(166.434799,30.124243)(169.347672,33.209969)(169.347672,33.209969)
              \qbezier(176.397018,33.547234)(179.588928,30.751587)(179.588928,30.751587)
              \qbezier(176.597763,30.555555)(179.389191,33.747181)(179.389191,33.747181)
              \qbezier(186.419464,34.339188)(189.704727,31.657372)(189.704727,31.657372)
              \qbezier(186.720581,31.351528)(189.402603,34.641117)(189.402603,34.641117)
              \qbezier(196.405762,35.472324)(199.785385,32.904228)(199.785385,32.904228)
              \qbezier(196.809280,32.500317)(199.377869,35.876232)(199.377869,35.876232)
              \qbezier(261.649506,34.409855)(264.249207,31.057430)(264.249207,31.057430)
              \qbezier(261.273102,31.433836)(264.625610,34.033447)(264.625610,34.033447)
              \qbezier(271.706085,33.421764)(274.541656,30.261456)(274.541656,30.261456)
              \qbezier(271.545502,30.426170)(274.702271,33.257050)(274.702271,33.257050)
              \qbezier(281.779694,33.072758)(284.770905,30.065424)(284.770905,30.065424)
              \qbezier(281.771667,30.073252)(284.778931,33.064926)(284.778931,33.064926)
              \qbezier(291.840302,33.202141)(294.953918,30.320276)(294.953918,30.320276)
              \qbezier(291.957733,30.202635)(294.835510,33.319782)(294.835510,33.319782)
              \qbezier(301.881805,33.727604)(305.100830,30.963379)(305.100830,30.963379)
              \qbezier(302.108673,30.735928)(304.873016,33.954952)(304.873016,33.954952)
              \qbezier(311.892212,34.613716)(315.204590,31.963213)(315.204590,31.963213)
              \qbezier(312.225464,31.633783)(314.875366,34.946960)(314.875366,34.946960)
              \qbezier(376.872986,35.664539)(378.676727,31.821985)(378.676727,31.821985)
              \qbezier(375.853180,32.841496)(379.695526,34.645031)(379.695526,34.645031)
              \qbezier(387.071075,33.511898)(389.960876,30.402683)(389.960876,30.402683)
              \qbezier(386.961670,30.516308)(390.070282,33.402092)(390.070282,33.402092)
              \put(154.882629,20.285786){$\s_n$}
              \put(272.188873,20.285786){$\s_{n+1}$}
              \color[rgb]{1.000000,1.000000,1.000000}
              \qbezier(75.869347,97.289268)(70.168396,100.061272)(59.197762,99.563362)
              \qbezier(59.197762,99.563362)(111.267685,134.278656)(111.267685,134.278656)
              \put(111.267685,134.278656){\line(1,0){341.435608}}
              \qbezier(452.703278,134.278656)(400.633759,99.563362)(400.633759,99.563362)
              \qbezier(400.633759,99.563362)(391.151581,100.061272)(377.135956,97.289268)
              \qbezier(377.135956,97.289268)(382.833740,94.515152)(393.807251,95.011261)
              \qbezier(393.807251,95.011261)(341.737701,60.299477)(341.737701,60.299477)
              \put(341.737701,60.299477){\line(-1,0){341.439728}}
              \qbezier(0.297988,60.299477)(52.367546,95.011261)(52.367546,95.011261)
              \qbezier(52.367546,95.011261)(61.851196,94.515152)(75.869347,97.289268)
              \color[rgb]{0.000000,0.000000,0.000000}
              \qbezier(75.869347,97.289268)(70.168396,100.061272)(59.197762,99.563362)
              \qbezier(59.197762,99.563362)(111.267685,134.278656)(111.267685,134.278656)
              \put(111.267685,134.278656){\line(1,0){341.435608}}
              \qbezier(452.703278,134.278656)(400.633759,99.563362)(400.633759,99.563362)
              \qbezier(400.633759,99.563362)(391.151581,100.061272)(377.135956,97.289268)
              \qbezier(377.135956,97.289268)(382.833740,94.515152)(393.807251,95.011261)
              \qbezier(393.807251,95.011261)(341.737701,60.299477)(341.737701,60.299477)
              \put(341.737701,60.299477){\line(-1,0){341.439728}}
              \qbezier(0.297988,60.299477)(52.367546,95.011261)(52.367546,95.011261)
              \qbezier(52.367546,95.011261)(61.851196,94.515152)(75.869347,97.289268)
              \put(75.869347,97.289268){\line(1,0){55.790535}}
              \put(198.604980,97.289268){\line(1,0){55.790436}}
              \put(321.345520,97.289268){\line(1,0){55.790436}}
              \put(92.455414,105.080582){$\g_{n-1}^{(1)}$}
              \put(217.445358,103.638184){$\g_{n}^{(1)}$}
              \put(337.929474,105.080582){$\g_{n+1}^{(1)}$}
              \qbezier(198.604980,97.289268)(173.292908,102.729385)(131.659882,97.289268)
              \qbezier(131.659882,97.289268)(156.971954,91.848999)(198.604980,97.289268)
              \qbezier(321.345520,97.289268)(296.033203,102.729385)(254.395416,97.289268)
              \qbezier(254.395416,97.289268)(279.712952,91.848999)(321.345520,97.289268)
              \put(409.677551,120.077606){$\L_1$}
              \qbezier(402.232727,100.959183)(399.445313,104.155121)(399.445313,104.155121)
              \qbezier(402.437469,103.951363)(399.241547,101.162941)(399.241547,101.162941)
              \qbezier(392.007507,100.932076)(388.749359,103.649231)(388.749359,103.649231)
              \qbezier(391.733490,103.919235)(389.019379,100.661072)(389.019379,100.661072)
              \qbezier(381.957977,99.708405)(378.484009,102.147621)(378.484009,102.147621)
              \qbezier(381.441040,102.664551)(379.001923,99.190872)(379.001923,99.190872)
              \qbezier(316.839691,99.743736)(314.330322,103.162415)(314.330322,103.162415)
              \qbezier(317.294403,102.707710)(313.875641,100.198540)(313.875641,100.198540)
              \qbezier(306.791138,100.959183)(304.010773,104.162148)(304.010773,104.162148)
              \qbezier(307.001953,103.951363)(303.798981,101.170967)(303.798981,101.170967)
              \qbezier(296.721558,101.461052)(293.765503,104.507431)(293.765503,104.507431)
              \qbezier(296.764709,104.460258)(293.722321,101.508224)(293.722321,101.508224)
              \qbezier(286.656921,101.453018)(283.571411,104.366905)(283.571411,104.366905)
              \qbezier(286.570618,104.453239)(283.657715,101.366699)(283.657715,101.366699)
              \qbezier(276.607391,101.029442)(273.416473,103.821869)(273.416473,103.821869)
              \qbezier(276.407654,104.021622)(273.612213,100.829697)(273.612213,100.829697)
              \qbezier(266.585938,100.237778)(263.296661,102.919510)(263.296661,102.919510)
              \qbezier(266.279816,103.221634)(263.601807,99.935860)(263.601807,99.935860)
              \qbezier(256.599640,99.104645)(253.219009,101.672844)(253.219009,101.672844)
              \qbezier(256.191101,102.076355)(253.623520,98.696831)(253.623520,98.696831)
              \qbezier(191.355911,100.167221)(188.756180,103.519745)(188.756180,103.519745)
              \qbezier(191.732315,103.143341)(188.379776,100.543625)(188.379776,100.543625)
              \qbezier(181.298325,101.154907)(178.463730,104.311699)(178.463730,104.311699)
              \qbezier(181.458923,104.151108)(178.303131,101.316521)(178.303131,101.316521)
              \qbezier(171.225693,101.504219)(168.233505,104.511452)(168.233505,104.511452)
              \qbezier(171.233719,104.503418)(168.226501,101.512245)(168.226501,101.512245)
              \qbezier(161.165115,101.374733)(158.047455,104.252480)(158.047455,104.252480)
              \qbezier(161.046677,104.373940)(158.164902,101.257294)(158.164902,101.257294)
              \qbezier(151.123596,100.849770)(147.904572,103.610085)(147.904572,103.610085)
              \qbezier(150.895737,103.836929)(148.131409,100.617905)(148.131409,100.617905)
              \qbezier(141.109177,99.963257)(137.795807,102.613365)(137.795807,102.613365)
              \qbezier(140.779953,102.943588)(138.129044,99.630013)(138.129044,99.630013)
              \qbezier(76.132126,98.912529)(74.328484,102.750870)(74.328484,102.750870)
              \qbezier(77.151535,101.735085)(73.309082,99.931946)(73.309082,99.931946)
              \qbezier(65.933823,101.065582)(63.044132,104.170174)(63.044132,104.170174)
              \qbezier(66.043633,104.060760)(62.934319,101.174980)(62.934319,101.174980)
              \put(280.349365,111.758530){$\s_{n+1}$}
              \put(163.043121,110.924408){$\s_{n}$}
              \qbezier(52.265564,93.517380)(55.159172,91.815720)(55.159172,91.815720)
              \qbezier(55.159172,91.815720)(52.061699,90.521790)(52.061699,90.521790)
              \qbezier(52.061699,90.521790)(52.265564,93.517380)(52.265564,93.517380)
              \qbezier(62.491264,93.780060)(65.616241,92.556686)(65.616241,92.556686)
              \qbezier(65.616241,92.556686)(62.761776,90.792297)(62.761776,90.792297)
              \qbezier(62.761776,90.792297)(62.491264,93.780060)(62.491264,93.780060)
              \qbezier(72.524849,95.128799)(75.736046,94.168205)(75.736046,94.168205)
              \qbezier(75.736046,94.168205)(73.042381,92.172455)(73.042381,92.172455)
              \qbezier(73.042381,92.172455)(72.524849,95.128799)(72.524849,95.128799)
              \qbezier(137.647247,94.607346)(140.383453,92.666496)(140.383453,92.666496)
              \qbezier(140.383453,92.666496)(137.192535,91.639259)(137.192535,91.639259)
              \qbezier(137.192535,91.639259)(137.647247,94.607346)(137.647247,94.607346)
              \qbezier(147.707840,93.513466)(150.593613,91.803871)(150.593613,91.803871)
              \qbezier(150.593613,91.803871)(147.496048,90.517876)(147.496048,90.517876)
              \qbezier(147.496048,90.517876)(147.707840,93.513466)(147.707840,93.513466)
              \qbezier(157.784485,93.093895)(160.760605,91.549026)(160.760605,91.549026)
              \qbezier(160.760605,91.549026)(157.737289,90.094383)(157.737289,90.094383)
              \qbezier(157.737289,90.094383)(157.784485,93.093895)(157.784485,93.093895)
              \qbezier(167.850082,93.168373)(170.888428,91.752991)(170.888428,91.752991)
              \qbezier(170.888428,91.752991)(167.932388,90.168869)(167.932388,90.168869)
              \qbezier(167.932388,90.168869)(167.850082,93.168373)(167.850082,93.168373)
              \qbezier(177.894608,93.646759)(180.985153,92.348907)(180.985153,92.348907)
              \qbezier(180.985153,92.348907)(178.091339,90.655083)(178.091339,90.655083)
              \qbezier(178.091339,90.655083)(177.894608,93.646759)(177.894608,93.646759)
              \qbezier(187.913040,94.489716)(191.045746,93.301674)(191.045746,93.301674)
              \qbezier(191.045746,93.301674)(188.215164,91.505966)(188.215164,91.505966)
              \qbezier(188.215164,91.505966)(187.913040,94.489716)(187.913040,94.489716)
              \qbezier(197.891312,95.677742)(201.067169,94.595612)(201.067169,94.595612)
              \qbezier(201.067169,94.595612)(198.295822,92.701828)(198.295822,92.701828)
              \qbezier(198.295822,92.701828)(197.891312,95.677742)(197.891312,95.677742)
              \qbezier(263.135071,94.223114)(265.927490,92.356743)(265.927490,92.356743)
              \qbezier(265.927490,92.356743)(262.758667,91.247200)(262.758667,91.247200)
              \qbezier(262.758667,91.247200)(263.135071,94.223114)(263.135071,94.223114)
              \qbezier(273.204681,93.340919)(276.121582,91.682320)(276.121582,91.682320)
              \qbezier(276.121582,91.682320)(273.043091,90.345322)(273.043091,90.345322)
              \qbezier(273.043091,90.345322)(273.204681,93.340919)(273.204681,93.340919)
              \qbezier(283.277313,93.070412)(286.272491,91.560867)(286.272491,91.560867)
              \qbezier(286.272491,91.560867)(283.273285,90.070900)(283.273285,90.070900)
              \qbezier(283.273285,90.070900)(283.277313,93.070412)(283.277313,93.070412)
              \qbezier(293.337891,93.262527)(296.396332,91.878456)(296.396332,91.878456)
              \qbezier(296.396332,91.878456)(293.455353,90.263016)(293.455353,90.263016)
              \qbezier(293.455353,90.263016)(293.337891,93.262527)(293.337891,93.262527)
              \qbezier(303.375397,93.842796)(306.480988,92.576363)(306.480988,92.576363)
              \qbezier(306.480988,92.576363)(303.603241,90.851120)(303.603241,90.851120)
              \qbezier(303.603241,90.851120)(303.375397,93.842796)(303.375397,93.842796)
              \qbezier(313.385803,94.779892)(316.530548,93.619255)(316.530548,93.619255)
              \qbezier(316.530548,93.619255)(313.715027,91.799965)(313.715027,91.799965)
              \qbezier(313.715027,91.799965)(313.385803,94.779892)(313.385803,94.779892)
              \qbezier(378.284241,95.156296)(380.597900,92.725319)(380.597900,92.725319)
              \qbezier(380.597900,92.725319)(377.264465,92.333252)(377.264465,92.333252)
              \qbezier(377.264465,92.333252)(378.284241,95.156296)(378.284241,95.156296)
              \qbezier(388.568695,93.458557)(391.513672,91.847038)(391.513672,91.847038)
              \qbezier(391.513672,91.847038)(388.459259,90.459053)(388.459259,90.459053)
              \qbezier(388.459259,90.459053)(388.568695,93.458557)(388.568695,93.458557)
            \end{picture}
            \label{fig1}
            \caption{Structure of the Riemann surface $\L$.
            Crosses and triangles show how rims on $\L_1$ and $\L_2$ are joined together.}
        \end{center}
    \end{figure}

    We also introduce the following notation: if $\l \in \L_1$, then $\l_* \in \L_2$
    is projection of $\l$ on the other sheet. Conversely, if $\l \in \L_2$, then
    we introduce $\l_* \in \L_1$. We denote a clockwise oriented arc
    on any open circle gap $\g_n^c$ between $\a, \b \in \g_n^c$ by $\lan \a,\b \ran$ (see Fig.~2
    and definition in Section~\ref{p2}).

    \begin{figure}[h!]
        \begin{center}
            \begin{picture}(220.103,86.6066)(0,-5.79967)
              \put(166.634537,65.673355){$\a$}
              \put(170.409637,61.456200){\circle*{3}}
              \put(100.997322,6.249347){$\g_n^{(2)}$}
              \qbezier(170.409637,61.456200)(110.052307,74.789108)(49.693455,61.456200)
              \qbezier(49.693455,61.456200)(-0.309708,45.360840)(49.693455,29.265635)
              \linethickness{0.75mm}
              \qbezier(49.693455,29.265635)(110.052307,15.932574)(170.409637,29.265635)
              \qbezier(170.409637,29.265635)(220.412949,45.360840)(170.409637,61.456200)
              \put(46.045025,16.716452){$\b$}
              \put(49.693455,29.265635){\circle*{3}}
              \put(100.997322,74.672974){$\g_n^{(1)}$}
              \put(201.525894,42.250847){$\a_n^+$}
              \put(195.410034,45.360867){\circle*{3}}
              \put(3.916632,42.250847){$\a_n^-$}
              \put(24.693855,45.360867){\circle*{3}}
            \end{picture}
            \label{fig2}
            \caption{The circle gap $\g_n^{c}$ for some $n \in \Z$. The thick line denotes the
            clockwise oriented arc $\lan \a,\b \ran$ and the thin line denotes the clockwise
            oriented arc $\lan \b,\a \ran$.}
        \end{center}
    \end{figure}
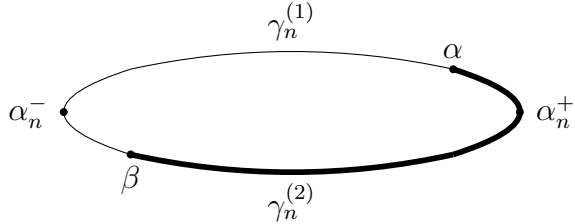

    Let $\eta \in C_o^\iy(\R,\C^2)$. We introduce $f_{\eta}(\l,t)=((H_t-\l)^{-1}\eta,\eta)$,
    $(\l,t) \in \C_+ \ts \R$. Below we show that each function $f_{\eta}(\cdot,t)$,
    $t \in \R$, is analytic on $\C_+$ and admits a meromorphic continuation
    from $\C_+ \ss \L_1$ onto the Riemann
    surface $\L$ with at most two simple poles on each open circle gap $\g_n^c$. Moreover,
    their positions depend on $t$ and there are no other poles.
    \begin{definition*}
        Let $H_t$ be a Dirac operator with dislocation $t \in \R$. Then
        \begin{enumerate}[i), leftmargin=*]
            \item if $f_{\eta}(\cdot,t)$ has a pole $\l \in \L_1$ for some
            $\eta \in C_o^\iy(\R, \C^2)$, then $\l$ is an eigenvalue of $H_t$;
            \item if $f_{\eta}(\cdot,t)$ has a pole $\l \in \L_2$ for some
            $\eta \in C_o^\iy(\R,\C^2)$, then $\l$ is called a resonance of $H_t$;
            \item if $\l = \a_n^+$ or $\l = \a_n^-$ for some $n \in \Z$ and the function
            $z \mapsto f_{\eta}(\a_n^{\pm} \mp z^2)$ has a pole at $0$ for some
            $\eta \in C_o^{\iy}(\R,\C^2)$, then $\l$ is called a virtual state of $H_t$.
        \end{enumerate}
        If $\l$ is an eigenvalue, a resonance or a virtual state, then $\l$ is called a state of
        $H_t$.
    \end{definition*}

    \subsection{Schr{\"o}dinger operator}
    In our paper about Dirac operators with dislocation potentials we
    use the methods from the paper \cite{Kor01a}, where
    the same problem for the Schr{\"o}dinger operator on the line was
    considered. The main difference between them is the following:

\no    i)  The periodic potential for Dirac operators consists of
two         functions $q_1$, $q_2$.

\no    ii)  Roughly speaking, the spectral problem for Dirac
operators
        corresponds to Schr{\"o}dinger operator with distributions.

\no  iii) The Dubrovin equation for Dirac operators is
        more complicated than for Schr{\"o}dinger operator (see e.g. \cite{MokKor}). Maybe it is the
        main point.

    In order to explain these differences we describe the results from
    \cite{Kor01a} about the dislocation problem for the Schr{\"o}dinger
    operator on the line. We introduce a operator $h_t$ with a
    dislocation potential $p_t$ on $L^2(\R)$ given by
    $$
        h_t f = - f'' + p_t f, \qqq p_t = p\c_- + p (\cdot + t)\c_+,
    $$
    where the potential $p \in L^2(\T)$. The operator $h_t$ has only
    absolutely continuous spectrum (the union of bands $s_n$ separated
    by gaps $g_n$) plus at most two eigenvalues in each open gap. The
    bands $s_n$ and the gaps $g_n$ are given by
    $$
        s_n = [a_{n-1}^+,a_n^-],\qq g_n = (a_{n}^-,a_n^+),\qq
        {\rm satisfying}\qq a_{n-1}^+ < a_n^- \leq a_n^+\qq \forall n \ge 1.
    $$
    The sequence $a_0^+$, $a_n^\pm$, $n\ge 1$, is the spectrum of the equation $-y'' + py = \l y$
    with 2-periodic boundary conditions, i.e. $y(x+2)=y(x)$, $x\in \R$.
    If some gap degenerates, i.e. $g_n=\es$, then the corresponding bands $s_{n}$
    and $s_{n+1}$ touch. This happens when the number $a_n^- = a_n^+$ is a double eigenvalue
    of the 2-periodic problem.  Generally, the eigenfunctions corresponding to eigenvalues
    $a_{2n}^{\pm}$ are 1-periodic, those for $a_{2n+1}^{\pm}$ are anti-periodic, i.e.
    $y(x+1)=-y(x)$, $x\in\R$.

    We shortly recall the main properties of the operator
    $h^+_t y= -y''+p(\cdot+t) y$, $y(0)=0$ on $L^2(\R_+)$, see e.g.
    \cite{Dub, Kor99, trub77, Zh}. The operator $h^+_t$ has only the
    absolutely continuous spectrum $\s_{ac}(h^+_t)=\s_{ac}(h_t)$ and
    at most one eigenvalue in each open gap $g_n$, $n\ge 1$.

    The operators $h^+_t$, $h_t$ have the same two-sheeted Riemann
    surface $\cL$, which
    is obtained by joining the upper and lower rims of two copies of the cut plane $\C \sm \s(h_0)$
    in the usual (crosswise) way. We denote the $n$-th gap on the first, physical, sheet $\cL_1$
    by $g_n^{(1)}$ and its counterpart on the second, nonphysical, sheet $\cL_2$ by $g_n^{(2)}$,
    and we set a circle gap $g_n^c=\ol g_n^{(1)}\cup \ol g_n^{(2)}$. As above, if $\l \in \cL$, then
    $\l_* \in \cL$ is projection of $\l$ on the other sheet.

    For the operators $h^+_t$, $h_t$ we define eigenvalues, resonances and
    virtual states, and arcs as we have defined them for the
    Dirac operator. The positions of these states depend on $t$. In
    each open circle gap $g_n^c$, $n\ge 1$, the operator $h^+_t$ has the
    unique state $\r_n(t)$, and its projection on the complex plane
    coincides with the $n$-th Dirichlet eigenvalue on the unit interval,
    i.e. $-y'' + q(x+t)y=\r_n(t) y$, $y(0)=y(1)=0$.  The function
    $\r_n(\cdot)$ belongs to $C^1(\T)$, and $\r_n(t)$ changes sheets when it
    hits $a_n^+$ or $a_n^-$ and makes $n$ complete revolutions around
    the circle gap $g_n^c$ when $t$ runs through $[0,1]$. It is
    important that the point $\r_n(t)$ changes the direction of motion
    when it hits the gap end.
    In \cite{Kor01a} it is proved that each operator $h_t$, $t \in \R$, has exactly
    two distinct states $\x_n^{\pm}(t)$ in any open circle gap $g_n^c$, $n \ge 1$,  and they satisfy:

        i) The mappings $\R \ni t \mapsto \x_n^{\pm}(t) \in \ol \g_n^c$ are continuous,
        $2$-periodic, and satisfy\\
        $
            \x_n^{\pm}(0) = a_n^{\pm},\qq \x_{2n}^{\pm}(t+1) = \x_{2n}^{\pm}(t),
            \qq \x_{2n+1}^{\pm}(t+1) = \x_{2n+1}^{\mp}(t),\qq t \in \R;
        $

        ii) $\x_n^{\pm}(t)$ changes sheets when it hits $a_n^+$ or $a_n^-$ and makes ${n\/2}$
        complete revolutions around the circle gap $g_n^c$ when $t$ runs through $[0,1]$;

        iii) $\r_n(t)$ is a state of $h_t$ if and only if $\r_n(t) = \r_n(0)_*$;

        iv) $\x_n^{\pm}(t) \in \lan \r_n(t), \r_n(0)_* \ran$ if and only if
        $\x_n^{\mp}(t) \in \lan \r_n(0)_*,\r_n(t) \ran$.

    It is important that states $\r_n(t)$ and $\r_n(0)_*$ control the dynamics of $\x^{\pm}_n(t)$.

    \subsection{Literature survey}

    The Schr\"odinger operator with dislocation was studied in \cite{Kor01a}, where
    the different properties were obtained.
    In other cases, the various  dislocation problems were discussed in the papers
    \cite{DPR09, D19, Fef, Kor01a, Kor05, HK11a, Hemp, HKS15} (see a good review in \cite{D19}).
    Below we show that, the dynamics of
    the states of $H_t$ is controlled by states of Dirac operator on the half-line, as in case of
    Schr\"odinger operator. Thus, we need results for the Dirac operator on the half-line
    with the Dirichlet boundary condition at zero. This operator was studied in \cite{Kor01, MokKor}
    (see also reference therein). Note that the states of the operator $H_t$ are less smooth as
    function of $t$ than the states $\x^{\pm}(t)$ of the operator $h_t$. Thus we need
    a specific version of the implicit function theorem from \cite{MokKor} to eliminate
    arising problems and to use arguments from \cite{Kor01a}.

    The Dirac and the Schr{\"o}dinger operators with dislocation of a periodic potential are
    examples of the operators such that the potential has different asymptotics at positive
    and negative infinity. In general, such operators have the form $A = A_0 + W$, where $A_0$
    is a free operator acting on $L^2(\R)$ and the potential $W$ coincides in some sense with
    the model potentials $W^{\pm}$ at $\pm \iy$. Thus, spectral properties of the operator $A$
    are related to spectral properties of the operators $A^{\pm} = A_0 + W^{\pm}$ acting on
    $L^2(\R_{\pm})$. For a more detailed description of such operators, see \cite{DavSim78, Kor05}.
    In case of dislocation $W^-$ is a periodic potential, and $W^+(x) = W^-(x+t)$. In general,
    a periodic potential with a dislocation has a step discontinuity at zero. But adding a small
    local perturbation to a dislocation potential, we obtain a continuous
    potential with a different asymptotics. We also note that we can rewrite the potential $V_t$ as
    follows
    $$
        V_t(x) = V(x) + \chi_+(x) (V(x+t) - V(x)),\qq x \in \R,
    $$
    where $\chi_+(x) (V(x+t) - V(x))$ is a periodic potential (a perturbation) on the positive half-line.
    Moreover, this perturbation is not small for some $t \in \R$. These facts imply that we
    cannot use the standard perturbation theory.

    Operators with different asymptotics at positive and negative infinity have various physical
    and mathematical applications. Many of them arise from solid state physics, where such a
    model describes contact of two crystals, contact of crystal and vacuum, and also describes
    some irregularity within a crystal structure. See the papers \cite{D19, Fef, Hemp} about
    applications of these models to solid state physics. Applications also arise from the
    connection of operators with nonlinear equations (see e.g. \cite{Egor09, MonKot}).

    Scattering theory for Schr\"odinger operator on the line with different spatial asymptotics
    was considered in the papers \cite{DavSim78, Ges97}. The scattering theory for Dirac operator
    with different spatial asymptotics was investigated in the article \cite{RuiBon} in connection
    with the Klein paradox. Multi-dimensional Schr{\"o}dinger operators periodic in all but one
    dimension was considered in the papers \cite{DavSim78, Hemp}.

    Our paper is organized as follows:\\
    In Section \ref{p1} we present our main results.\\
    In Section \ref{p2} we describe notations and facts about the periodic Dirac operator.\\
    In Section \ref{p8} we study properties of the Weyl-Titchmarsh function.\\
    In Section \ref{p4} we consider the Dirac operator with dislocation.\\
    In Section \ref{p6} we prove the main theorems.\\
    In Section \ref{p5} (Appendix) we give technical lemmas and the specific implicit function
    theorem.

\section{Main results} \label{p1}
    In order to formulate main theorem we introduce a Dirac operator $H^{+}_t$, $t \in \R$,
    acting on $L^2(\R_{+},\C^2)$ and given by
    $$
        H^{+}_t f(x) = J f'(x) + V(x+t) f(x),\ \ x \in \R_{+},\ \
        f(x) = \ma f_1(x) \\ f_2(x) \am,\ \ f_1(0) = 0.
    $$
    Its spectrum consists of an absolutely continuous part
    $\s_{ac}(H^{+}_t) = \s(H_0) = \cup_{n \in \Z} \s_n$ plus at most one eigenvalue in
    each non-empty gap $\g_n$ (see \cite{W87, MokKor}). For this operator we introduce
    states (eigenvalues, resonances and  virtual states) as we have defined them above for
    $H_t$. In \cite{MokKor} it is proved that the operator $H^+_t$ has exactly one state $\m_n(t)$
    in each open circle gap $\g_n^c$ for any $t \in \R$ and there
    are no other states. Its projection on the complex plane coincides with
    the Dirichlet eigenvalue on the unit interval, i.e. $Jy' + V(\cdot+t)y=\m_n(t) y$,
    $y_1(0) =y_1(1)=0$, where $y = \ma y_1 & y_2 \am^{\top}$. Moreover,
    $\m_n(\cdot) \in \cH^1(\T, \g_n^c)$, where $\cH^{1}(\T,\g_n^c)$ is the Sobolev space of
    1-periodic functions $u: \R \to \g_n^c$, which will be defined in Section \ref{p2}.

    There exists a simple correspondence between resonances and eigenvalues of $H_t$.
    \begin{proposition} \label{p4l4}
        Let $H_t$ be the Dirac operator with the potential
        $$
            V_t(x) = V(x) \c_-(x) + V(x+t) \c_+(x),\qq x \in \R
        $$
        for some $(t,V) \in \R \ts \cP$, and let $\widetilde{H}_t$ be the Dirac operator with the potential
        $$
            \widetilde{V}_t(x) = V(x+t) \c_-(x) + V(x) \c_+(x),\qq x \in \R.
        $$
        Then $\l \in \L$ is an eigenvalue (a resonance) of $H_t$ if and only if $\l_ {*}$ is
        a resonance (an eigenvalue) of $\widetilde{H}_t$, where $\l_{*}$ is the projection of $\l$
        on the other sheet of $\L$.
    \end{proposition}
    \begin{remark}
        It is easy to see that $\widetilde{V}_t$ is also a dislocation potential, i.e.
        $\widetilde{V}_t(x) = Q_{\t}(x)$, $x \in \R$, where $Q(x) = V(x+t)$, $x \in \R$, and $\t = -t$.
        Thus we can study the motion of resonances of a Dirac operator with dislocation as
        the motion of eigenvalues of another Dirac operator with dislocation.
    \end{remark}

    We present our main theorem, about existence and smoothness of the states.

    \begin{theorem} \label{p0t1}
        Let $(t, V) \in \R \ts \cP$ and let a gap $\g_n$ be open for some $n \in \Z$. Then the
        operator $H_t$ has exactly two distinct states $\l_n^{\pm}(t) \in \g_n^c$ such that
        $\l_n^{\pm}(\cdot) \in \cH^1(2\T, \g_n^c)$, $\l_n^{\pm}(0) = \a_n^{\pm}$, and they satisfy:
        \begin{enumerate}[i)]
            \item $\l_n^{\pm}(t) \in \lan \m_n(t),\m_n(0)_* \ran$ if and only if
            $\l_n^{\mp}(t) \in \lan \m_n(0)_*,\m_n(t) \ran$, where
            $\lan \cdot,\cdot \ran$ is the arc on $\g_n^c$ defined above.
            \item $\l_n^{+}(t) = \m_n(t)$ or $\l_n^{-}(t) = \m_n(t)$ if and only if
            $\m_n(t) = \m_n(0)_*$.
            \item if $\m_n(t)$ makes $r(n)$ revolutions around $\g_n^c$ when $t$ runs through
            $[0,1]$, then $\l_n^{+}(t)$ and $\l_n^{-}(t)$ make $r(n)/2$ revolutions around $\g_n^c$.
            \item $\l_n^{\pm}(t+1) = \l_n^{\pm}(t)
            \text{ if $r(n)$ is even};\ \ \l_n^{\pm}(t+1) = \l_n^{\mp}(t)\text{ if $r(n)$ is odd}$.
        \end{enumerate}
    \end{theorem}

    \begin{remark}
        Results of Theorem \ref{p0t1} are similar to the case of the Schr\"odinger operator from
        \cite{Kor01a}. The difference is that there is no control of $r(n)$,
        since the motion of the Dirichlet eigenvalue $\m_n(t)$ is non-monotone and more
        complicated for the Dirac operator than for the Schr\"odinger operator (see \cite{MokKor}).
        However, Theorem \ref{p0t1} associates the number of revolutions of the $\m_n(t)$ and
        $\l^{\pm}_n(t)$ as well as for the Schr\"odinger operator.
    \end{remark}

    We describe the number of the revolutions for specific case.
    \begin{theorem} \label{p0t2}
        Let $V \in \cP$ and let a gap $\g_n$ be open for some $n \in \Z$. Suppose that
        \[ \label{p1e16}
            \sign(q_1(t) + \a_n^-) = \sign(q_1(t) + \a_n^+) = \const \neq 0
            \text{ for almost all $t \in [0,1]$}.
        \]
        Then each state of $H_t$ in $\g_n^c$ makes $\left| n \right|/2$ complete revolutions
        when $t$ runs through $[0,1]$.
    \end{theorem}
    \begin{remark}
        1) As an example, if $q_1 = 0$, then condition (\ref{p1e16}) is satisfied for any
        $n \in \Z$, in this case the Dirac operator is a supersymmetry charge
        (see Chapter 5 in \cite{Thall}).

        2) If $q_1 \in L^{\iy}(\T)$, then condition (\ref{p1e16})
        holds true for large enough $n \in \Z$.
    \end{remark}

    By Theorem \ref{p0t1}, $H_t$ at $t = 0$ has exactly two virtual states
    $\l_n^{\pm}(0) = \a_n^{\pm}$ in each open circle gap $\g_n^c$, $n \in \Z$.
    In order to describe the motion of states, we determine their local asymptotics as $t \to 0$.
    We introduce the following functions
    \[
        \begin{aligned}
            g_n^{\pm}(z,t) &= z^2(q_1(t) + \a_n^{\pm}) - 2zq_2(t) - (q_1(t)-\a_n^{\pm}),\quad
            Q_n^{\pm}(t) = \int_{0}^{t} \left| q_1(\t) + \a_n^{\pm} \right| d\t,\\
            G_n^{\pm}(t) &= \int_{0}^t \left| g_n^{\pm}(m_+(\a_n^{\pm},0),\t) \right| d\t,\quad
            W_n^{\pm}(t) = \left| t \right|^{1/2}
            \left| \int_{0}^t \left| V(\t) + \a_n^{\pm} \1_2 \right|^2 d\t \right|^{1/2},\\
            \kappa_n^{\pm}(t) &= \frac{\mp \left| 2 M_n^{\pm} \right|^{1/2}}
            {2\left\|\vp(\cdot,\a_n^{\pm},t)\right\|^2},\quad
            \varkappa_n^{\pm} = \frac{(-1)^n\vp_1(1,\a_n^{\pm},0)}
            {2 \left| 2 M_n^{\pm} \right|^{1/2}},
        \end{aligned}
    \]
    where $(n,t,z) \in \Z \ts \R \ts \C$, and $\1_2$ is the identity matrix $2 \times 2$. We
    define the norms and introduce the fundamental
    solution $\vp(x,\l,t)$, the Weyl-Titchmarsh function $m_+(\l,t)$, and the effective
    masses $\pm M_n^{\pm} > 0$ in Section \ref{p2}.
    % $\left| \cdot \right|$ is defined by $\left| A \right|^2 = \Tr A^* A$, norm of any vector-valued function is defined by $\left\| \vp \right\|^2 = \int_0^1 (\left|\vp_1(x)\right|^2 + \left|\vp_2(x)\right|^2)dx$.
    \begin{theorem} \label{p0t5}
        Let $V \in \cP$ and let a gap $\g_n$ be open for some $n \in \Z$. Then there exist
        $z_n^{\pm} \in \cH^1([-\ve,\ve])$  for some $\ve > 0$ such that $z_n^{\pm}(0) = 0$ and:
        \begin{enumerate}[i)]
            \item $\l_n^{\pm}(t) = \a_n^{\pm} \mp (z_n^{\pm}(t))^2$ are states of $H_t$ in $\g_n^c$
            for any $t \in [-\ve,\ve]$;
            \item $\l_n^{\pm}(t) \in \L_j$ if and only if $(-1)^j z_n^{\pm}(t) < 0$ for any
            $t \in [-\ve,\ve]$ and $j = 1,2$;
        \end{enumerate}
        and for $t \to 0$ the following asymptotics hold true:
        \begin{align}
            z_n^{\pm}(t) &= \kappa_n^{\pm}(0) \int_{0}^t (q_1(\t) + \a_n^{\pm}) d\t +
            O(Q_n^{\pm}(t)W_n^{\pm}(t)), && \a_n^{\pm} = \m_n(0), \label{p1e7} \\
            z_n^{\pm}(t) &= \varkappa_n^{\pm} \int_{0}^t g_n^{\pm}(m_+(\a_n^{\pm},0),\t) d\t +
            O(G_n^{\pm}(t)W_n^{\pm}(t)), && \a_n^{\pm} \neq \m_n(0). \label{p1e8}
        \end{align}
    \end{theorem}
    \begin{remark}
        1) In Theorem \ref{p0t5} the states $\l_n^{\pm}(t) \in \g_n^c$, and $z_n^{\pm}(t)$ are local
        coordinates on the Riemann surface $\L$ in neighborhoods of $\a_n^{\pm}$.

        2) Theorem \ref{p6t1} describes similar asymptotics as $t \to t_o$ for any $t_o \in \R$.

        3) If a gap $\g_n$ is open for some $n \in \Z$, then $\mp \kappa_n^{\pm}(0) > 0$,
        and $\mp \varkappa_n^{\pm} > 0$ (see Lemma \ref{p4l11}).
    \end{remark}
    It follows from (\ref{p1e7}-4) that if the integrands do not change their sign, then the
    states of $H_t$ are localy monotone functions of $t$. Using this fact, we prove that
    high-energy states are globally monotone if the potential is bounded. We introduce the following
    notation

    \no $\| V \|_{\iy} = \esssup_{x \in \T} \Tr V^2(x)$, $V \in \cP$.

    \begin{theorem} \label{p0t6}
        Let $V \in \cP$ and let a gap $\g_n$ be open for some $n \in \Z$.
        Suppose that $\| V \|_{\iy} < 2\left| \a_n^{\pm} \right|$.
        Then $\l_n^{+}(t)$ and $\l_n^{-}(t)$ run strictly monotonically and in one direction around
        $\g_n^c$, changing sheets when they hit $\a_n^+$ or $\a_n^-$ and making
        $\left| n \right|/2$ complete revolutions when $t$ runs through $[0,1]$. Moreover,
        $\l_n^{\pm}(\cdot)$ run clockwise (counterclockwise) around $\g_n^c$ if
        $\a_n^{\pm} > 0$ ($\a_n^{\pm} < 0$).
    \end{theorem}
    \begin{remark}
        Under the conditions of Theorem \ref{p0t6} $\m_n(t)$ also runs strictly monotonically
        around $\g_n^c$, changing sheets when it hits $\a_n^+$ or $\a_n^-$ and making
        $\left| n \right|$ complete revolutions when $t$ runs through $[0,1]$
        (see Theorem 2.3 in \cite{MokKor}).
    \end{remark}

    Now we describe low-energy states of $H_t$ for specific potentials with
    small dislocation. We introduce the Neumann eigenvalues $\nu_n$, $n \in \Z$,
    as eigenvalues of the equation $Jy'+Vy = \l y$, $y_2(0) = y_2(1) = 0$, where
    $y = \ma y_1 & y_2 \am^{\top}$. Note that in each interval $[\a_n^-,\a_n^+]$ there exists a unique
    Neumann eigenvalue and there are no other (see e.g. \cite{LevSar}). We also introduce the
    subspace of "even" potentials $\cP_{e}$ by
    $$
        \cP_{e} = \rt\{ \, \left. V= \ma q_1 & q_2 \\ q_2 & -q_1 \am \in \cP \, \right| \,
        q_1(x) = q_1(1-x),\, q_2(x) = -q_2(1-x),\, x \in \T \, \rt\}.
    $$
    If $V \in \cP_e$, then $\mu_n = \a_n^{j}$, and $\nu_n = \a_n^{-j}$ for some $j = \pm$
    and for any $n \in \Z$. Note also that if $q_1 = 0$, then the spectrum of $H_t$ is symmetric
    with respect to zero, i.e.
    $\a_n^{\pm} = -\a_{-n}^{\mp}$, $\m_n(t) = -\m_{-n}(t)$, $\nu_n(t) = -\nu_{-n}(t)$
    for any $(n,t) \in \Z \ts \R$.

    % If $V \in \cP_{e}$, then the Neumann and Dirichlet eigenvalues lie on the gaps boundary
    % (see e.g. \cite{GreGui}). See also \cite{Kor01} about an inverse problem for the periodic Dirac
    % operator with "even" potentials. Remark that using Lemma \ref{p4l4}, we
    % can swap the resonances and the eigenvalues of the Dirac operator with dislocation.
    \begin{theorem} \label{p0t4}
        Let $V \in \cP_{e}$ such that $q_1 = 0$, and let $N \in \N$. Then there exists $\ve > 0$
        such that each $H_t$, $t \in (0, \ve)$, has exactly one eigenvalue and one resonance
        in each open circle gap $\g_n^c$, $1 \leq \left| n \right| \leq N$. Moreover, in this case
        $\g_0^c = \es$.
    \end{theorem}
    % \begin{remark}
        % 1) If $q_1 = 0$, then the spectrum of $H_t$ is symmetric with respect to zero, i.e.
        % $\a_n^{\pm} = -\a_{-n}^{\mp}$, $\m_n(t) = -\m_{-n}(t)$, $\nu_n(t) = -\nu_{-n}(t)$
        % for any $(n,t) \in \Z \ts \R$.

        % 2) If $V \in \cP_e$, then $\mu_n = \a_n^{j}$, and $\nu_n = \a_n^{-j}$ for some $j = \pm$
        % and for any $n \in \Z$.
    % \end{remark}
    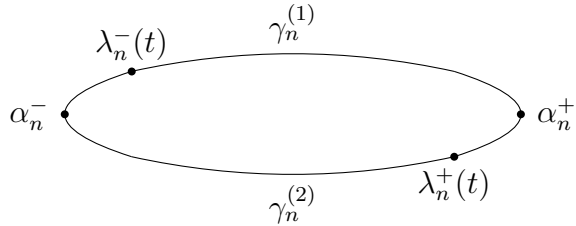
\begin{figure}[h!]
        \begin{center}
            \begin{picture}(220.103,86.4008)(0,-5.79967)
              \qbezier(170.409637,29.163654)(220.412949,45.258934)(170.409637,61.354218)
              \qbezier(170.409637,61.354218)(110.052307,74.687286)(49.693455,61.354218)
              \put(36.020573,68.570175){$\l_n^-(t)$}
              \put(49.693455,61.354218){\circle*{3}}
              \put(101,74.468216){$\g_n^{(1)}$}
              \qbezier(49.693455,61.354218)(-0.309708,45.258934)(49.693455,29.163654)
              \qbezier(49.693455,29.163654)(110.052307,15.830591)(170.409637,29.163654)
              \put(156.735550,15.758874){$\l_n^{+}(t)$}
              \put(170.409637,29.163654){\circle*{3}}
              \put(101,6.249347){$\g_n^{(2)}$}
              \put(201.525894,42.148464){$\a_n^+$}
              \put(195.410034,45.258888){\circle*{3}}
              \put(3.916632,42.148464){$\a_n^-$}
              \put(24.693855,45.258888){\circle*{3}}
            \end{picture}
            \label{fig3}
            \caption{The position of the states $\l_n^{\pm}(t)$ in $\g_n^{c}$ for some $n > 0$,
            and $t \in (0,\ve)$ under the conditions of Theorem \ref{p0t4}.}
        \end{center}
    \end{figure}
    Using Theorem \ref{p0t4}, we can construct the
    operator $H_t$ with exactly one eigenvalue in a finite number of gaps.
    \begin{corollary} \label{p0t8}
        For any $N \geq 1$ there exist a potential
        $V \in \cP_e$
        and $\ve > 0$ such that $q_1 = 0$, each gap $\g_n$ is open and each
        $H_t$, $t \in (0,\ve)$, has exactly one eigenvalue in $\g_n$,
        $1 \leq |n| \leq N$, and $\g_0 = \es$.
    \end{corollary}

    \begin{remark}
        1) If $q_1 = 0$, then the following identity holds true
        $$
            H_t^2 = h^{(1)}_t \os h^{(2)}_t,
        $$
        where $h^{(j)}_t = -\frac{d^2}{dx^2} + p^{(j)}_t$, $j = 1,2$, are the Schr\"odinger
        operators acting on $L^2(\R)$ and the potentials $p^{(j)}_t$ have the form
        $$
            p^{(j)}_t(x) = v_t^2(x) - (-1)^j v_t'(x),\qq
            v_t(x) = q_2(x) + (q_2(x + t) - q_2(x)) \c_+(x),\qq x \in \R.
        $$
        It is well known that the spectra of
        $h^{(1)}_t$ and $h^{(2)}_t$ coincide away from zero (see \cite{Deift78}).
        Thus, using this identity, one can obtain similar result for the Schr\"odinger operator with
        dislocation of a singular periodic potential.

        2) Suppose in addition that $q_2 \in \cH^1(\T)$. We introduce the
        Schr\"odinger operators $\tilde{h}^{(j)}_t = -\frac{d^2}{dx^2} + \tilde{p}^{(j)}_t$,
        $j = 1,2$, acting on $L^2(\R)$, where the potentials $\tilde{p}^{(j)}_t$ have the form
        $$
            \tilde{p}^{(j)}_t(x) = u(x) + (u(x + t) - u(x)) \c_+(x),\qq
            u(x) = q_2^2(x) - (-1)^j q_2'(x),\qq x \in \R.
        $$
        Note that $u \in L^2(\T)$ and $\tilde{h}^{(j)}_t$ is the operator, which was considered in
        \cite{Kor01a}. It is easy to see that in this case we have for any $t \in \R$
        $$
            \tilde{h}^{(j)}_t = h^{(j)}_t + (q_2(0) - q_2(t)) \d(\cdot),
        $$
        where $\d$ is the Dirac delta function. So that one need take into consideration this
        singular point perturbation to obtain results for $\tilde{h}^{(j)}_t$ from $H_t$
(see a good review about such operators in \cite{KM13}). Thus,
        the results of Theorem \ref{p0t4}, \ref{p0t8} cannot be obtained directly from the results
        of paper \cite{Kor01a}.
    \end{remark}
    Note that in Theorem \ref{p0t6}, \ref{p0t4} the states move in one direction.
    Now we prove that there exist the operator $H_t$ such that its states move in opposite direction,
    i.e. it has exactly two eigenvalues or two resonance in any finite number of gaps.
    \begin{theorem} \label{p0t7}
        Let $V \in \cP_{e}$ such that $q_1 = c \chi_{[0,\d] \cup [1-\d,1]}$ for some $c \in \R$ and
        $\d > 0$. Then for any $N \in \N$ there exist $c \in \R$, $\d > 0$, $\ve > 0$, and
        $i_o \in \Z$ such that each $H_t$, $t \in (0,\ve)$, has exactly two
        eigenvalues if $\m_n(0) < \nu_n(0)$ or two resonances if $\m_n(0) > \nu_n(0)$ in each
        open circle gap $\g_n^c$, $\left| n + i_o \right| \leq N$.
    \end{theorem}
    \begin{remark}
        The constants $c$ and $\d$ depend only on the norm $\|q_2\|$, and $N$. One can choose
        the constant $i_o$ as the number of the gap closest to zero (see details in
        the proof of the theorem).
    \end{remark}
        \begin{figure}[h!]
        \begin{center}
            \begin{picture}(416.984,75.225)(0,-5.79967)
              \qbezier(141.293869,49.462090)(121.948639,57.215729)(89.717979,56.676544)
              \qbezier(89.717979,56.676544)(63.561687,55.839439)(48.062359,49.462090)
              \put(34.390484,56.678749){$\l_n^-(t)$}
              \put(48.062359,49.462090){\circle*{3}}
              \put(85.623886,63.292458){$\g_n^{(1)}$}
              \qbezier(48.062359,49.462090)(22.221745,35.838989)(62.040382,25.685760)
              \qbezier(62.040382,25.685760)(107.454300,17.936134)(141.293869,29.877220)
              \qbezier(141.293869,29.877220)(161.872742,39.672665)(141.293869,49.462090)
              \put(127.618774,56.678749){$\l_n^{+}(t)$}
              \put(141.293869,49.462090){\circle*{3}}
              \put(85.623886,6.249347){$\g_n^{(2)}$}
              \put(157.700150,36.671001){$\n_n(0)$}
              \put(151.586319,39.671612){\circle*{3}}
              \put(3.916632,36.671001){$\m_n(0)$}
              \put(37.773922,39.671612){\circle*{3}}
              \qbezier(368.917267,29.877220)(394.760315,43.502346)(354.943054,53.657463)
              \qbezier(354.943054,53.657463)(309.526611,61.408218)(275.689972,49.462090)
              \qbezier(275.689972,49.462090)(255.111069,39.672665)(275.689972,29.877220)
              \put(262.012878,16.474548){$\l_n^{-}(t)$}
              \put(275.689972,29.877220){\circle*{3}}
              \put(313.246277,63.292458){$\g_n^{(1)}$}
              \qbezier(275.689972,29.877220)(295.032501,22.126665)(327.261627,22.662767)
              \qbezier(327.261627,22.662767)(353.423889,23.502806)(368.917267,29.877220)
              \put(355.241180,16.474548){$\l_n^+(t)$}
              \put(368.917267,29.877220){\circle*{3}}
              \put(313.246277,6.249347){$\g_n^{(2)}$}
              \put(385.322540,36.671001){$\m_n(0)$}
              \put(379.209717,39.671612){\circle*{3}}
              \put(232.801743,36.671001){$\n_n(0)$}
              \put(265.397522,39.671612){\circle*{3}}
            \end{picture}
            \label{fig4}
            \caption{The position of the states $\l_n^{\pm}(t)$ in $\g_n^{c}$ for some $n \in \Z$,
            and $t \in (0,\ve)$ under the conditions of Theorem \ref{p0t7}.}
        \end{center}
    \end{figure}
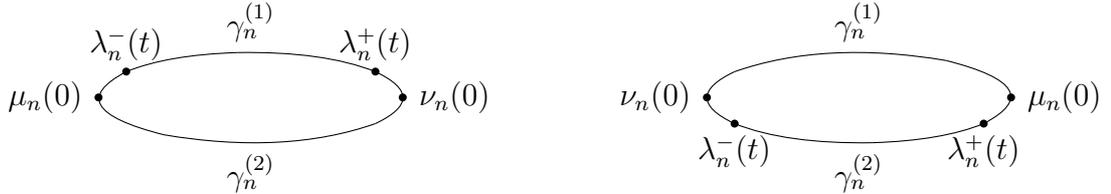
    Using the gap length mapping from \cite{Kor05b}, we can construct a potential $V \in \cP_{e}$
    such that $\m_n - \n_n$ have a given sign for any $n \in \Z$. Thus, we can construct the operator
    $H_t$, which has exactly two eigenvalues in any finite number of gaps.
    \begin{theorem} \label{p0t9}
        For any $N \geq 1$ there exist a potential $V \in \cP_{e}$ and $\ve > 0$
        such that each gap $\g_n$ is open and each $H_t$, $t \in (0,\ve)$, has exactly two
        eigenvalues in $\g_n$, $\left| n \right| \leq N$.
    \end{theorem}
    \begin{remark}
        Due to Proposition \ref{p4l4}, we can swap eigenvalues and resonances of a Dirac operator
        with dislocation. So that we can construct the Dirac operator ${\widetilde H}_t$, which
        does not have an eigenvalue and has exactly two resonances in $\g_n^c$,
        $\left| n \right| \leq N$, for any $t \in (0,\ve)$.
    \end{remark}
    Finally we consider a Dirac operator $H$ with mass $m > 0$ given by
    $$
        H y = Jy' + \ma m & q_2 \\ q_2 & m\am y,\qq
        y = \ma y_1 \\ y_2 \am,\qq q_2 \in L^2(\T),\qq q_2(x) = -q_2(1-x),\qq x \in [0,1].
    $$
    In this case $H$ describes a particle with mass $m$ and with an anomalous magnetic moment in an
    external field given by $q_2$ (see e.g. Chapter 4 in \cite{Thall}). We show that
    the spectrum of $H$ is symmetric with respect to zero and there exists a mass gap in the
    spectrum $\g_0 = (-m,m)$. We consider a dislocation problem for such operator. We prove
    that $H_t$ has exactly two virtual states $\pm m$ in the mass gap $\g_0^c$ for any $t \in \R$,
    i.e. it does not have an eigenvalue or a resonance in the mass gap.
    \begin{theorem} \label{p0t10}
        Let $V \in \cP_e$ such that $q_1 = m$ for some $m > 0$. Then each $H_t$,
        $t \in \R$, has two virtual states $\pm m$ in $\g_0^c$, and for any $t \in \R$ we have
        \begin{gather*}
            \a_n^{\pm} = -\a_{-n}^{\mp},\qq \m_n(t) = -\m_{-n}(t),\qq \n_n(t) = -\n_{-n}(t),\qq
            n \in \Z \sm \{0\};\\
            \a_0^{\pm} = \pm m,\qq \m_0(t) = -m,\qq \n_0(t) = m.
        \end{gather*}
    \end{theorem}
    \begin{remark}
        Due to the fact that $\pm m$ are virtual states of $H_t$ in the mass gap $\g_0^c$ for any $t \in \R$,
        it follows that adding sufficiently small local perturbation to $H_t$, $t \in \R$, we
        obtain two resonances, or two eigenvalues, or one eigenvalue and one resonance
        in the mass gap close to $\pm m$. But we have not such effect
        in other open gaps.
    \end{remark}

\section{Periodic Dirac operator} \label{p2}
    \subsection{Notations}

    Now we introduce definitions and notations used in our paper.
    Let $\L$ be the Riemann surface constructed as above for some Dirac operator with
    dislocation $H_t$, $t \in \R$. We introduce the following mappings: $\pr:\L \to \C$ is a projection
    mapping on the complex plane, and $\pr_j:\C \sm \s_{ac}(H_t) \to \L$, $j = 1,2$, are embeddings
    of $\C \sm \s_{ac}(H_t)$ on the first and second sheets of $\L$ respectively. Note that
    $\pr \pr_j z = z$ for any $z \in \C \sm \s_{ac}(H_t)$, and $j = 1,2$. We also have the following
    identities for any open gap $\g_n$, $n \in \Z$
    $$
        \pr \g_n^c = \ol \g_n,\qq \pr_j \g_n = \g_n^{(j)},\qq j = 1,2.
    $$
    Using these mapping, we introduce the projection of any $\l \in \L$ on the other sheet of $\L$ by
    $$
        \l_* =
        \begin{cases}
            \pr_2 \pr \l & \text{if $\l \in \L_1$, $\l \neq \a_n^{\pm}$ for any $n \in \Z$,}\\
            \pr_1 \pr \l & \text{if $\l \in \L_2$, $\l \neq \a_n^{\pm}$ for any $n \in \Z$,}\\
            \l & \text{if $\l = \a_n^{\pm}$ for some $n \in \Z$.}
        \end{cases}
    $$
    We give the definition of the clockwise oriented arc $\lan \a,\b \ran$.
    \begin{definition*}
        Let $\L$ be the Riemann surface constructed as above for some Dirac operator with
        dislocation. Let a gap $\g_n$ in the spectrum of this operator be open for some $n \in \Z$.
        Then a clockwise oriented arc $\lan \a,\b \ran$ for any $\a,\b \in \g_n^c$ is a subset
        of the circle gap $\g_n^c$ given by
        $$
            \lan \a,\b \ran =
                \begin{cases}
                    \pr_1 (\pr \a, \pr \b) & \text{if $\a,\b \in \ol \g_n^{(1)}$,
                    $\pr \a - \pr \b \leq 0$},\\
                    \pr_2 (\pr \b, \pr \a) & \text{if $\a,\b \in \ol \g_n^{(2)}$,
                    $\pr \b - \pr \a \leq 0$},\\
                    \lan \a,\a_n^+ \ran \cup \{ \a_n^+ \} \cup \lan \a_n^+,\b \ran &
                    \text{if $\a \in \{ \a_n^- \} \cup\g_n^{(1)}$,
                    $\b \in \{ \a_n^- \} \cup \g_n^{(2)}$, $\a \neq \b$},\\
                    \lan \a,\a_n^- \ran \cup \{ \a_n^- \} \cup \lan \a_n^-,\b \ran &
                    \text{if $\a \in \{ \a_n^+ \} \cup\g_n^{(2)}$,
                    $\b \in \{ \a_n^+ \} \cup \g_n^{(1)}$, $\a \neq \b$},\\
                    \lan \a,\a_* \ran \cup \{ \a_* \} \cup \lan \a_*,\b \ran &
                    \text{if $\a,\b \in \g_n^{(j)}$, $(-1)^{j}(\pr \b - \pr \a) > 0$, $j = 1,2$}.
                \end{cases}
        $$
    \end{definition*}
    We introduce the Hilbert space $M_2(\C)$ of $2 \ts 2$ matrix with complex entries equipped with
    the norm $|A|^2 = \Tr A^* A$, $A \in M_2(\C)$.
    Let $p \in [1,\iy]$, $I \ss \R$, and $B$ is a Banach space equipped
    with the norm $\| \cdot \|_{B}$. Then $L^p(I,B)$ is the standard Lebesgue space of functions
    $f: I \to B$ such that for $p \neq \iy$ the norm
    $\| f \|_{L^p(I,B)} = \left( \int_{I} \| f(x) \|^p_B dx \right)^{1/p}$ is
    finite and for $p = \iy$ the norm
    $\| f \|_{L^{\iy}(I,B)} =  \esssup_{x \in I} \| f(x) \|_B $ is finite. We also introduce
    the Banach space $L^{2}_{loc, u}(\R,B)$ of functions $f: \R \to B$ such that the norm
    $\| f \|_{loc, u} = \sup_{s \in \R} \left( \int_{s}^{s+1} \| f(x) \|^2_{B} dx \right)^{1/2}$
    is finite. We introduce the
    following abbreviations for the norms often used in our paper
    $$
    \begin{aligned}
        \| f \|^2 &= \int_0^1 \left( |f_1(x)|^2 + |f_2(x)|^2 \right) dx &&\text{ for any }
        f \in L^2([0,1],\C^2),\\
        \| f \|_{\R}^2 &= \int_{\R} \left( |f_1(x)|^2 + |f_2(x)|^2 \right) dx &&\text{ for any }
        f \in L^2(\R,\C^2),\\
        \| f \|_{\iy} &= \esssup_{x \in [0,1]} |f(x)| &&\text{ for any }
        f \in L^{\iy}([0,1],M_2(\C)),
    \end{aligned}
    $$
    and for any bounded linear operator $A: L^2(\R,\C^2) \to L^2(\R,\C^2)$
    $$
        \| A \|_{2,2} = \sup_{u \in L^2(\R,\C^2), \left\| u \right\|_2 = 1} \left\| A u \right\|_2.
    $$

    Let $I \ss \R$, and let $l \geq 0$. By $\cH^l(I)$, we denote the Sobolev space of functions
    $f: I \to \C$ such that $f^{(i)} \in L^2(I,\C)$ for each $0 \leq i \leq l$. Note that if a gap
    in the spectrum of some Dirac operator with dislocation $\g_n$ is open for some $n \in \Z$,
    then the corresponding circle $\g_n^c$ is a $1$-dimensional smooth manifold and
    it is diffeomorphic to $1$-dimensional sphere $S^1 \ss \R^2$. In order to construct an atlas
    for $\g_n^c$ we consider the following maps: in a neighborhood of $\a_n^{\pm}$ we set
    a coordinate map $\R \ni z \mapsto \l = \pr_{h(z)} (\a_n^{\pm} \mp z^2) \in \g_n^c$, where
    $h(z) = 1\chi_+(z) +2\chi_-(z)$, i.e. $\l \in \g_n^{(1)}$ if $z > 0$, and
    $\l \in \g_n^{(2)}$ if $z < 0$; in a neighborhood of $\l \in \g_n^{(j)}$, $j = 1,2$, we set
    a coordinate map $\R \ni z \mapsto \pr_j z \in \g_n^{(j)}$.
    \begin{definition*}
        Let $I \ss \R$, $l \geq 0$, and let $\g_n$ be open gap in the spectrum of some
        Dirac operator with dislocation for some $n \in \Z$. We denote by $\cH^{l}(I,\g_n^c)$ the
        set of all continuous functions $f: I \to \g_n^c$ such that a composition of $f$ with
        any inverse coordinate map belongs to $\cH^{l}(U)$ for some open non-empty $U \ss I$.
    \end{definition*}

    Finally, we introduce the following class of functions.
    \begin{definition*}
        Let $I_1,I_2 \ss \R$ be open bounded intervals. We denote by $\mH(I_1,I_2)$ the set of
        all functions $F:I_1 \times I_2 \to \R$ satisfying the following conditions
        \begin{enumerate}[i)]
            \item $F(x,\cdot) \in C^1(I_2)$ for each fixed $x \in I_1$,
            \item $F(\cdot,y) \in \cH^1(I_1)$ for each fixed $y \in I_2$,
            \item $\left|F'_{x} (x,y)\right| \leq g(x)$ for each $(x,y) \in I_1 \ts I_2$,
            for some $g \in L^2(I_1)$.
        \end{enumerate}
    \end{definition*}
    Note that for any open bounded intervals $I_1,I_2 \ss \R$ we get
    $C^1(I_1 \ts I_2)\ss \mH(I_1,I_2) \ss C(I_1 \ts I_2)$.

    \subsection{Dirac equation}
    We introduce the $2 \times 2$ matrix-valued fundamental solution $\p(x, \l)$ of the Dirac equation
    \[ \label{p2e1}
        J y'(x) + V(x) y(x) = \l y(x),\ \ (x,\, \l) \in \R \ts \C,
    \]
    satisfying the initial condition $\p(0,\l) = I_2$, where $I_2$ is
    the identity matrix $2 \times 2$. Any matrix solution of equation (\ref{p2e1}) is expressed in
    terms of $\p$ by multiplying on the right by the initial data.
    The matrix $\p(1,\l)$ is called the monodromy matrix and satisfied
    \[ \label{p2e2}
        \p(1+x,\l) = \p(x,\l) \p(1,\l),\ \  (x,\, \l) \in \R \ts \C.
    \]
    We introduce the vector-valued fundamental solutions $\vt(x,\l)$ and
    $\vp(x,\l)$ of equation (\ref{p2e1}) satisfying the initial
    conditions
    \[ \label{p2e3}
        \vt(0,\l) = \left( \begin{array}{c} 1 \\ 0 \end{array} \right),\ \
        \vp(0,\l) = \left( \begin{array}{c} 0 \\ 1 \end{array} \right).
    \]
    Moreover, $\p$ is expressed in terms of $\vp$ and $\vt$ by $\p = \left( \vt, \vp \right)$.
    We define the Wronskian of two vector-valued functions
    $u = \left( \begin{smallmatrix} u_1 \\ u_2 \end{smallmatrix} \right)$ and
    $v = \left( \begin{smallmatrix} v_1 \\ v_2 \end{smallmatrix} \right)$ by
    $$
        W(u,v) = u_1 v_2 - u_2 v_1.
    $$
    It is known (see e.g. \cite{LevSar}) that Wronskian of $\vt$ and
    $\vp$ does not depend at $x$ and satisfies
    \[ \label{p2e4}
        W(\vt, \vp) = \det \p = \vt_1 \vp_2 - \vp_1\vt_2 = 1.
    \]
    We recall the known results about $\p$ in the following
    theorem (see e.g. \cite{Kor01}).

    \begin{theorem} \label{p2t1}
    Let $V \in \cP$. Then for each $\l \in \C$ there exists a unique
    solution $\p$ of equation (\ref{p2e1}). For each $x \in [0,1]$ the
    function $\p(x, \cdot)$ is entire. For each $\l \in \C$ the function
    $\p(\cdot, \l) \in \cH^1([0,1])$ and satisfies
    \[
    \label{p2e8}
    \left|\p(x,\l)\right| \leq e^{\left\|V\right\|_{\cP} +
    \left|\Im \l\right| x},\qq \forall \ \  (x,\l) \in [0,1] \ts \C.
    \]
    \end{theorem}

    Above we have introduced the Dirichlet eigenvalues $\m_n$, $n \in \Z$, for
    the Dirac equation (\ref{p2e1}) with the Dirichlet boundary
    conditions $y_1(0) = y_1(1) = 0$ and the Neumann eigenvalues $\n_n$,
    $n \in \Z$, for the Dirac equation (\ref{p2e1}) with
    the Neumann boundary conditions $y_2(0) = y_2(1) = 0$, where $y = \ma y_1 & y_2
    \am^{\top}$. It is well known (see e.g. \cite{LevSar}) that $\m_n,
    \n_n \in [\a_n^{-},\a_n^+]$, $n \in \Z$. If the gap $\g_n$
    degenerates, then $\m_n = \n_n$. It follows from
    (\ref{p2e3}) that eigenvalues $\m_n$ are zeros of an entire function
    $\vp_1(1,\cdot)$, i.e. $\vp_1(1, \m_n) = 0$ for all $n \in \Z$. Thus,
    $\vp(x,\m_n)$ is the eigenfunction for the eigenvalue $\m_n$ for
    each $n \in \Z$. Below we need the following identity (see e.g. \cite{Kor01})
    \[ \label{p2e19}
        \left\|\vp(\cdot,\m_n)\right\|^2 = - \vp_2(1,\m_n)
        \left. \partial_{\l} \vp_1(1,\l) \right|_{\l = \m_n},
    \]
    where $n \in \Z$ and $\partial_{\l} u = u'_{\l} = \frac{\partial u}{\partial \l}$,
    and $\left\| \cdot \right\|$ was defined above. Note also that each
    Dirichlet eigenvalue is simple, i.e. $\pa_{\l} \vp_1(1,\m_n) \neq 0$ for any $n \in \Z$.

    \subsection{Periodic Dirac operator}
    For equation (\ref{p2e1}) we introduce the Lyapunov function by
    \[ \label{p2e9}
        \D(\l) = \frac{1}{2} \Tr \p(1,\l) = \frac{1}{2}(\vp_2(1,\l) + \vt_1(1,\l)), \ \ \l \in \C.
    \]
    Due to Theorem \ref{p2t1} the function $\D$ is entire and describes the spectrum of a periodic
    Dirac operator on the line by:
    $$
        \D(\a_n^{\pm}) = (-1)^{n};\qq \left|\D(\l)\right| \leq 1,\, \l \in \s_n;\qq
        \left|\D(\l)\right| > 1,\, \l \in \g_n,\qq n \in \Z.
    $$
    We define the Weyl-Titchmarsh functions $m_{\pm}$ and the Bloch solutions $\p^{\pm}$ of
    equation (\ref{p2e1}) as follows
    \[ \label{p2e10}
        m_{\pm}(\l) = \frac{a(\l) \mp b(\l)}{\vp_1(1,\l)},\qq
        \p^{\pm}(x,\l) = \vt(x,\l) + m_{\pm}(\l) \vp(x,\l),\qq (x,\l) \in \R \times \mathbb{C}_+,
    \]
    where the functions $a(\cdot)$ and $b(\cdot)$ have the forms:
    \[ \label{p2e11}
        a(\l) = \frac{\vp_2(1,\l) - \vt_1(1,\l)}{2},\ \ b(\l) = \sqrt{\D(\l)^2 - 1},\ \ \l \in \C_+,
    \]
    and the branch of the square root is defined by $(-1)^n i\sqrt{\D(\l+i0)^2 - 1} \leq 0$ for
    $\l \in \s_n$.
    The function $a(\cdot)$ is entire and it is easy to see that $b(\cdot)$ admits an analytic
    continuation from $\C_+$ onto the Riemann surface $\L$ introduced above. In addition,
    the following identity holds in each open circle gap $\g_n^c = \overline{\g}_n^{(1)} \cup
    \overline{\g}_n^{(2)} \ss \L$:
    \[ \label{p2e23}
        b(\l) = (-1)^{n+j+1} \left| \D^2(\l) - 1 \right|^{1/2},\qq \l \in \ol \g_n^{(j)},\ \ j=1,2.
    \]
    Due to (\ref{p2e10}), $m_{\pm}(\cdot)$ and $\p^{\pm}(x,\cdot)$ admit a meromorphic continuation
    from $\C_+$ onto the Riemann surface $\L$. It follows from definition of $b(\cdot)$ that
    $b(\l_*) = -b(\l)$ for any $\l \in \L$, which yields
    \[ \label{p2e17}
        m_{\pm}(\l_*) = m_{\mp}(\l),\qq \p^{\pm}(x,\l_*) = \p^{\mp}(x,\l),\qq (x,\l) \in \R \times \L.
    \]
    If $\l \in \L$ is not a pole of $m_{\pm}(\cdot)$,
    then $e^{\mp i k(\l) x} \p^{\pm}(x,\l)$ is 1-periodic as function of $x$,
    where $k(\l)$ is quasimomentum defined by $\D(\l)=\cos
    k(\l)$. One can introduce quasimomentum as conformal mapping (see
    \cite{Mi, Kor96, KK}). We introduce the effective masses $M_n^{\pm} = 1/\l''(\a_n^{\pm})$,
    where $\l(k)$ is the inverse function for $k(\l)$. In \cite{KK} it was shown that
    $M_n^{\pm} = -\D(\a_n^{\pm})\D'(\a_n^{\pm})$ and $\pm M_n^{\pm} > 0$. Note that
    quasimomentum is real-valued on the
    spectral bands and $k(\a_n^{\pm}) = \pi n$, $n \in \Z$. This implies
    that if $\l = \a_n^{+}$ or $\a_n^{-}$ and if $\l$ is not a pole of
    $m_{\pm}(\cdot)$, then $\p^{+}(\cdot,\l) = \p^{-}(\cdot,\l)$ is periodic or
    antiperiodic solutions of the Dirac equation. For any $\l \in
    (\a_n^+, \a_{n+1}^-)$, $n \in \Z$, the solutions
    $\p^{+}(\cdot,\l)$ and $\p^{-}(\cdot,\l)$ are linearly independent, uniformly bounded on
    the line, and do not decrease at infinity. On the other hand, if $\l
    \in \g_n^{(1)}$ is not a pole of $m_{\pm}(\cdot)$, then $\p^{\pm}(x,\l)$
    decrease exponentially as $x \to \pm \iy$ and increase
    exponentially as $x \to \mp \iy$, which yields $\p^{\pm}(\cdot,\l) \in
    L^2(\mathbb{R}_{\pm},\C^2)$. If $\l = \m_n \neq \a_n^{\pm}$, i.e. $\l$ is a pole of $m_+(\cdot)$,
    then $\vp(\cdot,\l)$ and $\vt(\cdot,\l)$ belong to $L^2(\R_+,\C^2)$ or $L^2(\R_-,\C^2)$.
    We need the following lemma (see e.g. Lemma 3.2 in \cite{MokKor}).

    \begin{lemma} \label{p2l3}
    \begin{enumerate}[i)]
        \item In any open gap $\g_n^c$, $n \in \Z$, the following asymptotic holds true:
        \[ \label{p2e12}
            b(\a_n^{\pm} \mp z^2) = (-1)^n z \sqrt{2\left|M_n^{\pm}\right|} + O(z^3)
        \]
        as $z \to 0$. Moreover, if in addition $\m_n = \a_n^{\pm}$, then we get
        \[ \label{p2e18}
            2 (-1)^n M_n^{\pm} = \vt_2(1,\a_n^{\pm})\left\|\vp(\cdot,\a_n^{\pm})\right\|^2.
        \]
        \item For any $\l \in \L$ the following identities hold true:
        \[ \label{p2e22}
            \begin{aligned}
                a^2(\l) - b^2(\l) &= -\vp_1(1,\l) \vt_2(1,\l), \\
                m_+(\l) m_-(\l) &= - \frac{\vt_2(1,\l)}{\vp_1(1,\l)}.
            \end{aligned}
        \]
    \end{enumerate}
    \end{lemma}

    Note that (\ref{p2e22}) implies that $\m_n = \a_n^{j}$ or $\n_n = \a_n^{j}$ for some $n \in \Z$,
    and $j = \pm$ if and only if $a(\a_n^{j}) = 0$. The second identity in (\ref{p2e22}) allows us to
    compare $m_+$ and $m_-$ on the circle gaps. We also need the following simple lemma.
    \begin{lemma} \label{p2l4}
        If $\l \in \L$ is a pole of $m_+(\cdot)$, and $\l \neq \a_n^{\pm}$, $n \in \Z$, then $a(\l) = -b(\l)$.
    \end{lemma}
    \begin{proof}
        Since $\l$ is a pole of $m_+(\cdot)$, it follows that $\vp_1(1,\l) = 0$ and
        $a(\l) - b(\l) \neq 0$. Using the first identity in (\ref{p2e22}), we get $a(\l) + b(\l) = 0$.
    \end{proof}

    \subsection{Shifted Dirac equation}
    We consider the shifted Dirac equation
    \[ \label{p2e14}
        Jy'(x) + V(x+t)y(x) = \l y(x),\ \ (x, \l) \in \R \ts \C,
    \]
    where $t \in \R$ is the shift parameter. The potential $V(\cdot+t) \in \cP$ for any
    $(t,V) \in \R \ts \cP$. Hence for equation (\ref{p2e14}) there are all objects introduced
    for (\ref{p2e1}). We add the dependence on $t$ to these objects if they do not constant as
    function of $t$.

    Note that if $y(x,\l)$ is a solution of equation (\ref{p2e1}), then
    $\widetilde{y}(x,\l,t) = y(x+t,\l)$ is a solution of equation (\ref{p2e14}).
    Thus using (\ref{p2e2}) and $\det \psi(x,\l) = 1$ for any $(x,\l) \in \R \times \C$, we get
    \[ \label{p2e15}
        \begin{aligned}
            \p(x,\l,t) & = \p(x+t,\l)\p^{-1}(t,\l),\\
            \p(1,\l,t) & = \p(t,\l)\p(1,\l)\p^{-1}(t,\l),\\
        \end{aligned}
    \]
    where $(x,\l,t) \in \R \ts \C \ts \R$. Using (\ref{p2e9}), (\ref{p2e15}), and the fact that the
    traces of similar matrices are equal, we get for any $(\l,t) \in \C \ts \R$
    $$
        \begin{aligned}
            \D(\l,t) = \frac{1}{2} \Tr \left( \p(t,\l)\p(1,\l)\p^{-1}(t,\l) \right)
            = \frac{1}{2} \Tr \p(1,\l) = \D(\l,0).
        \end{aligned}
    $$
    It gives that $b(\l)$, $k(\l)$, $M_n^{\pm}$ also does not depend on $t$ and we
    do not write the argument $t$ of these functions.

\section{Weyl-Titchmarsh functions} \label{p8}
    In this section we describe properties of the Weyl-Titchmarsh functions $m_{\pm}(\l,t)$.
    Above we show that $m_{\pm}(\cdot,t)$ are meromorphic functions on $\L$ for any
    $t \in \R$. Moreover, it follows from \cite{MokKor} that $m_+(\cdot,t)$ has exactly
    one pole $\m_n(t)$ in each open gap $\g_n^c$, which is a state of $H^+_t$, and
    there are no other poles.
    \begin{lemma} \label{p3l1}
        Let a gap $\g_n$ be open for some $n \in \Z$ and let $t \in \R$. Then we have
        \begin{enumerate}[i)]
            \item $\partial_{\l} m_+(\l,t) < 0$ for any $\l \in \ol \g_n^{(1)}$, $\l \neq \m_n(t)$;
            \item $\partial_{\l} m_-(\l,t) > 0$ for any $\l \in \ol \g_n^{(1)}$, $\l \neq \m_n(t)_*$.
        \end{enumerate}
        % $$
            % \partial_{\l} m_+(\l,t) < 0,\qq \partial_{\l} m_-(\l,t) > 0,
            % \qq \l \in \ol \g_n^{(1)},\ \ \l \neq \m_n(t).
        % $$

        % i) If $\m_n(t) \in \ol \g_n^{(1)}$, then $m_+(\cdot,t)$ is a strictly decrease
        % function on $\lan \a_n^-,\m_n(t) \ran \ss \g_n^{(1)}$ and
        % $\lan \m_n(t),\a_n^+ \ran \ss \g_n^{(1)}$ and $m_-(\cdot,t)$ is a strictly increase
        % function on $\lan \a_n^+,\a_n^- \ran \ss \g_n^{(2)}$.

        % ii) If $\m_n(t) \in \ol \g_n^{(2)}$, then $m_+(\cdot,t)$ is a strictly decrease function
        % on $\lan \a_n^-,\a_n^+ \ran \ss \g_n^{(1)}$ and $m_-(\cdot,t)$ is a strictly increase
        % function on $\lan \a_n^-,\m_n(t)_* \ran \ss \g_n^{(1)}$ and $(\m_n(t)_*,\a_n^+) \ss \g_n^{(1)}$.
    \end{lemma}
    \begin{proof}
        It is well known that $\mp m_{\pm}(\cdot,t)$ are the Herglotz functions on $\C_+$.
        Each Herglotz function $u$ admits a representation
        $$
            u(z) = a + bz + \int_{\R} \frac{1 + tz}{t - z} d \varrho(t),\qq z \in \C_+
        $$
        where $a,b \in \R$, $b \geq 0$, and $\varrho: \R \to \R$ is nondecreasing function
        (see e.g. Section 59, Theorem 2 in \cite{Tes14}). Differentiating this representation by $z$,
        we get
        $$
            u'(z) = b + \int_{\R} \frac{1 + t^2}{(t - z)^2} d \varrho(t),\qq z \in \C_+,
        $$
        which yields that if a Herglotz function admits a continuation on an interval
        $I \ss \R$ from $\C_+$, then it is strictly increase on $I$.
        The function $m_+(\cdot,t)$ has only one pole $\m_n(t)$ on $\g_n^c$ and it admit a
        continuation on $\ol \g_n^{(1)} \sm \{ \m_n(t) \}$. Since $m_+(\l,t) = m_-(\l_*,t)$ it
        follows that $m_-(\cdot,t)$ has only one pole $\m_n(t)_*$ on $\g_n^c$ and it admit a
        continuation on $\ol \g_n^{(1)} \sm \{ \m_n(t)_* \}$.
    \end{proof}
    Now we discuss properties of $m_+(\l,t)$ as function of $t$. We introduce the following functions
    \[ \label{p6e2}
        \begin{aligned}
            g_{\l}(z,t) &= z^2(q_1(t) + \l) - 2zq_2(t) - (q_1(t)-\l),\quad
            Q_{\l}(t_0, t) = \int_{t_0}^{t} \left| q_1(\t) + \l \right| d\t,\\
            G_{\l}(t_0,t) &= \int_{t_0}^t \left| g_{\l}(m_+(\l,t_0),\t) \right| d\t,\quad
            W_{\l}(t_0,t) =\left| t- t_0 \right|^{1/2}
            \left| \int_{t_0}^t \left| V(\t) + \l \1_2 \right|^2 d\t \right|^{1/2},
        \end{aligned}
    \]
    where $t,t_0 \in \R$, $z,\l \in \C$, and $\1_2$ is the identity matrix $2 \times 2$.
    \begin{lemma} \label{a1l3}
        \begin{enumerate}[i)]
            \item Let $\l \in \L$ not be a pole of $m_+(\cdot, t_0)$ for some $t_0 \in \R$. Then
            there exists $\dot{m}_+(\l,t)$ for almost all $t$ in a sufficiently small neighborhood
            of $t_0$ and
            \begin{align}
                \dot{m}_+(\l,t) &= g_{\l}(m_+(\l,t),t), \label{dotm} \\
                m_+(\l,t) &= m_+(\l,t_0) + \int_{t_0}^t g_{\l}(m_+(\l,t_0),\t) d\t +
                O(G_{\l}(t_0,t)W_{\l}(t_0,t)) \text{ as $t \to t_0$}. \label{p3e3}
            \end{align}
            \item Let $\l \in \L$ be a pole of $m_+(\cdot, t_0)$ for some $t_0 \in \R$ and
            let $u(\l,t) = 1/m_+(\l,t)$. Then there exists $\dot{u}(\l,t)$ for almost all $t$
            in a sufficiently small neighborhood of $t_0$ and
            \begin{align}
                \dot{u}(\l,t) &= -\frac{g_{\l}(m_+(\l,t),t)}{m_+(\l,t)^2}, \label{dotm2} \\
                u(\l,t) &= -\int_{t_0}^t (q_1(\t)+\l) d\t  + O(Q_{\l}(t_0,t)W_{\l}(t_0,t))
                \text{ as $t \to t_0$}. \label{p3e6}
            \end{align}
        \end{enumerate}
    \end{lemma}
    \begin{proof}
        i) Differentiating $m_+(\l,t)$ by $t$ and using (\ref{p6e1}), we get
        \[ \label{p3e8}
            \dot{m}_+(\l,t) = -(q_1(t)-\l)-2m_+(\l,t)q_2(t) + \frac{\vt_2(1,\l,t) +
            2m_+(\l,t)a(\l,t)}{\vp_1(1,\l,t)}(q_1(t)+\l)
        \]
        for any $(\l,t) \in \L \times \R$ such that $\l$ is not a pole of $m_+(\cdot,t)$.
        Using (\ref{p2e22}), we get the following identity
        \[ \label{p3e9}
            \frac{\vt_2(1,\l,t) + 2m_+(\l,t)a(\l,t)}{\vp_1(1,\l,t)} = m_+(\l,t)^2.
        \]
        Substituting (\ref{p3e9}) in (\ref{p3e8}), we obtain (\ref{dotm}).

        We introduce the notation $m_+^0 = m_+(\l,t_0)$. By Lemma \ref{p6l1} and \ref{a1l2},
        $m_+(\l,\cdot)$ is an absolutely continuous function. Integrating (\ref{dotm}), we get
        \[ \label{p3e1}
                m_+(\l,t) - m_+^0 = \int_{t_0}^t g_{\l}(m_+^0,\t) d\t
                + \int_{t_0}^t (g_{\l}(m_+(\l,\t),\t) - g_{\l}(m_+^0,\t)) d\t .
        \]
        We introduce the following functions for any $t,t_0 \in \R$, $\l \in \C$
        $$
                M_{\l}(t_0,t) = \max_{\t \in [t_0,t]} \left| m_+(\l,\t) - m_+^0\right|, \quad
                S_{\l}(t_0,t) = \max_{x \in [m_+(\l,t),m_+^0]} \left| \partial_x g_{\l}(x,t) \right|.
        $$
        Using the definition of $g_{\l}$, we get
        $\left|\partial_x g_{\l}(x,\t) \right| \leq 2 \left| q_1(\t) +
        \l \right| \left| x \right| + 2\left| q_2(\t) \right|$, which yields
        \[ \label{p3e2}
            \begin{aligned}
                \int_{t_0}^t S_{\l}(t_0,\t) d\t &= 2 \int_{t_0}^t \left| q_1(\t) + \l \right|
                \max_{x \in [m_+(\l,\t),m_+^0]}\left| x \right| d\t +
                2\int_{t_0}^t \left| q_2(\t) \right| d\t \\
                &\leq 2C\int_{t_0}^t \left| q_1(\t) + \l \right| d\t +
                2\int_{t_0}^t \left| q_2(\t) \right| d\t = O(W_{\l}(t_0,t))
            \end{aligned}
        \]
        as $t \to t_0$. Here we used the fact that $m_+(\l,\cdot)$ is absolutely continuous in
        some neighborhood of $t_0$ and thus it follows that
        $$
            \max_{\t \in [t_0,t]} \max_{x \in [m_+(\l,\t),m_+^0]}\left| x \right| =
            \max_{\t \in [t_0,t]} \left| m_+(\l,\t) \right| \leq C,
        $$
        for some constant $C > 0$ and $t$ in some neighborhood of $t_0$. The application of the
        mean value theorem yields
        \[ \label{p3e5}
            \begin{aligned}
                \int_{t_0}^t \left| g_{\l}(m_+(\l,\t),\t) - g_{\l}(m_+^0,\t)\right| d\t &=
                \int_{t_0}^t \left| m_+(\l,\t) - m_+^0 \right| S_{\l}(t_0, \t) d\t \\
            &\leq M_{\l}(t_0, t) \int_{t_0}^t S_{\l}(t_0, \t) d\t.
            \end{aligned}
        \]
        Then estimating $M_{\l}(t_0,t)$ from (\ref{p3e1}) and using (\ref{p3e2}), (\ref{p3e5}),
        we get for $t \to t_0$
        $$
            M_{\l}(t_0,t) = G_{\l}(t_0,t) + M_{\l}(t_0,t) O(W_{\l}(t_0,t)),
        $$
        which yields $M_{\l}(t_0,t) = O(G_{\l}(t_0,t))$ as $t \to t_0$. Substituting this
        asymptotic and (\ref{p3e2}) in (\ref{p3e5}), we prove (\ref{p3e3}).

        ii) If $\l$ is a pole of $m_+(\cdot,t_0)$, then $u(\l,t_0) = 1/m_+(\l,t_0) = 0$.
        Using (\ref{dotm}), we obtain (\ref{dotm2}). Integrating (\ref{dotm2}), we get
        \[ \label{p3e4}
            u(\l,t) = -\int_{t_0}^t (q_1(\t) + \l) d\t + \int_{t_0}^t (q_1(\t) - \l)u(\l,\t)^2 d\t +
            2\int_{t_0}^t q_2(\t)u(\l,\t) d\t.
        \]
        We introduce the following function
        $$
            M_{\l}(t_0,t) = \max_{\t \in [t_0,t]} \left| u(\l,\t) \right|,\ \ t_0,t \in \R.
        $$
        It follows from the definition of $W_{\l}$ that for any $t_0,t \in \R$ and $\l \in \C$
        $$
            \int_{t_0}^t \left| q_1(\t) - \l \right| d\t \leq W_{\l}(t_0,t),\ \
            \int_{t_0}^t \left| q_2(\t) \right| d\t \leq W_{\l}(t_0,t).
        $$
        Using this inequality and estimating $M_{\l}(t_0,t)$ from (\ref{p3e4}), we get
        $$
            M_{\l}(t_0,t) \leq Q_{\l}(t_0,t) + M_{\l}(t_0,t)^2 W_{\l}(t_0,t) +
            2M_{\l}(t_0,t) W_{\l}(t_0,t)
        $$
        which yields $M_{\l}(t_0,t) = O(Q_{\l}(t_0,t))$ as $t \to t_0$. Substitution of this
        asymptotic in (\ref{p3e4}) yields (\ref{p3e6}).
    \end{proof}
\section{Dirac operator with dislocation} \label{p4}
    \subsection{Dirac equation with dislocation}
    Consider the Dirac equation with a dislocated potential $V_t$, where $(t,V) \in \R \times \cP$
    \[ \label{p4e1}
        Jy'(x) + V_t(x) y(x) = \l y(x),\ \ (x,\l,t) \in \R \ts \C \ts \R.
    \]
    Let $\F(x,\l,t) = \left(\begin{smallmatrix} \F_1 \\ \F_2 \end{smallmatrix}\right)(x,\l,t)$ and
    $\vT(x,\l,t) = \left(\begin{smallmatrix}\vT_1 \\ \vT_2 \end{smallmatrix}\right)(x,\l,t)$ be
    vector-valued fundamental solutions of equation (\ref{p4e1}) satisfying the initial conditions:
    $$
        \F_1(0,\l,t) = \vT_2(0,\l,t) = 0, \ \ \F_2(0,\l,t) = \vT_1(0,\l,t) = 1.
    $$
    We introduce a fundamental $2 \times 2$ matrix-valued solution $\P = \left( \vT , \F \right)$
    of equation (\ref{p4e1}). Note that equation (\ref{p4e1}) coincides with equation (\ref{p2e14})
    on the half-lines. Hence a solution of (\ref{p4e1})
    coincides with a solution of (\ref{p2e14}) on the half-lines. Therefore, we get for any
    $(x,\l,t) \in \R \ts \C \ts \R$
    \[ \label{p4e21}
        \begin{aligned}
            \F (x,\l,t) & = \ca \vp(x,\l,0), &\text{if } x\leq0, \\ \vp(x,\l,t), &\text{if } x>0, \ac &\ \
            \vT (x,\l,t) & = \ca \vt(x,\l,0), &\text{if } x\leq0, \\ \vt(x,\l,t), &\text{if } x>0, \ac
        \end{aligned}
    \]
    where $\vp(x,\l,t)$ and $\vt(x,\l,t)$ are solutions of (\ref{p2e14}). Let us define
    exponentially decreasing solutions as $x \to \pm \iy $. Let $\P^{\pm}(\cdot,\l,t)$,
    $(\l,t) \in \C_+ \ts \R$, be a solution of equation (\ref{p4e1}) such that
    $$
        \begin{aligned}
            \P^{+}(x,\l,t) &= \p^+(x,\l,t), && \text{if $x > 0$},\\
            \P^{-}(x,\l,t) &= \p^-(x,\l,0), && \text{if $x < 0$},
        \end{aligned}
    $$
    where $\p^{\pm}(\cdot,\l,t)$ are the Bloch solutions of equation (\ref{p2e14}). It follows from this
    definition that $\P^{\pm}(x,\l,t)$ are exponentially decrease as $x \to \pm \iy$ for any
    $(\l,t) \in \C_+ \ts \R$. Using (\ref{p2e10}), we obtain for any $(x,\l,t) \in \R \ts \C_+ \ts \R$
    \[ \label{p4e3}
        \begin{aligned}
            \P^{+}(x,\l,t) & = \vT(x,\l,t) + m_{+}(\l,t) \F(x,\l,t),\\
            \P^{-}(x,\l,t) & = \vT(x,\l,t) + m_{-}(\l,0) \F(x,\l,t).
        \end{aligned}
    \]
    Now using (\ref{p4e3}), we compute the Wronskian of $\P^{-}$ and $\P^{+}$ and introduce a
    function $w$ by
    \[ \label{p4e32}
        w(\l,t) =  W( \P^{-}(\cdot,\l,t), \P^{+}(\cdot,\l,t) ) = m_+(\l,t) - m_-(\l,0),\qq
        (\l,t) \in \C_+ \ts \R.
    \]
    We also introduce the following function
    \[ \label{p4e33}
        \G(\l,t) = \frac{w(\l,t)}{m_+(\l,t)m_-(\l,0)},\qq (\l,t) \in \C_+ \ts \R.
    \]
    The resolvent $(H_t-\l)^{-1}$ is an analytic on $\C_+$ operator-valued function of
    $\l$ and for any $\l \in \C_+$ it is an integral operator with the kernel given by
    \[ \label{p4e25}
        R(x,y,\l,t) = \frac{-1}{w(\l,t)}
        \ca
            \P^{+}(x,\l,t) \otimes \P^{-}(y,\l,t), & \text{if $x > y$}, \\
            \P^{-}(x,\l,t) \otimes \P^{+}(y,\l,t), & \text{if $x < y$},
        \ac
    \]
    where $x,y \in \R$, $(\l,t) \in \C_+ \ts \R$, and $\otimes$ means
    the tensor product of vectors, i.e.
    $$
        A \otimes B =
        \ma
            A_1 B_1 & A_1 B_2 \\
            A_2 B_1 & A_2 B_2
        \am,
        \qq A = \ma A_1 \\ A_2 \am,\, B = \ma B_1 \\ B_2 \am,\qq A, B \in \C^2.
    $$
    We describe analytical properties of these functions.

    \begin{lemma} \label{p4l1}
        Functions $\P^{\pm}(x,\cdot,t)$, $w(\cdot,t)$, $\G(\cdot,t)$ and $R(x,y,\cdot,t)$ admit
        a meromorphic continuation from $\C_+ \ss \L_1$ onto $\L$ for any $x,y,t \in \R$ and
        satisfy
        \[ \label{p4e5}
            \begin{aligned}
                &w(\l_*,t) = m_-(\l,t) - m_+(\l,0),\ \ \G(\l_*,t) =
                \frac{m_-(\l,t) - m_+(\l,0)}{m_-(\l,t) m_+(\l,0)},\\
                &\begin{aligned}
                    \P^{+}(x,\l_*,t) & = \left\{ \begin{aligned}
                        &\p^{-}(x,\l,t) && \text{if} \qq x \geq 0,\\
                        &\p^{+}(x,\l,t) + w(\l,t) \vp(x,\l,0) && \text{if} \qq x < 0,
                        \end{aligned}\right. \\
                    \P^{-}(x,\l_*,t) & = \left\{ \begin{aligned}
                        &\p^{-}(x,\l,t) + w(\l,0) \vp(x,\l,t) && \text{if} \qq x > 0,\\
                        &\p^{+}(x,\l,0) && \text{if} \qq x \leq 0,
                        \end{aligned} \right.
                \end{aligned}
            \end{aligned}
        \]
        where $\l_* \in \L_2$ is projection of $\l \in \L_1$ on the second sheet.
    \end{lemma}
    \begin{proof}
        Since $\p(x,\cdot,t)$ is entire and $m_{\pm}(\cdot,t)$, $\p^{\pm}(x,\cdot,t)$ are
        meromorphic on $\L$ for any $x,t \in \R$, it follows that the functions from
        statement of the lemma admit a meromorphic continuation from $\C_+$ onto $\L$ as rational
        functions of $\p$, $m_{\pm}$, and $\p^{\pm}$.
        Using identities (\ref{p2e17}) and $\p(x,\l_*,t) = \p(x,\l,t)$, for any
        $(x,\l,t) \in \R \ts \L \ts \R$, we obtain (\ref{p4e5}).
    \end{proof}

    Now we show that the resolvent admits a meromorphic continuation from $\C_+$
    onto $\L$ in weak sense .
    \begin{lemma} \label{p4l2}
        Let $\eta \in C_o^{\iy}(\R,\C^2)$, and let $t \in \R$. Then
        $f_{\eta}(\l,t) = \left( \left( H_t - \l \right)^{-1} \eta, \eta \right)$ admits a
        meromorphic continuation from $\C_+ \ss \L_1$ onto the Riemann surface $\L$ as
        function of $\l$. Moreover, $\l \in \L$ is a state of $H_t$ if and only if it is a pole of
        $(v(x),R(x,y,\cdot,t)v(y))$ for some $v \in C_o^{\iy}(\R,\C^2)$ and for each $(x,y) \in U$,
        where $U \ss \R^2$ is some set of nonzero measure.
    \end{lemma}
    \begin{proof}
        In Lemma \ref{p4t1} we show independently that $H_t$ is self-adjoint for any $t \in \R$,
        which yields that $f_{\eta}(\cdot,t)$ is an analytic function on $\C_+$ for any
        $t \in \R$ and $\eta \in C_o^{\iy}(\R,\C^2)$. Since the resolvent is an integral operator,
        it follows that
        \[ \label{p4e2}
            f_{\eta}(\l,t) = \int_{\R} \int_{\R} (\eta(x), R(x,y,\l,t) \eta(y)) dx dy =
            \int_{\R} \int_{\R} R_{\eta}(x,y,\l,t) dx dy,
        \]
        where we introduce $R_{\eta}(x,y,\l,t) = (\eta(x), R(x,y,\l,t) \eta(y))$ for any
        $x,y,t \in \R$ and $\l \in \C_+$.
        By Lemma \ref{p4l1}, $R_{\eta}(x,y,\cdot,t)$ admits a meromorphic continuation
        from $\C_+$ onto the Riemann surface $\L$ for any $x,y,t \in \R$.
        For any $(\l,t) \in \L \ts \R$ such that $\l$ is not a pole
        of $R_{\eta}(x,y,\cdot,t)$, the function $R_{\eta}(x,y,\l,t)$ is a continuous function
        of $x,y$ with compact support. Thus, $f_{\eta}(\l,t)$ is correctly defined for such
        $(\l,t)$, which yields that $f_{\eta}(\cdot,t)$ admits a meromorphic continuation from
        $\C_+$ onto $\L$.

        By definition, a state of $H_t$ is a pole of $f_{v}(\cdot,t)$ for some
        $v \in C_o^{\iy}(\R,\C^2)$. It follows from integral representation (\ref{p4e2}),
        that $\l$ is a pole of $f_{v}(\cdot,t)$ for some $t \in \R$ and
        $v \in C_o^{\iy}(\R,\C^2)$ if and only if it is a pole of
        $R_v(x,y,\l,t)$ for any $(x,y) \in U$, where $U \ss \R^2$ is some set of
        nonzero measure.
    \end{proof}

    \subsection{Spectrum}
    We introduce a free Dirac operator $H^0 = J \frac{d}{dx}$ acting on $L^2(\R,\C^2)$.
    This operator is self-adjoint on the domain $\cD(H^0) = \cH^1(\R,\C^2)$.
    % We introduce the Hilbert space $L^2_{loc,u}(\R,M_2(\R))$ of $2 \ts 2$ matrix-valued function $V$ equipped with the norm $\left\| V \right\|_{loc,u}^2 = \sup_{t \in \R} {1\/2} \int_{t}^{t+1} \Tr V^2(x)dx$. Now we show that each $V \in L^2_{loc,u}(\R,M_2(\R))$ is $H^0$-bounded operator in $L^2(\R,\C^2)$. So we need some notation. Recall that $\left\| \vp \right\|^2 = \int_{\T} \left( \left|\vp_1(x) \right|^2 + \left|\vp_2(x) \right|^2 \right) dx$. We also introduce the following norm $\left\| \vp \right\|_{2}^2 = \int_{\R} \left( \left|\vp_1(x) \right|^2 + \left|\vp_2(x) \right|^2 \right) dx$ if $\vp$ is vector-valued and $\left\| \vp \right\|_{2}^2 = \int_{\R} \left|\vp(x) \right|^2 dx$ if $\vp$ is scalar function. Let also $L^2_{loc,u}(\R,\C)$ be the Hilbert space of a function $\vp: \R \to \C$ equipped with the norm $\left\| \vp \right\|_{loc,u}^2 = \sup_{t \in \R} \int_{t}^{t+1}\left|\vp(x) \right|^2 dx$.
    \begin{lemma} \label{p2l1}
        Let $V \in L^2_{loc,u}(\R, M_2(\C))$ and $\ve > 0$. Then for any $u \in \cD(H^0)$ we get
        $$
            \left\|V u\right\|_{\R}^2 \leq \left\|V\right\|^2_{loc,u}
            \left( 2\ve \left\|H^0 u\right\|_{\R}^2 + \frac{16}{\ve} \left\|u\right\|_{\R}^2 \right).
        $$
    \end{lemma}
    \begin{proof}
        Let $V =
        \left( \begin{smallmatrix} v_{11} & v_{12} \\ v_{21} & v_{22} \end{smallmatrix} \right)
        \in L^2_{loc,u}(\R, M_2(\C))$ and $u =
        \left( \begin{smallmatrix} u_1 \\ u_2 \end{smallmatrix} \right) \in \cD(H^0)$. Then we get
        $$
            \left\|V u\right\|_{\R}^2 \leq \left\|v_{11} u_1\right\|_{L^2(\R)}^2 +
            \left\|v_{12} u_2\right\|_{L^2(\R)}^2 + \left\|v_{21} u_1\right\|_{L^2(\R)}^2 +
            \left\|v_{22} u_2\right\|_{L^2(\R)}^2.
        $$
        Each $v_{ij} \in L^2_{loc, u}(\R)$ and $u_i \in \cH^1(\R)$, where $i,j \in \{ 1,2 \}$.
        Consider one of the components, where the indices are omitted. We have for any $\ve > 0$
        $$
        \begin{aligned}
            \left\|v u\right\|_{L^2(\R)}^2 &= \int_{\R} \left|v(x)u(x)\right|^2 dx =
            \sum_{n \in \Z} \int_0^1 \left|v(x+n) u(x+n)\right|^2 dx \\
            &\leq \sum_{n \in \Z} \max_{s \in [0,1)}
            \left|u(s+n)\right|^2 \int_0^1 \left|v(x+n)\right|^2 dx \\
            &\leq \left\|v\right\|_{loc,u}^2 \sum_{n \in \Z}
            \left( \ve \int_0^1 \left|u'(x+n)\right|^2 dx +
            \frac{8}{\ve} \int_0^1 \left|u(x+n)\right|^2 dx \right) \\
            &= \left\|v\right\|_{loc,u}^2 \left( \ve \left\|u'\right\|_{L^2(\R)}^2
            + \frac{8}{\ve} \left\|u\right\|_{L^2(\R)}^2  \right).
        \end{aligned}
        $$
        Here we used the following inequality for $u \in \cH^1(\R)$ and $\ve > 0$
        (see e.g. Theorem IX.28 in \cite{RSvol2})
        $$
            \max_{s \in [0,1)}\left|u(s)\right|^2 \leq \ve \int_0^1 \left|u'(x)\right|^2 dx
            + \frac{8}{\ve} \int_0^1 \left|u(x)\right|^2 dx.
        $$
        Since $\left\|v_{ij}\right\|_{loc,u} \leq \left\|V\right\|_{loc,u}$ and
        $\left\|u\right\|_{\R}^2 = \left\|u_1\right\|_{L^2(\R)}^2 + \left\|u_2\right\|_{L^2(\R)}^2$,
        we get for any $\ve > 0$
        $$
            \left\|V u\right\|_{\R}^2 \leq \left\|V\right\|^2_{loc,u}
            \left( 2\ve \left\|H^0 u\right\|_{\R}^2 + \frac{16}{\ve} \left\|u\right\|_{\R}^2 \right).
        $$
    \end{proof}
    Now we describe spectrum of the Dirac operator with dislocation.
    \begin{theorem} \label{p4t1}
        Let $(t,V) \in \R \ts \cP$. Then the operator $H_t$ is self-adjoint and satisfies:
        \begin{enumerate}[i)]
            \item $\s(H_t) = \s_{ac}(H_t) \cup \s_{disc}(H_t),\ \ \s_{ac}(H_t) \cap \s_{disc}(H_t) = \es$,
            \item $\s_{ac}(H_t) = \s_{ac}(H_0) = \bigcup_{n \in \Z} \s_n$,
            \item $\s_{disc}(H_t) \subset \bigcup_{n \in \Z} \g_n$.
        \end{enumerate}
    \end{theorem}
    \begin{proof}
        It follows from Lemma \ref{p2l1} and the Kato-Rellich theorem (see e.g. Theorem X.12 in
        \cite{RSvol2}) that $H_t = H^0 + V_t$ is a self-adjoint operator on $\cD(H_t) = \cD(H^0)$.
        We define a symmetric operator $A$ by
        $$
            \begin{aligned}
                A f & = H_t f, \ \ f \in \cD(A) = \{ \, f \in \cD(H^0) \mid f_1(0) = 0 \}.
            \end{aligned}
        $$
        The operator $H_t$ is a self-adjoint extension of the operator $A$. The operator $A$
        also admits a self-adjoint extension $A_e = H^-_0 \os H^{+}_t $. Since the spectra
        of $H^{\pm}_t$ are known, we get
        $$
            \s(A_e) = \s_{ac}(A_e) \cup \s_{disc}(A_e) = \s_{ac}(H_0) \cup \s_{disc}(H^{-}_0)
            \cup \s_{disc}(H^{+}_t).
        $$
        Operators $H_t$ and $A_e$ are self-adjoint extensions of the symmetric operator $A$ and
        the defect indices of $A$ are $n_{\pm}(A) = 1$. So it follows from \cite{Kr}
        (see formula (6.15)) that
        $$
            \dim \Ran \left( (H_t - \l)^{-1} - (A_e - \l)^{-1} \right) = 1,\qq \l \in \C_+.
        $$
        Using the Weyl's theorem on the stability of the essential spectrum, we get
        \[ \label{p4e6}
            \s_{ess}(H_t) = \s_{ess}(A_e) = \s(H_0).
        \]
        By Lemma \ref{p4l2}, for each $\eta \in C_o^{\iy}(\R,\C^2)$ and $t \in \R$ the
        function $f_{\eta}(\cdot,t)$ is meromorphic on $\L$. Using (\ref{p4e25}) and (\ref{p2e10}),
        we get that there exists $\lim_{\ve \downarrow 0} \Im f_{\eta}(\l + i \ve,t) < \iy$ for any
        $\l \in \s_{ess}(H_t) \ss \L_1$, $\l \neq \m_n(t)$ for any $n \in \Z$. Thus,
        Theorem XIII.20 from \cite{RS4} yields that singular continuous spectrum of $H_t$ is
        absent and it follows from (\ref{p4e6}) that $\s_{ac} (H_t) = \s_{ess} (H_t) = \s(H_0)$. If
        $\l \in \s(H_0)$, then we get $\left| \D(\l) \right| \leq 1$ and then any solution
        of the Dirac equation does not belong to $L^2(\R_{\pm},\C^2)$. It follows that there are
        no eigenvalues embedded into the absolutely continuous spectrum.
    \end{proof}

    %Let $\left\| A \right\|_{2,2} = \sup_{\vp \in L^2(\R,\C^2), \left\| \vp \right\|_2 = 1} \left\| A \vp \right\|_2$ be the norm of operator $A$ on $L^2(\R,\C^2)$.
    We need the following lemma to prove the smoothness of eigenvalues.
    \begin{lemma} \label{p4l5}
        Let $t_0 \in \R$. Then $H_t \to H_{t_0}$ in the norm resolvent sense as $t \to t_0$.
    \end{lemma}
    \begin{proof}
        By Theorem VIII.19 in \cite{Ques1}, it is needed to show that for $t \to t_0$
        $$
            \left\|\left( H_t+i \right)^{-1} - \left( H_{t_0}+i \right)^{-1}\right\|_{2,2} \to 0.
        $$
        Using the second resolvent identity and $\cD( H_{t_0} ) = \cD(H^0)$, we get for any
        $t,t_0 \in \R$
        \[ \label{p4e78}
            \left( H_t+i \right)^{-1} - \left( H_{t_0}+i \right)^{-1} =
            \left( H_t+i \right)^{-1} \left( V_{t_0} - V_{t} \right) \left( H^0 + i \right)^{-1}
            \left( H^0 + i \right) \left( H_{t_0}+i \right)^{-1}.
        \]
        It is well known that $\left\|\left( H_t+i \right)^{-1}\right\|_{2,2} \leq 1$ for any
        $t \in \R$.
        If $\left\|V_{t_0} - V_{t}\right\|_{loc,u} \to 0$ as $t \to t_0$, then it follows from
        Lemma \ref{p2l1} that
        $ \left\|\left( V_{t_0} - V_{t} \right) \left( H^0 + i \right)^{-1}\right\|_{2,2} \to 0$.
        We prove the required condition for each scalar component of the matrix-valued function
        $V_{t_0} - V_{t}$.
        For any scalar function $h \in L^2(\T)$ and $t \in \R$ we introduce
        $h_t(x) = h(x) \chi_-(x) + h(x+t) \chi_+(x)$, $x \in \R$. Let $v \in L^2(\T)$ and $\ve > 0$.
        Then there exists $g \in C_o^{\iy}(\T)$ such that
        $\left\|v-g\right\| < \ve$. Since $v,g \in L^2(\T)$, it follows that
        $\left\| v_t-g_t \right\|_{loc,u} = \left\| v-g \right\|_{loc,u} = \| v - g \|$ for any
        $t \in \R$. Let $\t = t - t_0$. Then we have for $\t \to 0$
        $$
                \left\|g_{t_0} - g_{t}\right\|_{loc,u}^2 = \left\|g_{\t} - g\right\|^2 =
                \int_0^{1} \left|g(\t+x) - g(x)\right|^2 dx = \int_0^{1} \left|\t g'(x) +
                O\left( \t^2 \right)\right|^2 dx.
        $$
        Since $g \in C_o^{\iy}(\T)$, there exists $C > 0$ such that
        $\max_{x \in \T} \left| g'(x) \right| < C$ and then there exists $\d > 0$ depending on
        $g$ such that $\left\|g_{t_0} - g_{t}\right\|_{loc,u}^2 \leq \ve^2$ when
        $\left|t-t_0\right| < \d$. Thus, one can estimate $\left\|v_t - v_{t_0}\right\|_{loc,u}$
        when $\left|t-t_0\right| < \d$ as follows
        $$
            \left\|v_t-v_{t_0}\right\|_{loc,u} \leq \left\|v_t-g_t\right\|_{loc,u} +
            \left\|v_{t_0}-g_{t_0}\right\|_{loc,u} + \left\|g_t-g_{t_0}\right\|_{loc,u} \leq 3\ve.
        $$
        Finally, we show that $\left(H^0 + i \right) \left( H_{t_0}+i \right)^{-1}$ is a bounded
        operator. It is easy to see that
        $$
            \left(H^0 + i \right) \left( H_{t_0}+i \right)^{-1} =
            I - V_{t_0} \left( H_{t_0}+i \right)^{-1}.
        $$
        By Lemma \ref{p2l1} the operator $V_{t_0}$ is $H^0$-bounded with any positive relative-bound,
        i.e. for any $a > 0$ there exists $b > 0$ such that for each $u \in \cD(H^0)$ we get
        $$
            \left\|V_{t_0}u\right\|_{\R} \leq a \left\|H^0 u + V_{t_0}u - V_{t_0} u\right\|_{\R}
            + b \left\|u\right\|_{\R} \leq
            a \left\|H_{t_0} u\right\|_{\R} + a \left\|V_{t_0} u\right\|_{\R} +
            b \left\|u\right\|_{\R}.
        $$
        Expressed $\left\|V_{t_0}u\right\|_{\R}$ from this inequality, we obtain
        $$
            \left\|V_{t_0}u\right\|_{\R} \leq \frac{a}{1-a} \left\|H_{t_0} u\right\|_{\R} +
            \frac{b}{1-a} \left\|u\right\|_{\R},
        $$
        which yields that $V_{t_0} \left( H_{t_0}+i \right)^{-1}$ is bounded operator on
        $L^2(\R,\C^2)$. This completes the proof of the lemma.
    \end{proof}

    \subsection{States}
    Now we prove that states of $H_t$ are zeros of some analytic function.
    Moreover, $\m_n(t)$ is a state of $H_t$ only in the case of $\m_n(t) = \m_n(0)_*$. First, we
    need the following technical lemma about solutions of dislocated Dirac equation (\ref{p4e1}).
    \begin{lemma} \label{p4l12}
        Let $(t,\l,V) \in \R \ts \C \ts \cP$, and let $\P^{\pm}$, $\F$ be solutions of (\ref{p4e1})
        defined in (\ref{p4e21}-3). Then
        there exist $\eta \in C_o^{\iy}(\R,\C^2)$ and $U \ss \R^2$ of nonzero measure such that
        for any $(x,y) \in U$ we get
        $$
        \begin{aligned}
            (\eta(x), \P^{\pm}(x,\l,t) \otimes \P^{\mp}(y,\l,t) \eta(y)) & \neq 0 &&
            \text{if $\l \neq \m_n(t)$, $\l \neq \m_n(0)_*$ for any $n \in \Z$;}\\
            (\eta(x), \F(x,\l,t) \otimes \F(y,\l,t) \eta(y)) & \neq 0 &&
            \text{if $\l = \m_n(t) = \m_n(0)_*$ for some $n \in \Z$.}
        \end{aligned}
        $$
    \end{lemma}
    \begin{proof}
        Let $\l \neq \m_n(t)$, $\l \neq \m_n(0)_*$ for any $n \in \Z$. We prove the statement by
        contradiction. Let $(\eta(x), \P^{\pm}(x,\l,t) \otimes \P^{\mp}(y,\l,t) \eta(y)) = 0$
        for almost all $(x,y) \in \R^2$ and for any $\eta \in C_o^{\iy}(\R,\C^2)$.
        Choosing $\eta \in C_o^{\iy}(\R,\C^2)$ such that
        $\eta(x) = \eta(y) =(\begin{smallmatrix} 1 \\ 0 \end{smallmatrix})$ or
        $\eta(x) = \eta(y) = (\begin{smallmatrix} 0 \\ 1 \end{smallmatrix})$ for some
        $x,y \in \R$, we get $\P^{\pm}_1(x,\l,t)\P^{\mp}_1(y,\l,t) = 0$,
        $\P^{\pm}_2(x,\l,t)\P^{\mp}_2(y,\l,t) = 0$.
        Since we can choose such $\eta$ for any $x,y \in \R$, it follows that
        $\P^{\pm}_1(x,\l,t) = \P^{\mp}_2(x,\l,t) = 0$ for any $x \in \R$ and for some
        sign $\pm$. Thus, $\P^{+}(x,\l,t)$ or $\P^{-}(x,\l,t)$ satisfies the Dirichlet boundary
        conditions and $\l = \m_n(t)$ for some $n \in \Z$. So we have a contradiction.
        The second statement is proved similarly.
    \end{proof}
    Second, we prove the needed lemma.
    \begin{lemma} \label{p4l3}
        \begin{enumerate}[i),leftmargin=*]
            \item Let $\l \in \L$, and let $\l \neq \m_n(t)$, $\l \neq \m_n(0)_*$ for each
            $n \in \Z$ and for some $t \in \R$. Then $\l$ is a state of $H_t$ if and only if
            $w(\l,t) = 0$.
            \item Let $\l \in \L$, and let $\l = \m_n(t)$ or $\l = \m_n(0)_*$ for some $n \in \Z$
            and $t \in \R$. Then $\l$ is a state of $H_t$ if and only if $\G(\l,t) = 0$, moreover,
            in this case $\m_n(t) = \m_n(0)_*$.
        \end{enumerate}
    \end{lemma}
    \begin{proof}
        First, we prove that each state of $H_t$, $t \in \R$, is a zero of corresponding function.
        It follows from Lemma \ref{p4l2} that a state of $H_t$ is a pole of $R_{\eta}(x,y,\cdot,t)$
        on the Riemann surface $\L$ for some $\eta \in C_o^{\iy}(\R,\C^2)$ and for any
        $(x,y) \in U$, where $U \ss \R^2$ is some set of
        nonzero measure. Substituting (\ref{p4e3}) in (\ref{p4e25}) for
        $\left|x\right| > \left|y\right|$, we get
        \[ \label{p4e24}
            \begin{aligned}
                &R(x,y,\l,t) = \frac{-1}{w(\l,t)} \P^{+}(x,\l,t) \otimes \P^{-}(y,\l,t) \\
                &= \frac{-1}{w(\l,t)} \lt( \vT (x,\l,t) \otimes \vT (y,\l,t) +
                m_-(\l,0) \vT (x,\l,t) \otimes \F (y,\l,t) \\
                &+ m_+(\l,t)m_-(\l,0) \F (x,\l,t) \otimes \F (y,\l,t)+ m_+(\l,t) \F (x,\l,t) \otimes \vT (y,\l,t) \rt).
            \end{aligned}
        \]

        i) Since $\l \neq \m_n(t)$ and $\l \neq \m_n(0)_*$ for each $n \in \Z$, it follows
        that $\l$ is not a pole of $m_+(\cdot,t)$ and $m_-(\cdot,0)$. Using (\ref{p4e24}), we
        deduce that $\l$ is a pole of $R_{\eta}(x, y, \cdot, t)$ only if $w(\l, t) = 0$.

        ii) Let $\l = \m_n(t) \neq \m_n(0)_*$ for some $n \in \Z$. Then $1/w(\l,t) = 0$,
        $m_+(\l,t)/w(\l,t) =1$ and $m_-(\l,0)/w(\l,t) = 0$. Substituting these identities in
        (\ref{p4e24}), we conclude that $\l$ is not a pole of $R_{\eta}(x,y,\cdot,t)$.
        The case $\l = \m_n(0)_* \neq \m_n(t)$ is considered similarly.

        Let $\l = \m_n(t) = \m_n(0)_*$ for some $n \in \Z$. Then we have
        $$
            \G(\l,t) = \frac{w(\l,t)}{m_+(\l,t) m_-(\l,0)} = \frac{1}{m_-(\l,0)} -
            \frac{1}{m_+(\l,t)} = 0
        $$
        and it follows from (\ref{p4e24}) that $\l$ is a pole of $R_{\eta}(x, y, \cdot, t)$ in
        this case.

        Second, we prove that each zero of the corresponding functions is a state of $H_t$,
        $t \in \R$.

        i) Let $\l \neq \m_n(t)$, $\l \neq \m_n(0)_*$ for any $n \in \Z$ and let $w(\l,t) = 0$.
        Then, by Lemma \ref{p4l12}, there exist $\eta \in C_o^{\iy}(\R,\C^2)$ and $U \ss \R^2$ of
        nonzero measure such that
        $(\eta(x), \P^{\pm}(x,\l,t) \otimes \P^{\mp}(y,\l,t) \eta(y)) \neq 0$
        for any $(x,y) \in U$. It follows that
        $$
            R_{\eta}(x,y,\cdot,t) = \frac{-1}{w(\cdot,t)}
            (\eta(x), \P^{\pm}(x,\cdot,t) \otimes \P^{\mp}(y,\cdot,t) \eta(y))
        $$
        has a pole $\l$ for any $(x,y) \in U$. Thus, by Lemma \ref{p4l2}, $H_t$ has a state $\l$.

        ii) Let $\l = \m_n(t) = \m_n(0)_*$ for some $n \in \Z$. Then we get $\G(\l,t) = 0$ and, by
        Lemma \ref{p4l12}, there exist $\eta \in C_o^{\iy}(\R,\C^2)$ and $U \ss \R^2$ of
        nonzero measure such that
        $(\eta(x), \F(x,\l,t) \otimes \F(y,\l,t) \eta(y)) \neq 0$
        for any $(x,y) \in U$. As above, $R_{\eta}(x,y,\cdot,t)$ has a pole $\l$ for any
        $(x,y) \in U$ and, by Lemma \ref{p4l2}, $H_t$ has a state $\l$.
    \end{proof}
    \begin{remark}
        Let $\l \neq \m_n(t) \neq \m_n(0)_*$ and $m_+(\l,t) \neq 0$, $m_-(\l,0) \neq 0$.
        Then we have $\G(\l,t) = 0$ if and only if $w(\l,t) = 0$. If $\l_0 = \m_n(t_0) = \m_n(0)_*$
        for some $t_0 \in \R$, then we get $m_+(\l,t) \neq 0$, $m_-(\l,0) \neq 0$ for $(\l,t)$
        in some punctured neighbourhood of $(\l_0,t_0)$. It now follows that one can use the
        equation $\G(\l,t) = 0$ instead of $w(\l,t) = 0$ in this neighborhood.
    \end{remark}

    Proposition \ref{p4l4} shows that the resonances of $H_t$ are the eigenvalues of
    another Dirac operator with dislocation. We use this proposition to analyse the
    motion of resonances of $H_t$.

    \begin{proof}[\bf{Proof of Proposition \ref{p4l4}}]
        Note that the operator $\widetilde{H}_t$ is the Dirac operator with the dislocation
        $\t = -t$ and the potential $Q(x) = V(x + t)$. Indeed, we get
        $$
            Q_{\t}(x) = Q(x) \c_-(x) + Q(x+\t) \c_+(x) =
            V(x+t) \c_-(x) + V(x) \c_+(x) = \widetilde{V}_t(x).
        $$
        Thus, by Lemma \ref{p4l3}, $\l$ is a state of $\widetilde{H}_t$ if and only if
        $\l$ is a zero of $\widetilde{w}(\cdot,\t)$ or $\widetilde{\G}(\cdot,\t)$, where the tilde
        indicates that the object corresponds to the operator $\widetilde{H}_t$. Since
        $Q(x) = V(x + t)$ for any $x,t \in \R$, it follows that
        $\widetilde{m}_{\pm}(\l,s) = m_{\pm}(\l,s+t)$ for any $\l \in \L$ and $s,t \in \R$. Using
        this fact and (\ref{p4e5}), we get
        \[ \label{p4e37}
            \begin{aligned}
                \widetilde{w}(\l_*,\t) &= \widetilde{m}_-(\l,\t) - \widetilde{m}_+(\l,0) =
                m_-(\l,0) - m_+(\l,t) = -w(\l,t),\\
                \widetilde{\G}(\l_*,\t) &= \frac{1}{\widetilde{m}_+(\l,0)} -
                \frac{1}{\widetilde{m}_-(\l,\t)} = \frac{1}{m_+(\l,t)} - \frac{1}{m_-(\l,0)} =
                -\G(\l,t).
            \end{aligned}
        \]
        Note also that $\widetilde{\m}_n(\t) = \m_n(0)$, and $\widetilde{\m}_n(0) = \m_n(t)$.
        Thus, we get that $\m_n(t) = \m_n(0)_*$ if and only if
        $\widetilde{\m}_n(\t) = \widetilde{\m}_n(0)_*$. This fact and (\ref{p4e37}) prove the
        proposition.
    \end{proof}
    \begin{remark}
        It follows from this proposition and Theorem \ref{p4t1} that each state of $H_t$ belongs to
        $\g_n^c$ for some $n \in \Z$. Indeed, each eigenvalue of $H_t$ belongs to $\g_n^{(1)}$
        for some $n \in \Z$ and each resonance of $H_t$ is an eigenvalue of $\widetilde{H}_t$,
        which also belongs to $\g_n^{(1)}$ for some $n \in \Z$.
    \end{remark}

    Now we show that each state is a simple zero of $w$ or $\G$. We introduce the following functions
    \[
        S_0(\l,t) = \frac{\vp_1}{\vp_1 + \vp_1^0},\quad
        S_1(\l,t) = \frac{\partial_{\l} \vp_1}{\partial_{\l} \vp_1 + \partial_{\l} \vp_1^0},\quad
        \Omega(\l,t) = \frac{1 - (\vp_2)^2}{\left\| \vp \right\|^2},
    \]
    where $(\l,t) \in \C \ts \R$, and $\vp_i = \vp_i(1,\l,t)$, $\vp_i^0 = \vp_i(1,\l,0)$, $i=1,2$,
    $\vp = \vp(\cdot,\l,t)$. It is easy to see that if
    $\l \neq \m_n(t), \m_n(t)_*, \m_n(0), \m_n(0)_*$ is real, then we get $S_0(\l,t) \in (0,1)$
    and if $\l = \m_n(t) = \m_n(0)_*$, then we get $S_1(\l,t) \in (0,1)$. Lemma 3.5 from
    \cite{MokKor} implies $\sign \Omega(\m_n(t),t) = -(-1)^j$ when $\m_n(t) \in \g_n^{(j)}$
    for some $n \in \Z$, $j=1,2$.
    \begin{lemma} \label{p4l11}
        Let a gap $\g_n$ be open for some $n \in \Z$ and let $\l \in \L$ be a state of $H_t$ for
        some $t \in \R$. Then for any $j = 1,2$ and any choose of sign $\pm$ we get
        \[ \label{p3e7}
            \begin{aligned}
                & & \sign w'_{\l}(\l,t) &= (-1)^j,  &&\text{if $\l \neq \m_n(t)$,
                $\l \in \g_n^{(j)}$},\\
                \G'_{\l}(\l,t)&= -(\Omega(\l,t) S_1(\l,t))^{-1}, & \sign \G'_{\l}(\l,t) &= (-1)^j,
                &&\text{if $\l = \m_n(t)$, $\l \in \g_n^{(j)}$},\\
                w'_z(\l,t) &= -(2\varkappa_n^{\pm} S_0(\l,t))^{-1}, & \sign w'_z(\l,t) &= \pm,
                &&\text{if $\l \neq \m_n(t)$, $\l = \a_n^{\pm}$},\\
                \G'_z(\l,t) &= -(2\kappa_n^{\pm}(t) S_1(\l,t))^{-1}, & \sign \G'_z(\l,t) &= \pm,
                &&\text{if $\l = \m_n(t)$, $\l = \a_n^{\pm}$},
            \end{aligned}
        \]
        where $w'_z(\a_n^{\pm},t) = \left. w'_z(\a_n^{\pm} \mp z^2,t)\right|_{z=0}$, and
        $\G'_z(\a_n^{\pm},t) = \left. \G'_z(\a_n^{\pm} \mp z^2,t)\right|_{z=0}$.
    \end{lemma}
    \begin{proof}
        If $\l \neq \m_n(t)$, then it follows from Lemma \ref{p4l3} that $\l \neq \m_n(0)_*$ and
        thus, by Lemma \ref{p3l1},
        $\sign \partial_{\l} m_+(\l,t) = - \sign \partial_{\l} m_-(\l,0) = (-1)^j$ if
        $\l \in \g_n^{(j)}$. Using the definition of $w$, we get the first line in (\ref{p3e7}).

        Let $\l = \m_n(t)$ and $\l \in \g_n^{(j)}$, $j = 1,2$. Then Lemma \ref{p4l3} implies
        that $\l = \m_n(0)_*$. Using the definition of $\G$, we get
        $$
            \G'_{\l}(\l,t) = \frac{\partial_{\l} \vp_1(1,\l,t) +
            \partial_{\l} \vp_1(1,\l,t)}{2b(\l)} =
            \frac{1}{S_1(\l,t)} \frac{\partial_{\l} \vp_1(1,\l,t)}{2b(\l)} =
            \frac{-1}{S_1(\l,t) \Omega(\l,t)}.
        $$
        Here we used (\ref{p2e19}) and $a(\l,t) = -b(\l)$, which holds true by Lemma \ref{p2l4}.
        Now using $\sign \Omega(\m_n(t),t) = -(-1)^j$, $\m_n(t) \in \g_n^{(j)}$ and
        $S_1(\l,t) \in (0,1)$, we get the second line in (\ref{p3e7}).

        Let $\l = \a_n^{\pm} \neq \m_n(t)$. Then, by Lemma \ref{p4l3}, we have
        $\l \neq \m_n(t)_*, \m_n(0), \m_n(0)_*$. Note that if $f$ is analytic in a neighborhood
        of $0$, then $\left. \partial_z f(z^2) \right|_{z=0} = 0$. Using these facts and asymptotic
        (\ref{p2e12}), we get
        $$
            \left. w'_z(\a_n^{\pm} \mp z^2,t)\right|_{z=0} =
            -(-1)^n \sqrt{\left|2M_n^{\pm}\right|}\left( \frac{1}{\vp_1(1,\a_n^{\pm},t)} +
            \frac{1}{\vp_1(1,\a_n^{\pm},0)} \right) = \frac{-1}{2\varkappa_n^{\pm} S_0(\l,t)}.
        $$
        Since $\l \neq \m_n(t), \m_n(t)_*, \m_n(0), \m_n(0)_*$, we get $S_0(\l,t) \in (0,1)$.
        Let $\b(x,\l)$ be the Pr{\"u}fer angle of $\vp(x,\l,0)$, i.e.
        $$
            \vp(x,\l,0) = \varrho(x,\l) \ma \sin \b(x,\l) \\ \cos \b(x,\l) \am,\qq (x,\l) \in \R \ts \C,
        $$
        where $\varrho(x,\l) > 0$ for any $x \in \R$, $\b(0,\l) = 0$, and $\b(\cdot,\l)$ is a
        continuous function for any $\l \in \C$
        (see e.g. Section 16 in \cite{W87} about Pr{\"u}fer transformation). Thus we get
        $\sign \vp_1(1,\l,0) = \sign \sin \b(1,\l)$. Moreover, for any $n \in \Z$ we have
        $\b(1,\m_n(0)) = -\pi n$ and $\partial_{\l} \beta(1,\l) < 0$
        (see e.g. Section 5.3 in \cite{MokKor}). It follows that if $\m_n(0) \neq \a_n^-$, then
        $\b(1,\a_n^-) \in (-\pi n, -\pi (n-1))$ and if $\m_n(0) \neq \a_n^+$, then
        $\b(1,\a_n^+) \in (-\pi (n+1), -\pi n)$. Then we obtain $\sign \vp_1(1,\l,0) = \mp(-1)^n$
        if $\l = \a_n^{\pm} \neq \m_n(t)$, which yields $\sign \varkappa_n^{\pm} = \mp$.

        Let $\l = \a_n^{\pm} = \m_n(t) = \m_n(0)_*$. Then $a(\l,t) = a(\l,0) = 0$ and
        $\D(\l) = (-1)^n$ imply that $\vp_2(1,\l,t) = \vp_2(1,\l,0) = (-1)^n$.
        Combining (\ref{p2e19}) and (\ref{p2e18}), and substituting
        $\vp_2(1,\l,t) = \vp_2(1,\l,0) = (-1)^n$, we obtain
        $\vt_2(1,\l,t) = -2M_n^{\pm}/\partial_{\l} \vp_1(1,\l,t)$ and
        $\vt_2(1,\l,0) = -2M_n^{\pm}/\partial_{\l} \vp_1(1,\l,0)$. Using these facts and
        (\ref{p2e12}), we get
        $$
            \left. \G'_z(\a_n^{\pm} \mp z^2,t)\right|_{z=0} =
            (-1)^n \sqrt{\left|2M_n^{\pm}\right|}\left( \frac{1}{\vt_2(1,\a_n^{\pm},t)} +
            \frac{1}{\vt_2(1,\a_n^{\pm},0)} \right) = \frac{-1}{2\kappa_n^{\pm}(t) S_1(\l,t)}.
        $$
        Substituting $\vp_2(1,\l,t) = (-1)^n$ in (\ref{p2e19}), we get
        $\sign \kappa_n^{\pm}(t) = \mp$, which yields with $S_1(\l,t) \in (0,1)$
        the fourth line in (\ref{p3e7}).
    \end{proof}

    \subsection{Continuity of states}
    Now we describe local properties of the states. Let a gap $\g_n$ be open for some
    $n \in \Z$ and let $\l \in \L$ be a state of $H_t$ for some $t \in \R$. We introduce a
    function $F_{\l,t}(x,y)$ defined on some open neighborhood $I_1 \ts I_2 \ss \R^2$ of $(0,0)$ by
    \[ \label{p4e36}
        F_{\l,t}(x,y) =
        \begin{cases}
            w(\l+y,t+x) & \text{if $\l \neq \a_n^{\pm}$, $\l \neq \m_n(t)$},\\
            w(\l \mp y^2,t+x) & \text{if $\l = \a_n^{\pm}$, $\l \neq \m_n(t)$},\\
            \G(\l+y,t+x) & \text{if $\l \neq \a_n^{\pm}$, $\l = \m_n(t)$},\\
            \G(\l \mp y^2,t+x) & \text{if $\l = \a_n^{\pm}$, $\l = \m_n(t)$}.\\
        \end{cases}
    \]
    Recall that the class $\mH(I_1,I_2)$ has been defined in Section \ref{p2}.
    \begin{lemma} \label{p4l9}
        Let a gap $\g_n$ be open for some $n \in \Z$, and let $\l \in \L$ be a state of $H_t$ for
        some $t \in \R$. Then there exist an open neighborhood $I_1 \ts I_2 \ss \R^2$ of $(0,0)$
        such that for $F_{\l,t}$ defined on $I_1 \ts I_2$ by (\ref{p4e36}) we get
        $F_{\l,t},\partial_y F_{\l,t} \in \mH(I_1,I_2)$, $F_{\l,t}(0,0) = 0$, and
        $\partial_y F_{\l,t}(0,0) \neq 0$.
    \end{lemma}
    \begin{proof}
        In order to prove that $F_{\l,t},\partial_y F_{\l,t} \in \mH(I_1,I_2)$ for some $I_1$,
        $I_2$ we need the following properties. Since $\p(1,\cdot,0)$ and $b(\cdot)$ are analytic
        functions on $\L$, it follows that $\p(1,\l + y, 0)$ and $b(\l+y)$ (if $\l \neq \a_n^{\pm}$) or
        $\p(1,\l \mp y^2,0)$ and $b(\l \mp y^2)$ (if $\l = \a_n^{\pm}$) are analytic functions of
        $y$ in some open neighborhood of $0$. Thus, it follows from Lemma \ref{a1l1} that they
        and their partial derivatives with respect to $y$ as functions of $(x,y)$ belong to $\mH(I_1,I_2)$ for
        some open neighborhood $I_1 \ts I_2 \ss \R^2$ of $(0,0)$.

        Using Lemma \ref{p6l1}, we get that $\p(1,\l + y,t + x)$ (if $\l \neq \a_n^{\pm}$) or
        $\p(1,\l \mp y^2,t + x)$ (if $\l = \a_n^{\pm}$) as functions of $(x,y)$ belong to
        $\mH(I_1,I_2)$ for some open neighborhood $I_1 \ts I_2 \ss \R^2$ of $(0,0)$. Moreover,
        their partial derivatives with respect to $y$ also belong to $\mH(I_1,I_2)$.

        Let $\l \neq \m_n(t)$. From Lemma \ref{p4l3} it follows that $\l \neq \m_n(0)_*$. Then we
        get $\vp_1(1,\l,t) \neq 0$ and $\vp_1(1,\l,0) \neq 0$. Using the above properties,
        definition of $w$ and Lemma \ref{a1l2}, we have
        $F_{\l,t}, \partial_y F_{\l,t} \in \mH(I_1,I_2)$ for some $I_1$, $I_2$.

        Let $\l = \m_n(t)$. Then, by Lemma \ref{p4l3}, we get $\l = \m_n(0)_*$. It follows from
        Lemma \ref{p2l4} that $a(\l,t) - b(\l) \neq 0$ and $a(\l,0) + b(\l) \neq 0$. Using the
        above properties, definition of $\G$ and Lemma \ref{a1l2}, we have
        $F_{\l,t}, \partial_y F_{\l,t} \in \mH(I_1,I_2)$ for some $I_1$, $I_2$.

        Let $I_1 \ts I_2$ be some open neighborhood of $(0,0)$ and let $F_{\l,t}$ be defined on
        $I_1 \ts I_2$ by (\ref{p4e36}). Since $\l$ is a state of $H_t$, it follows from
        (\ref{p4e36}) and Lemma \ref{p4l3} that $F_{\l,t}(0,0) = 0$ and Lemma \ref{p4l11} implies
        that $\partial_y F_{\l,t}(0,0) \neq 0$.
    \end{proof}

    The following lemma is the main result of this section and it describes the motion of the
    states locally in neighborhood of $\a_n^{\pm}$.
    We prove that a state is a solution of some nonlinear ordinary differential equation which is
    an analogue of the Dubrovin equation for the Dirichlet eigenvalues.

    \begin{lemma} \label{p4l6}
        Let a gap $\g_n$ be open for some $n \in \Z$ and let $\a_n^{\pm}$ be a state of $H_{t_0}$
        for some $t_0 \in \R$ and one of the signs $\pm$. Then there exists a unique function
        $z_n^{\pm} \in \cH^1(I)$, where $I = (t_0-\ve,t_0+\ve)$ for some $\ve > 0$, such that
        $z_n^{\pm}(t_0) = 0$, and:
        \begin{enumerate}[i)]
            \item $\l_n^{\pm}(t) = \a_n^{\pm} \mp {z_n^{\pm}(t)}^2$ is a state of $H_t$ in $\g_n^c$
            for any $t \in I$;
            \item $\l_n^{\pm}(t) \in \L_j$ if and only if $(-1)^j z_n^{\pm}(t) < 0$ for any
            $t \in I$ and $j = 1,2$;
            \item $z_n^{\pm} = 0$ if and only if $I = \R$.
        \end{enumerate}
        Moreover, for almost all $t \in I$ we have
        \[ \label{p4e64}
            \begin{aligned}
                \dot{z}_n^{\pm} &= \frac{1}{w'_z} \left( (q_1-\l_n^{\pm}) + 2m_+q_2 -
                m_+^2(q_1+\l_n^{\pm}) \right), &&\text{if $\m_n(t_0) \neq \a_n^{\pm}$},\\
                \dot{z}_n^{\pm} &= \frac{1}{\G'_z} \left( \frac{q_1-\l_n^{\pm}}{m_+^2} +
                \frac{2q_2}{m_+} - (q_1+\l_n^{\pm}) \right), &&\text{if $\m_n(t_0) = \a_n^{\pm}$},
            \end{aligned}
        \]
        % where $z_n^{\pm} = z_n^{\pm}(t)$, $\l_n^{\pm} = \l_n^{\pm}(t)$, $m_+ = m_+(\l_n^{\pm}(t),t)$, $w'_z = \left. w'_z(\a_n^{\pm} \mp z^2,t) \right|_{z = z_n^{\pm}(t)}$, and $\G'_z = \left. \G'_z(\a_n^{\pm} \mp z^2,t) \right|_{z = z_n^{\pm}(t)}$.
        where
        \begin{gather*}
            z_n^{\pm} = z_n^{\pm}(t),\qq \l_n^{\pm} = \l_n^{\pm}(t),\qq
            w'_z = \left. w'_z(\a_n^{\pm} \mp z^2,t) \right|_{z = z_n^{\pm}(t)},\\
            m_+ = m_+(\l_n^{\pm}(t),t),\qq \G'_z =
            \left. \G'_z(\a_n^{\pm} \mp z^2,t) \right|_{z = z_n^{\pm}(t)}.
        \end{gather*}
    \end{lemma}
    \begin{proof}
        Let $\l_0 = \a_n^{\pm}$ be a state of $H_{t_0}$ for some $t_0 \in \R$ and the sign $\pm$.
        If $\l_0 \neq \m_n(t_0)$, then, by Lemma \ref{p4l3}, $\l_n^{\pm}(t)$ is a state of $H_t$
        in some neighborhood of $\l_0$ if and only if $\l_n^{\pm}(t)$ is a solution of the
        equation $w(\l_n^{\pm}(t),t)=0$. If $\l_0 = \m_n(t_0)$, then, by Lemma \ref{p4l3},
        $\l_n^{\pm}(t)$ is a state of $H_t$ in some neighborhood of $\l_0$ if and only if
        $\l_n^{\pm}(t)$ is a solution of the equation $\G(\l_n^{\pm}(t),t)=0$.

        We introduce a local coordinate $z_n^{\pm}(t) \in \R$ on $\L$ by
        $\l_n^{\pm}(t) = \a_n^{\pm} \mp z_n^{\pm}(t)^2$, where
        $\l_n^{\pm}(t) \in \L_j$ if and only if $(-1)^j z_n^{\pm}(t) < 0$, $j =  1,2$.

        Using (\ref{p4e36}), we get that the equations $w(\l_n^{\pm}(t),t)=0$
        (if $\l_0 \neq \m_n(t_0)$) and $\G(\l_n^{\pm}(t),t)=0$ (if $\l_0 = \m_n(t_0)$) are
        equivalent to $F_{\l_0,t_0}(t-t_0, z_n^{\pm}(t)) = 0$. Since $\l_0$ is a state of $H_{t_0}$,
        by Lemma \ref{p4l9}, $F_{\l_0,t_0}$ is correctly defined on some open neighborhood of $(0,0)$
        and the conditions of Theorem \ref{p7t2} are satisfied for $F_{\l_0,t_0}$. Thus there exist
        an open neighborhood $I$ of $t_0$ and a unique function $z_n^{\pm} \in \cH^1(I)$ such that
        $z_n^{\pm}(t_0) = 0$ and $F_{\l_0,t_0}(t-t_0, z_n^{\pm}(t)) = 0$ for each $t \in I$.
        Therefore, $\l_n^{\pm}(t) = \a_n^{\pm} \mp z_n^{\pm}(t)^2$ is a state of
        $H_t$ for any $t \in I$.

        Differentiating $w(\l_n^{\pm}(t),t)=0$ and $\G(\l_n^{\pm}(t),t)=0$ by $t$, and using
        (\ref{dotm}), (\ref{dotm2}), we get (\ref{p4e64}) for almost all $t \in I$.

        It remains to prove that $z_n^{\pm} = 0$ iff $I = \R$. Let for simplicity
        $\m_n(0) \neq \a_n^{\pm}$, otherwise, the proof is similar but it is needed to use
        $1/m_+(\l,t)$ instead of $m_+(\l,t)$ in the proof. Let $z_n^{\pm}(t) = 0$ for each $t \in I$.
        Then using Lemma \ref{p4l3} and the definition of $w$, we get
        $m_+(\a_n^{\pm},t) = m_-(\a_n^{\pm},0)$ for each $t \in I$. We introduce
        $f(t) = m_+(\a_n^{\pm}, t)$. It is easy to see that there exist $t_1,t_2 \in \R$ such that
        $f \in C(t_1,t_2)$ and $t_1 < t_0 < t_2$. We determine $t_2$ by
        $t_2 = \inf \{ \t > t_0 \mid \m_n(\t) = \a_n^{\pm} \}$ or $t_2 = +\iy$ if there is no
        $\t > t_0$ such that $\m_n(\t) = \a_n^{\pm}$. Similarly, we have
        $t_1 = \sup \{ \t < t_0 \mid \m_n(\t) = \a_n^{\pm} \}$ or $t_2 = -\iy$ if there is no
        $\t < t_0$ such that $\m_n(\t) = \a_n^{\pm}$. Since $f \in C(t_1,t_2)$, it follows that
        $f^{-1}(m_-(\a_n^{\pm},0))$ is a closed set. Let $I_1 \ss (t_1,t_2)$ be the connected
        component of $t_0$ in $f^{-1}(m_-(\a_n^{\pm},0))$. By Lemma \ref{p4l3}, $\a_n^{\pm}$ is a
        state of $H_t$ for each $t \in I_1$. If $I_1 \neq \R$, then there exist
        $t_3 \in \partial I_1$ such that $t_3 \neq \pm \iy$, where $\partial I_1$ is boundary of
        $I_1$. We can apply the first part of the lemma to the point $t_3$. Thus there exists an
        open neighborhood $I_2$ of $t_3$ such that $z_n^{\pm} \in \cH^1(I_2)$ and so on. Moreover,
        since $t_3 \in \partial I_1$, it follows that $I_2 \sm f^{-1}(m_-(\a_n^{\pm},0)) \neq \es$.
        It now follows that there exists $t_4 \in I_2$ such that
        $m_+(\a_n^{\pm},t_4) \neq m_-(\a_n^{\pm},0)$ and hence $z_n^{\pm}(t_4) \neq 0$.
        Let $I = int(I_1) \cup I_2$. Then $I$ is an open neighborhood of $t_0$ such that
        $z_n^{\pm}(I) \neq \{ 0 \}$ and other statements of the lemma are satisfied.
    \end{proof}
    Now we prove the similar lemma when a state is not in some neighborhood of $\a_n^{\pm}$.
    \begin{lemma} \label{p4l10}
        Let a gap $\g_n$ be open for some $n \in \Z$ and let $\l_0 \in \g_n^c$,
        $\l_0 \neq \a_n^{\pm}$, be a state of $H_{t_0}$ for some $t_0 \in \R$.
        Then there exists a unique function $\l_n \in \cH^1(I)$, where $I = (t_0-\ve,t_0+\ve)$
        for some $\ve > 0$, such that $\l_n(t_0) = \l_0$, and $\l_n(t)$ is a state of $H_t$ in
        $\g_n^c$ for any $t \in I$.

        Moreover, for almost all $t \in I$ we have
        \[ \label{p4e34}
            \begin{aligned}
                \dot{\l}_n &= \frac{1}{w'_{\l}} \left( (q_1-\l_n) + 2m_+q_2 -
                m_+^2(q_1+\l_n) \right), &&\text{if $\m_n(t_0) \neq \a_n^{\pm}$},\\
                \dot{\l}_n &= \frac{1}{\G'_{\l}} \left( \frac{q_1-\l_n}{m_+^2} + \frac{2q_2}{m_+}
                - (q_1+\l_n) \right), &&\text{if $\m_n(t_0) = \a_n^{\pm}$},
            \end{aligned}
        \]
        % where $\l_n = \l_n(t)$, $m_+ = m_+(\l_n(t),t)$, $w'_{\l} = \left. w'_{\l}(\l,t) \right|_{\l = \l_n(t)}$, and $\G'_{\l} = \left. \G'_{\l}(\l,t) \right|_{\l = \l_n(t)}$.
        where
        $$
            \l_n = \l_n(t),\qq m_+ = m_+(\l_n(t),t),\qq w'_{\l} =
            \left. w'_{\l}(\l,t) \right|_{\l = \l_n(t)},\qq
            \G'_{\l} = \left. \G'_{\l}(\l,t) \right|_{\l = \l_n(t)}.
        $$
    \end{lemma}
    \begin{proof}
        Let $\l_0 \neq \a_n^{\pm}$ be a state of $H_{t_0}$ for some $t_0 \in \R$. If
        $\l_0 \neq \m_n(t_0)$, then, by Lemma \ref{p4l3}, $\l_n(t)$ is a state of $H_t$ in some
        open neighborhood of $\l_0$ if and only if $\l_n(t)$ is a solution of the equation
        $w(\l_n(t),t)=0$. If $\l_0 = \m_n(t_0)$, then, by Lemma \ref{p4l3}, $\l_n(t)$ is a state
        of $H_t$ in some neighborhood of $\l_0$ if and only if $\l_n(t)$ is a solution of the
        equation $\G(\l_n(t),t)=0$.

        Using (\ref{p4e36}), we get that the equations $w(\l_n(t),t)=0$ (if $\l_0 \neq \m_n(t_0)$),
        and $\G(\l_n(t),t)=0$ (if $\l_0 = \m_n(t_0)$) are equivalent to
        $F_{\l_0,t_0}(t-t_0, \l_n(t)-\l_0) = 0$. Since $\l_0$ is a state of $H_{t_0}$,
        by Lemma \ref{p4l9}, $F_{\l_0,t_0}$ is correctly defined on some open neighborhood of $(0,0)$
        and the conditions of Theorem \ref{p7t2} are satisfied for $F_{\l_0,t_0}$. Thus there exist
        an open neighborhood $I$ of $t_0$ and a unique function $\l_n \in \cH^1(I)$ such that
        $\l_n(t_0) = \l_0$ and $F_{\l_0,t_0}(t-t_0, \l_n(t)-\l_0) = 0$ for each $t \in I$.
        Therefore, $\l_n(t)$ is a state of $H_t$ for any $t \in I$.

        Differentiating $w(\l_n(t),t)=0$ and $\G(\l_n(t),t)=0$ by $t$, and using
        (\ref{dotm}), (\ref{dotm2}), we get (\ref{p4e34}) for almost all $t \in I$.
    \end{proof}

    Now we show that if a state of $H_t$ is close to $\m_n(t)$, then they move synchronously.
    \begin{lemma} \label{p4l8}
        Let a gap $\g_n$ be open for some $n \in \Z$ and let $\l_0 = \m_n(t_0)$ be a state of
        $H_{t_0}$ for some $t_0 \in \R$. Let $\l_n(\cdot) \in \cH^1(I,\g_n^c)$ be a state of $H_t$,
        and $\m_n(\cdot) \in \cH^1(I,\g_n^c)$ be a state of $H^+_t$
        for each $t \in I$, where
        $\m_n(t_0) = \l_n(t_0) = \l_0$ and $I = (t_0 - \ve,t_0 + \ve)$ for some $\ve > 0$.
        \begin{enumerate}[i)]
            \item Let $\l_0 \neq \a_n^{\pm}$. Then there exists $\tilde \ve > 0$ such that for each
            $t \in (t_0 - \tilde \ve, t_0 + \tilde \ve)$ we have
            $$
                \begin{aligned}
                    \l_n(t) &\in \lan \m_n(t_0), \m_n(t) \ran \qq
                    \text{if and only if}\qq \m_n(t) \in \lan \m_n(t_0), \m_n(t_0)_* \ran, \\
                    \l_n(t) &\in \lan \m_n(t), \m_n(t_0) \ran \qq
                    \text{if and only if}\qq \m_n(t) \in \lan \m_n(t_0)_*, \m_n(t_0) \ran,
                \end{aligned}
            $$
            where $\lan \cdot,\cdot \ran$ is the clockwise oriented arc on $\g_n^c$.

            \item Let $\l_0 = \a_n^{\pm}$ and let $\m_n(t)$, $\l_n(t)$ have the following
            form for any $t \in I$
            $$
                \begin{aligned}
                    \m_n(t) &= \a_n^{\pm} \mp \z_n(t)^2,\qq \z_n(t_0) = 0,\qq
                    \m_n(t) \in \L_j\ \ \text{iff}\ \ (-1)^j \z_n(t) < 0,\ \ j=1,2,\\
                    \l_n(t) &= \a_n^{\pm} \mp z_n(t)^2,\qq z_n(t_0) = 0,\qq
                    \l_n(t) \in \L_j\ \ \text{iff}\ \ (-1)^j z_n(t) < 0,\ \ j = 1,2.
                \end{aligned}
            $$
            Then there exists $\tilde \ve > 0$ such that for each
            $t \in (t_0 - \tilde \ve, t_0 + \tilde \ve)$ we have
            $$
                \begin{aligned}
                    z_n(t) &\in (0, \z_n(t))\qq \text{if and only if}\qq \z_n(t) > 0, \\
                    z_n(t) &\in (\z_n(t), 0)\qq \text{if and only if}\qq \z_n(t) < 0.
                \end{aligned}
            $$
        \end{enumerate}
    \end{lemma}

    \begin{proof}
        Since $\l_n(t_0) = \m_n(t_0)$ is a state of $H_t$, it follows from Lemma \ref{p4l3} that
        $\m_n(t_0) = \m_n(0)_*$, and $\l_n(t)$ is a solution of the equation $\G(\l_n(t),t) = 0$.
        We show that a solution of this equation is between $\m_n(t)$ and $\m_n(t_0)$ for each $t$
        in some neighborhood of $t_0$. Since $\G(\cdot,t)$ is continuous, it is enough to show
        $\sign \G(\m_n(t_0),t) = -\sign \G(\m_n(t),t) \neq 0$. Using the definition of $\G$, we get
        \[ \label{p4e43}
            \sign \frac{\G(\m_n(t_0),t)}{\G(\m_n(t),t)} =
            - \sign \frac{\vp_1(1,\m_n(t_0),t)}{\vp_1(1,\m_n(t),0)}
            \frac{a(\m_n(t_0),t) - b(\m_n(t_0))}{a(\m_n(t),0) + b(\m_n(t))}.
        \]
        Since $\vp_1(1,\cdot,t)$ is an entire function, it follows that for any $t \in \R$
        \[  \label{p4e38}
            \vp_1(1,\m_n(t_0),t) = \int_{\m_n(t)}^{\m_n(t_0)} \partial_{\l} \vp_1(1,\l,t) d\l.
        \]
        Note that we have $\partial_{\l} \vp_1(1,\m_n(t),t) \neq 0$ and
        $\partial_{\l} \vp_1(1,\m_n(t_0),0) \neq 0$. By Lemma \ref{p6l1},
        $\partial_{\l} \vp_1(1,\l,t)$, and $\partial_{\l} \vp_1(1,\l,0)$ belong to
        $\mH(I_1, I_2) \ss C(I_1 \ts I_2)$ as functions $(\l,t)$, so that there exist an open
        neighborhood $I_1 \ts I_2 \ss \R^2$ of $(t_0,\m_n(t_0))$ such that they do not change
        their sign for $(t,\l) \in I_1 \ts I_2$. Recall that $\m_n(\cdot) \in \cH^1(I,\g_n^c)$,
        so that choosing a smaller $I_1$, if necessary, we get $\m_n(t) \in I_2$ for each
        $t \in I_1$. It now follows from (\ref{p4e38}) that
        \[  \label{p4e39}
            \sign \vp_1(1,\m_n(t_0),t) = \sign(\m_n(t)-\m_n(t_0)) \partial_{\l} \vp_1(1,\m_n(t),t).
        \]
        Similarly we obtain
        \[ \label{p4e40}
            \sign \vp_1(1,\m_n(t),0) = \sign(\m_n(t_0)-\m_n(t)) \partial_{\l} \vp_1(1,\m_n(0),0),
        \]
        where we used $\vp_1(1,\m_n(0),0) = \vp_1(1,\m_n(t_0),0)$. By Lemma \ref{p6l1},
        \ref{p7l3}, $\partial_{\l} \vp_1(1,\m_n(t),t)$ is a continuous function of $t$
        and, by (\ref{p2e19}), it never vanishes. Thus, we get
        $\sign \partial_{\l} \vp_1(1,\m_n(t),t) = \sign \partial_{\l} \vp_1(1,\m_n(0),0) \neq 0$.
        Using this fact and combining (\ref{p4e39}), (\ref{p4e40}), we get
        $$
            \sign \vp_1(1,\m_n(t_0),t) = -\sign \vp_1(1,\m_n(t),0).
        $$
        Substituting this identity in (\ref{p4e43}), we see that it suffices to prove that
        \[ \label{p4e41}
            \sign (a(\m_n(t_0),t) - b(\m_n(t_0))) = - \sign(a(\m_n(t),0) + b(\m_n(t))).
        \]

        i) Let $\m_n(t_0) \neq \a_n^{\pm}$. Since $a(\cdot,t)$ and $b(\cdot)$ are analytic in a
        neighborhood of $\m_n(t_0)$,~we~get
        \[ \label{p4e44}
            \begin{aligned}
                a(\m_n(t_0),t) - b(\m_n(t_0)) &= a(\m_n(t),t) - b(\m_n(t)) +
                \int_{\m_n(t)}^{\m_n(t_0)} \partial_{\l} (a(\l,t) - b(\l)) d\l,\\
                a(\m_n(t),0) + b(\m_n(t)) &= a(\m_n(t_0),0) + b(\m_n(t_0)) +
                \int_{\m_n(t_0)}^{\m_n(t)} \partial_{\l} (a(\l,0) + b(\l)) d\l.
            \end{aligned}
        \]
        By Lemma \ref{p2l4}, we have $a(\m_n(t),t) - b(\m_n(t)) = -2b(\m_n(t)) \neq 0$ and
        $a(\m_n(t_0),0) + b(\m_n(t_0)) = 2b(\m_n(t_0)) \neq 0$. As above, using
        Lemma \ref{p6l1}, \ref{a1l1}, and \ref{a1l2}, we get that
        $\partial_{\l} (a(\l,t) - b(\l))$, and $\partial_{\l} (a(\l,0) + b(\l))$ are continuous
        function on $I_1 \ts I_2$ and hence they are bounded on $I_1 \ts I_2$. Using these facts
        and reducing $I_1 \ts I_2$, if necessary, we get from (\ref{p4e44})
        \[ \label{p4e45}
            \begin{aligned}
                \sign(a(\m_n(t_0),t) - b(\m_n(t_0))) &= -\sign b(\m_n(t)), \\
                \sign(a(\m_n(t),0) + b(\m_n(t))) &= \sign b(\m_n(t_0)).
            \end{aligned}
        \]
        Since $\m_n(t) \neq \a_n^{\pm}$ and $\m_n(t_0) \neq \a_n^{\pm}$ belong to one sheet of
        $\L$, it follows that $\sign b(\m_n(t)) = \sign b(\m_n(t_0)) \neq 0$. Substituting
        this identity in (\ref{p4e45}), we obtain (\ref{p4e41}), which proves i).

        ii) Now let $\m_n(t_0) = \a_n^{\pm}$ and $\m_n(t)$, $\l_n(t)$ have the form as in the
        statement of the lemma. The functions $a(\a_n^{\pm} \mp z^2, 0)$ and $b(\a_n^{\pm} \mp z^2)$
        are analytic functions of $z$ in a neighborhood of $0$ and $a(\l,t)$ is an analytic
        function of $\l$. Since $b(\m_n(t_0)) = 0$ and $a(\m_n(t_0),0) = 0$, we get
        \[ \label{p4e46}
            \begin{aligned}
                a(\m_n(t_0),t) &= a(\m_n(t),t) +
                \int_{\m_n(t)}^{\m_n(t_0)} \partial_{\l} a(\l,t) d\l,\\
                a(\m_n(t),0) + b(\m_n(t)) &=
                \int_0^{\z_n^{\pm}(t)} \partial_z (a(\m_n(t_0) \mp z^2,0) + b(\m_n(t_0) \mp z^2))dz.
            \end{aligned}
        \]
        As in previous case one can prove that $\partial_{\l} a(\l,t)$ is bounded on
        $I_1 \ts I_2$ and $a(\m_n(t),t) = -b(\m_n(t))$. It follows that
        $$
            a(\m_n(t_0),t) = -b(\m_n(t))(1 + O(\z_n^{\pm}(t)))\ \ \text{as $\z_n^{\pm}(t) \to 0$}.
        $$
        Moreover, using (\ref{p2e12}) and $\left. \partial_z f(z^2) \right|_{z=0} = 0$, where
        $f(\l)$ is an analytic function of $\l$ in an open neighborhood of $0$, we get
        $\sign \partial_z (a(\m_n(t_0) \mp z^2,0) + b(\m_n(t_0) \mp z^2)) = (-1)^n$ and it does
        not change sign for any $z$ in some open neighborhood of $0$. Thus, we get from
        (\ref{p4e46})
        $$
            \begin{aligned}
                \sign(a(\m_n(t_0),t) - b(\m_n(t_0))) &= -\sign b(\m_n(t)) =
                -(-1)^n \sign \z_n^{\pm}(t),\\
                \sign(a(\m_n(t),0) + b(\m_n(t))) &= (-1)^n \sign \z_n^{\pm}(t),
            \end{aligned}
        $$
        which yields (\ref{p4e41}) and proves ii).
    \end{proof}

\section{Proofs of the main theorems} \label{p6}

    \begin{proof}[\bf{Proof of Theorem \ref{p0t1}}]
        In order to prove the theorem we show that there exist exactly two continuous states in
        some open neighborhood of $t = 0$. After that, we show that a state is a continuous function
        of $t$ and it cannot appear or disappear, wherever it is. Thus, we can continuously
        propagate two initial states for all $t \in \R$ without getting any other states.

        Let a gap $\g_n$ be open for some $n \in \Z$. It is easy to see that
        $w(\a_n^{\pm},0) = 0$ (if $\m_n(0) \neq \a_n^{\pm}$) or $\G(\a_n^{\pm},0) = 0$
        (if $\m_n(0) = \a_n^{\pm}$) and $w(\l,0) \neq 0$, $\G(\l,0) \neq 0$ for any
        $\l \in \g_n^c$, $\l \neq \a_n^{\pm}$. Thus, by Lemma \ref{p4l3}, there exist only
        two states $\a_n^{\pm}$ of $H_0$ in $\g_n^c$. From Lemma \ref{p4l6},
        it follows that there exist an open neighborhood $I$ of $0$ and $\l_n^{\pm} \in \cH^1(I)$
        such that $\l_n^{\pm}(0) = \a_n^{\pm}$, $\l_n^{\pm}(t)$ are states of $H_t$ for each
        $t \in I$, and $\l_n^{\pm}(t_0) \neq \a_n^{\pm}$ for some $t_0 \in I$.

        Since Lemma \ref{p4l5} holds true, it follows from Theorem VIII.23 in \cite{Ques1} that
        each eigenvalue of $H_t$ can not appear or disappear inside the gap, moreover, it is a
        continuous function of $t$. By Proposition \ref{p4l4}, this also holds true for each resonance
        of $H_t$.

        We prove that a state cannot disappear at $\a_n^{\pm}$. Let $\l_n(t)$
        be a state of $H_t$ for each $t < t_0$ for some $t_0 \in \R$, and let
        $\l_n(t) \to \a_n^{\pm}$ as $t \to t_0$ for some sign $\pm$. We will prove
        that $\a_n^{\pm}$ is a state of $H_{t_0}$. Let for simplicity $\a_n^{\pm} \neq \m_n(0)_*$,
        in the case $\a_n^{\pm} = \m_n(0)_*$ the proof is similar. First, by Lemma \ref{p4l9},
        the mapping $(x,y) \mapsto w(\a_n^{\pm} \mp y^2, t_0 + x)$ belongs to
        $\mH(I_1, I_2) \ss C(I_1 \ts I_2)$ for some neighborhood $I_1 \ts I_2$ of $(0,0)$.
        Since $\l_n(t) \to \a_n^{\pm}$ as $t \to t_0$, it follows that
        $w(\a_n^{\pm},t_0) = \lim_{t \to t_0} w(\l_n(t),t) = 0$. Thus, by Lemma \ref{p4l3},
        $\l_0$ is a state of $H_{t_0}$. Now using Lemma \ref{p4l3}, we continuously propagate
        this state for $t$ in some neighborhood of $t_0$. Moreover, there exists $t_1 > t_0$ such
        that $\l_n(t_1) \neq \a_n^{\pm}$ and then we can propagate the state as above.
        The proof that a state cannot appear is similar.

        It now follows that we can continuously propagate two initial states for all $t \in \R$
        without getting any other states. Thus, there exist exactly two states $\l_n^{\pm}(t)$
        of $H_t$ in $\g_n^c$ for each $t \in \R$ such that $\l_n^{\pm}: \R \to \g_n^c$ are continuous
        mappings and $\l_n^{\pm}(0) = \a_n^{\pm}$.

        Since the potential $V \in \cP$, we have $V_{t+1} = V_{t}$ and either
        $\l_n^{\pm}(1+t) = \l_n^{\pm}(t)$ or $\l_n^{\pm}(1+t) = \l_n^{\mp}(t)$ for any
        $t \in \R$. In the first case we have
        $\l_n^{\pm}(2+t) = \l_n^{\pm}(1+t) = \l_n^{\pm}(t)$ and in the second case
        $\l_n^{\pm}(2+t) = \l_n^{\mp}(1+t) = \l_n^{\pm}(t)$. Hence
        $\l_n^{\pm}(2+t) = \l_n^{\pm}(t)$ for any $t \in \R$, i.e. $\l_n^{\pm}(\cdot)$ is 2-periodic.

        Using Lemma \ref{p4l6} and \ref{p4l10}, we get that for each $t_0 \in \R$ there
        exist an open neighborhood $I$ of $t_0$ such that $\l_n^{\pm}(\cdot) \in \cH^1(I)$, or
        $\l_n^{\pm}(\t) = \a_n^{\pm} \mp (z_n^{\pm}(\t))^2$, $\t \in I$, and
        $z_n^{\pm}(\cdot) \in \cH^1(I)$. It now follows from the definition of $\cH^1(2\T,\g_n^c)$
        that $\l_n^{\pm} \in \cH^1(2\T, \g_n^c)$.

        i) We consider the case when $\m_n(0) \neq \a_n^{\pm}$ and $\m_n(0) \in \L_1$. In the other
        cases the proof is similar. Then $\l_n^{+}(0) = \a_n^{+} \in \lan \m_n(0),\m_n(0)_* \ran$
        and $\l_n^{-}(0) = \a_n^{-} \in \lan \m_n(0)_*,\m_n(0) \ran$. It is easy to see that
        $\l_n^{\pm}(t_0)$ can leave the required arc only if it collide with $\m_n(t_0)$ or
        $\m_n(0)_*$ for some $t_0 \in \R$. By Lemma \ref{p4l3}, this is possible if and only if
        $\m_n(t_0) = \m_n(0)_*$. It follows from Lemma \ref{p4l8} that in this case $\l_n^{+}(t)$
        and $\l_n^{-}(t)$ also lie on the different arcs $\lan \m_n(t),\m_n(0)_* \ran$ and
        $\lan \m_n(0)_*, \m_n(t) \ran$ for $t$ in a neighborhood of $t_0$.

        ii) The statement follows from Lemma \ref{p4l3}.

        iii) It follows from i) of this theorem that $\l_n^{\pm}(t)$ are between $\m_n(t)$ and
        $\m_n(0)_*$ and change places when $\m_n(t)$ makes one revolution around $\g_n^c$. This
        imply that if $\m_n(t)$ makes $r(n)$ revolutions when $t$ runs through $[0,1)$, then
        $\l_n^{\pm}(t)$ make $r(n)/2$ revolutions.

        iv) The statement follows from iii).
    \end{proof}

    \begin{proof}[\bf{Proof of Theorem \ref{p0t2}}]
        Let $V \in \cP$ and let a gap $\g_n$ be open for some $n \in \Z$. Suppose that
        $$
            \sign(q_1(t) + \a_n^-) = \sign(q_1(t) + \a_n^+) = \const \neq 0
            \text{ for almost all $t \in [0,1]$}.
        $$
        Then the conditions of Theorem 2.3 from \cite{MokKor} are satisfied and it follows that
        $\m_n(t)$ makes exactly $\left|n\right|$ complete revolutions around $\g_n^c$ when $t$
        runs through $[0,1]$. Then Theorem \ref{p0t1}, iii) implies that $\l_n^{\pm}(t)$ make
        $\left|n\right|/2$ revolutions when $t$ runs through $[0,1]$.
    \end{proof}

    \begin{proof}[\bf{Proof of Theorem \ref{p0t5}}]
        Let a gap $\g_n$ be open for some $n \in \Z$ and let $\l_0 = \m_n(0) = \a_n^{\pm}$. Then it
        follows from Lemma \ref{p4l6} that there exists $z \in \cH^1((-\ve,\ve))$ for some
        $\ve > 0$ such that $z(0) = 0$, $\l(t) = \l_0 \mp z^2(t)$ is a state of $H_t$ and
        $\l(t) \in \L_j$ if and only if $(-1)^j z(t) < 0$, $j=1,2$ for any $t \in (-\ve,\ve)$.
        Let $\tilde \ve < \ve$. Then these properties hold for any $t \in [-\tilde \ve,\tilde \ve]$.
        Now we prove that asymptotic (\ref{p1e7}) holds true.
        We denote $\left. \G'_{z}(\l_0 \mp z^2,t_1) \right|_{z=z(t_2)}$ by
        $\G'_z(\l(t_2),t_1)$. Integrating (\ref{p4e64}), we get
        \[ \label{p4e48}
                z(t) = -\frac{1}{\G'_{z}(\l_0,0)}\int_{0}^t (q_1(\t) + \l_0) d\t + \eta_1(t),
        \]
        where
        $$
            \begin{aligned}
                \eta_1(t) &= \frac{1}{\G'_{z}(\l_0,0)}
                \int_{0}^t \left(\frac{q_1(\t) - \l_0}{m_+(\l(\t),\t)^2} +
                \frac{2q_2(\t)}{ m_+(\l(\t),\t)} - (\l(\t)-\l_0)
                \left( 1+\frac{1}{m_+(\l(\t),\t)^2} \right) \right) d\t \\
                &- \int_{0}^t \frac{g_{\l}(m_+(\l(\t),\t),\t)}{m_+(\l(\t),\t)^2}
                \left( \frac{1}{\G'_{z}(\l(\t),\t)} - \frac{1}{\G'_{z}(\l_0,0)}\right)d\t.
            \end{aligned}
        $$
        It follows from Lemma \ref{p4l11} that $-1/\G'_{z}(\l_0,0) = \kappa_n^{\pm}(0)$.
        In order to get (\ref{p1e8}), we prove that $\eta_1(t) = O(Q_{n}^{\pm}(t)W_{n}^{\pm}(t))$
        as $t \to 0$. In the proof we consider all asymptotics when $t \to 0$. We introduce
        $$
            Z(t) = \max_{\t \in [0,t]} \left| z(\t) \right|,\ \ t \in \R.
        $$
        As in proof of Lemma \ref{p4l9}, one can show that
        $F(t,z) = \partial_z^2 \G(\l_0 \mp z^2,t) \in \mH(I,I) \ss C(I \ts I)$, where $I \ts I$
        is some open neighborhood $(0,0)$. This fact and $\G'_z(\l_0,0) \neq 0$ yield that there
        exists a constant $C > 0$ such that
        $\left| \partial_z \left( \G'_z(\l_0 \mp z^2,t)\right)^{-1} \right| \leq C$ for each
        $(z,t) \in I \ts I$. Then the application of the mean value theorem yields
        \[ \label{p4e50}
            \max_{\t \in [0,t]} \left|  \frac{1}
            {\G'_{z}(\l(\t),\t)} - \frac{1}{\G'_{z}(\l_0,\t)}\right| = O(Z(t)).
        \]
        Using Lemma \ref{p6l1}, one can find $\partial_{t} (\G'_{z}(\l_0,t)^{-1})$.
        Integrating this derivative, we get
        \[ \label{p4e51}
            \max_{\t \in [0,t]} \left|  \frac{1}{\G'_{z}(\l_0,\t)} -
            \frac{1}{\G'_{z}(\l_0,0)}\right| = O(W_n^{\pm}(t)).
        \]
        Thus, combining (\ref{p4e50}) and (\ref{p4e51}), we get
        \[ \label{p4e52}
            \max_{\t \in [0,t]} \left|  \frac{1}{\G'_{z}(\l(\t),\t)} -
            \frac{1}{\G'_{z}(\l_0,0)}\right| = O(Z(t) + W_n^{\pm}(t)).
        \]
        Using Lemma \ref{a1l2}, \ref{p6l1}, it is easy to see that
        $F(t,z) = 1/m_+(\l_0 \mp z^2,t) \in \mH(I,I)$, where $I \ts I$ is some open
        neighborhood of $(0,0)$. Thus in the same way as we proved (\ref{p4e50})
        one can prove that
        \[ \label{p4e54}
            \max_{\t \in [0,t]} \left|  \frac{1}{m_+(\l(\t),\t)} -
            \frac{1}{m_+(\l_0,\t)} \right| = O(Z(t)).
        \]
        Combining (\ref{p4e54}) and (\ref{p3e6}), we obtain
        \[ \label{p4e55}
            \max_{\t \in [0,t]} \left|  \frac{1}{m_+(\l(\t),\t)} \right| =
            O(Z(t) + Q_n^{\pm}(t)),
        \]
        which yields
        \[ \label{p4e66}
                \int_{0}^t \left(\frac{q_1(\t) - \l_0}{m_+(\l(\t),\t)^2} +
                \frac{2q_2(\t)}{ m_+(\l(\t),\t)} \right) d\t = O((Z(t) + Q_n^{\pm}(t)) W_n^{\pm}(t)).
        \]
        Using (\ref{p4e55}) and $\l(\t) - \l_0 = \mp z(\t)^2$, we get
        \[ \label{p4e68}
            \int_{0}^t\left|(\l(\t)-\l_0)(1+m_+(\l(\t),\t)^{-2})\right| d\t =
            O(Z(t)^2(Z(t) + Q_n^{\pm}(t))\left|t\right|).
        \]
        It follows from (\ref{dotm2}) that
        $$
            \begin{aligned}
                \int_{0}^t \frac{g_{\l}(m_+(\l(\t),\t),\t)}{m_+(\l(\t),\t)^2} d\t &=
                -\int_{0}^t \left(\frac{q_1(\t) - \l_0}{m_+(\l(\t),\t)^2} +
                \frac{2q_2(\t)}{ m_+(\l(\t),\t)} - (q_1(\t) + \l_0) \right) d\t \\
                &+ \int_{0}^t (\l(\t)-\l_0) \left( 1+m_+(\l(\t),\t)^{-2} \right) d\t.
            \end{aligned}
        $$
        Substituting (\ref{p4e66}) and (\ref{p4e68}) in this formula, we get
        \[ \label{p4e57}
            \begin{aligned}
                \int_{0}^t \left| \frac{g_{\l}(m_+(\l(\t),\t),\t)}{m_+(\l(\t),\t)^2} \right| d\t &=
                O(Q_n^{\pm}(t)) + O((Z(t) + Q_n^{\pm}(t)) W_n^{\pm}(t))\\
                &+O(Z(t)^2(Z(t) + Q_n^{\pm}(t))\left|t\right|).
            \end{aligned}
        \]
        Estimating $Z(t)$ from (\ref{p4e48}), and using (\ref{p4e52}), (\ref{p4e66}-9), we obtain
        $$
            \begin{aligned}
                Z(t) &= Q_n^{\pm}(t) + O((Z(t) + Q_n^{\pm}(t)) W_n^{\pm}(t)) + O(Z(t)^2(Z(t) +
                Q_n^{\pm}(t))\left|t\right|)\\
                &+(O(Q_n^{\pm}(t)) + O((Z(t) + Q_n^{\pm}(t)) W_n^{\pm}(t)) + O(Z(t)^2(Z(t) +
                Q_n^{\pm}(t))\left|t\right|))O(Z(t) + W_n^{\pm}(t)),
            \end{aligned}
        $$
        which yields $Z(t) = O(Q_n^{\pm}(t))$. Recall that $Q_n^{\pm}(t) \leq W_n^{\pm}(t)$.
        Substituting $Z(t) = O(Q_n^{\pm}(t))$ in (\ref{p4e52}), (\ref{p4e66}-9) and estimating
        $\eta_1(t)$, we get $\eta_1(t) = O(Q_n^{\pm}(t)W_n^{\pm}(t))$.

        Let now $\l_0 = \a_n^{\pm} \neq \m_n(0)$, and let $\l(t) = \l_0 \mp z(t)^2 \in \g_n^c$
        be a state of $H_t$ for $t$ in some neighborhood of $0$ given by Lemma \ref{p4l6}.
        Using the definition of $g_{\l}$ and $g_n^{\pm}$, we get
        \[ \label{p4e69}
            g_{\l}(x,t) = g_n^{\pm}(x,t) - (\l - \l_0)(1 + x^2).
        \]
        We denote $\left. w'_{z}(\l_0 \mp z^2,t_1) \right|_{z=z(t_2)}$ by $w'_z(\l(t_2),t_1)$ and
        $m_+(\l_0,0)$ by $m_+^0$. Integrating (\ref{p4e34}) and using (\ref{p4e69}), we get
        \[ \label{p4e72}
                z(t) = -\frac{1}{w'_{z}(\l_0,0)}\int_{0}^t g_n^{\pm}(m_+^0,\t) d\t + \eta_2(t),
        \]
        where
        $$
            \begin{aligned}
                \eta_2(t) &= \int_{0}^t \frac{g_n^{\pm}(m_+^0,\t)-g_n^{\pm}(m_+(\l(\t),\t),\t)}
                {w'_{z}(\l_0,0)} d\t
                - \int_{0}^t \frac{(\l(\t)-\l_0)(1+m_+(\l(\t),\t)^2)}{w'_{z}(\l_0,0)} d\t \\
                &- \int_{0}^t g_{\l(\t)}(m_+(\l(\t),\t),\t) \left( \frac{1}{w'_{z}(\l(\t),\t)} -
                \frac{1}{w'_{z}(\l_0,0)}\right)d\t.
            \end{aligned}
        $$
        It follows from Lemma \ref{p4l11} that $-1/w'_{z}(\l_0,0) = \varkappa_n^{\pm}$. In order
        to get (\ref{p1e8}), we prove that $\eta_2(t) = O(G_n^{\pm}(t)W_n^{\pm}(t))$ as $t \to 0$.
        As above we consider all asymptotics when $t \to 0$ and introduce $Z(t)$. Similarly, one
        can prove that
        \[ \label{p4e70}
            \max_{\t \in [0,t]} \left|  \frac{1}{w'_{z}(\l(\t),\t)} - \frac{1}{w'_{z}(\l_0,0)}\right|
            = O(Z(t) + W_n^{\pm}(t)),
        \]
        and
        \[ \label{p4e71}
            \max_{\t \in [0,t]} \left|  m_+(\l(\t),\t) - m_+(\l_0,\t) \right| = O(Z(t)).
        \]
        Combining (\ref{p4e71}) and (\ref{p3e3}), we obtain
        \[ \label{p4e74}
            \max_{\t \in [0,t]} \left|  m_+(\l(\t),\t) - m_+^0 \right| = O(Z(t) + G_n^{\pm}(t)).
        \]
        Using the mean value theorem, we get
        \begin{multline} \label{p4e73}
                \int_{0}^t \left| g_n^{\pm}(m_+^0,\t) -
                g_n^{\pm} (m_+(\l(\t),\t),\t) \right| d\t \leq\\
                \max_{\t \in [0,t]} \left|  m_+(\l(\t),\t) - m_+^0 \right|
                \int_{0}^t \max_{x \in [m_+(\l(\t),\t), m_+^0]}
                \left| \partial_x g_n^{\pm}(x,\t) \right| d\t.
        \end{multline}
        Substituting (\ref{p4e74}) and (\ref{p3e2}) in (\ref{p4e73}), we have
        \[ \label{p4e76}
                \int_{0}^t \left| g_n^{\pm}(m_+(\l(\t),\t),\t) - g_n^{\pm}(m_+^0,\t)\right| d\t =
                O((Z(t) + G_n^{\pm}(t)) W_n^{\pm}(t)).
        \]
        Above we have proved that $m_+(\a_n^{\pm} \mp z^2,t_0+\t)$ belongs to $\mH^1(I,I)$ for some
        open $I$ and then it follows from Lemma \ref{p7l3} that $m_+(\l(\t),\t)$ is an
        absolutely continuous function of $\t$, which yields
        \[ \label{p4e75}
            \int_{0}^t\left|(\l(\t)-\l_0)(1+m_+(\l(\t),\t)^2)\right| d\t = O(Z(t)^2 \left| t \right|).
        \]
        Substituting (\ref{p4e76}) and (\ref{p4e75}) in (\ref{p4e69}), we get
        \[ \label{p4e77}
                \int_{0}^t \left| g_{\l(\t)}(m_+(\l(\t),\t),\t) \right| d\t =
                O(G_n^{\pm}(t)) + O((Z(t) + G_n^{\pm}(t)) W_n^{\pm}(t)) + O(Z(t)^2 \left| t \right|).
        \]
        Estimating $Z(t)$ from (\ref{p4e72}), and using (\ref{p4e70}), (\ref{p4e76}-18), we obtain
        $$
            \begin{aligned}
                Z(t) &= G_n^{\pm}(t) + O((Z(t) + G_n^{\pm}(t)) W_n^{\pm}(t)) +
                O(Z(t)^2 \left| t \right|)\\
                &+ (O(G_n^{\pm}(t)) + O((Z(t) + G_n^{\pm}(t)) W_n^{\pm}(t)) +
                O(Z(t)^2 \left| t \right|))O(Z(t) + W_n^{\pm}(t)),
            \end{aligned}
        $$
        which yields $Z(t) = O(G_n^{\pm}(t))$. Since
        $\left| ax^2 + bx + c\right| \leq \max \{1, x^2 \}
        (\left| a \right| + \left| b \right| + \left| c \right|)$ for any $a,b,c,x \in \R$, it
        follows that $G_n^{\pm}(t) = O(W_n^{\pm}(t))$. Substituting these asymptotics in
        (\ref{p4e70}), (\ref{p4e76}-18) and estimating $\eta_2(t)$, we get
        $\eta_2(t) = O(G_n^{\pm}(t)W_n^{\pm}(t))$.
    \end{proof}
    Now we get the asymptotics of
    $\l_n^{\pm}(t)$ or $z_n^{\pm}(t)$ as $t \to t_0$ for any $t_0 \in \R$.
    Recall that we have
    defined $g_{\l}$, $G_{\l}$, $Q_{\l}$, and $W_{\l}$ in (\ref{p6e2}). We introduce also
    $U_{\l}(t_0,t) = \left|t-t_0\right| + W_{\l}(t_0,t)$ for any $\l \in \L$, $t,t_0 \in \R$.
    \begin{theorem} \label{p6t1}
        Let $V \in \cP$ and let a gap $\g_n$ be open for some $n \in \Z$. Let $\l_0 \in \g_n^c$ be a
        state of $H_{t_0}$ for some $t_0 \in \R$.

        i) If $\l_0 \neq \a_n^{\pm}$, then
        there exists $\l \in \cH^1([t_0-\ve,t_0+\ve])$ for some $\ve > 0$ such that $\l(t_0) = \l_0$,
        $\l(t)$ is a state of $H_t$ in $\g_n^c$ for any $t \in [t_0-\ve,t_0+\ve]$, and we get
        for $t \to t_0$
        $$
            \begin{aligned}
                \l(t) &= \l_0 + \Omega(\l_0,t_0) S_1(\l_0,t_0) \int_{t_0}^t (q_1(\t) + \l_0) d\t +
                O(Q_{\l_0}(t_0,t)U_{\l_0}(t_0,t)), &&\text{if $\l_0 = \m_n(t_0)$},\\
                \l(t) &= \l_0 - \frac{1}{w'_{\l}(\l_0,t_0)} \int_{t_0}^t g_{\l_0}(m_+(\l_0,t_0),\t) d\t +
                O(G_{\l_0}(t_0,t)U_{\l_0}(t_0,t)), &&\text{if $\l_0 \neq \m_n(t_0)$}.
            \end{aligned}
        $$

        ii) If $\l_0 = \a_n^{\pm}$, then there exist
        $z \in \cH^1([-\ve,\ve])$  for some $\ve > 0$ such that $z(0) = 0$ and:
        \begin{enumerate}[1)]
            \item $\l(t) = \a_n^{\pm} \mp (z(t))^2$ are states of $H_t$ in $\g_n^c$
            for any $t \in [-\ve,\ve]$;
            \item $\l(t) \in \L_j$ if and only if $(-1)^j z(t) < 0$ for any
            $t \in [-\ve,\ve]$ and $j = 1,2$.
        \end{enumerate}

        Moreover, we get for $t \to t_0$
        $$
            \begin{aligned}
                z(t) &= 2\kappa_n^{\pm}(t) S_1(\l_0,t_0) \int_{t_0}^t (q_1(\t) + \l_0) d\t +
                O(Q_{\l_0}(t_0,t)W_{\l_0}(t_0,t)), &&\text{if $\l_0 = \m_n(t_0)$},\\
                z(t) &= 2\varkappa_n^{\pm} S_0(\l_0,t_0) \int_{t_0}^t g_{\l_0}(m_+(\l_0,t_0),\t) d\t +
                O(G_{\l_0}(t_0,t)W_{\l_0}(t_0,t)), &&\text{if $\l_0 \neq \m_n(t_0)$}.
            \end{aligned}
        $$
    \end{theorem}
    \begin{proof}
        Acting as in the proof of Theorem \ref{p0t5}, we prove this theorem.
    \end{proof}

    \begin{proof}[\bf{Proof of Theorem \ref{p0t6}}]
        Let $V \in \cP$ and let a gap $\g_n$ be open for some $n \in \Z$. Suppose that
        $\left\| V \right\|_{\iy} < 2\left|\a_n^{\pm}\right|$, i.e.
        $q_1(t)^2 + q_2(t)^2 < (\a_n^j)^2$ for almost all $t \in \R$ and for any $j = \pm$.
        It follows that $\sign (q_1(t) + \a_n^{\pm}) = \sign \a_n^{\pm}$ for almost all
        $t \in \R$. Using the definition of $g_{\l}$, we get for any $\l,x,t \in \R$
        $$
            g_{\l}(x,t) = (q_1(t) + \l)\left( \left( x - \frac{q_2(t)}{q_1(t)+\l}\right)^2 +
            \frac{\l^2-q_1(t)^2-q_2(t)^2}{(q_1(t) + \l)^2}\right),
        $$
        which yields that $\sign g_{\l}(x,t) = \sign \a_n^{\pm}$ for any $x \in \R$, $\l \in \g_n^c$
        and for almost all $t \in \R$. Thus for any $t,t_0 \in \R$ and $\l_0 \in \g_n^c$ we get
        \[ \label{p6e3}
            \left| \int_{t_0}^t g_{\l_0}(m_+(\l_0,t_0),\t) d\t \right| = G_{\l_0}(t_0,t),\quad
            \left| \int_{t_0}^t (q_1(\t) + \l_0) d\t \right| = Q_{\l_0}(t_0,t).
        \]
        Substituting (\ref{p6e3}) in asymptotics from Theorem \ref{p6t1}, it is easy to see that
        each state $\l(t) \in \g_n^c$ of $H_t$ is strictly monotone function of $t$ in neighborhood
        of any $t_0 \in \R$ and $\l(t)$ changes the sheets when $\l(t_0) = \a_n^{\pm}$.
        Using Lemma \ref{p4l11} and asymptotics from Theorem \ref{p6t1}, we get that $\l(t)$ runs
        clockwise around $\g_n^c$ if $\a_n^{\pm} > 0$ and runs counterclockwise if $\a_n^{\pm} < 0$.
        Since $\sign (q_1(t) + \a_n^{\pm}) = \sign \a_n^{\pm}$ for almost all $t \in \R$,
        it follows from Theorem \ref{p0t2} that $\l(t)$ makes $\left| n \right|/2$ complete
        revolutions around $\g_n^c$ when $t$ runs through $[0,1]$.
    \end{proof}

    \begin{proof}[\bf{Proof of Theorem \ref{p0t4}}]
        Since $V \in \cP_{e}$, it follows that $\m_n(0) = \a_n^{j}$ and $\n_n(0) = \a_n^{-j}$
        for some $n \in \Z$, and $j = +$ or $-$. It is easy to see that in this case
        $m_+(\a_n^{-j},0) = 0$. Using this identity, and $q_1 = 0$, we get
        $q_1(t) + \a_n^{j} = \a_n^{j}$ and
        $g_n^{-j}(m_+(\a_n^{-j},0),t) = \a_n^{-j}$ for any $t \in \R$, $n \in \Z$, $j = \pm$.
        Using this identities and asymptotics (\ref{p1e7}-4), we get for $t \to 0$
        \[ \label{p6e6}
            \begin{aligned}
                \sign z_n^{j}(t) &= \sign (\kappa_n^{j}(0) \a_n^{j} t) = -j\sign(\a_n^{j} t),\\
                \sign z_n^{-j}(t) &= \sign (\varkappa_n^{-j} \a_n^{-j} t) = j\sign(\a_n^{-j} t).
            \end{aligned}
        \]
        Since $q_1 = 0$, it follows that $\sign \a_n^{-} = \sign \a_n^{+}$ if $n \neq 0$.
        Using (\ref{p6e6}), we get that there exists $\ve > 0$ such that
        $\sign z_n^{-}(t) = - \sign z_n^{+}(t)$ for any $1 \leq \left| n \right| \leq N$, and
        $t \in (0,\ve)$, i.e. there exist one eigenvalue and one resonance if $\g_n \neq \es$.

        Now we prove that in this case $\g_0 = \es$. It is easy to see that if $q_1 = 0$,
        and $\l = 0$, then the solutions of (\ref{p2e14}) for $\l = 0$ have the form
        $$
            \vt(x,0,t) = \ma e^{\int_0^x q_2(t+\t)d\t} \\ 0 \am,\qq
            \vp(x,0,t) = \ma 0 \\ e^{-\int_0^x q_2(t+\t)d\t} \am,\qq
            x,t \in \R.
        $$
        Since $V \in \cP_{e}$, it follows that $\int_0^1 q_2(\t) d\t = 0$. Thus,
        $\vt(\cdot,0,t)$ and $\vp(\cdot,0,t)$ are 1-periodic, which yields
        $\a_0^{+} = \a_n^- = 0$, i.e. $\g_0 = \es$.
    \end{proof}

    \begin{proof}[\bf{Proof of Corollary \ref{p0t8}}]
        Using the gap length mapping for the Hill operator from \cite{Kor99}, we
        construct a potential $p \in C^1(\T)$ such that $p(x) = p(1-x)$, $x \in \T$, and the
        operator $hy = -y'' + py$ acting on $L^2(\R)$ has $N$ first open gaps in the spectrum.
        Using the Riccati mapping (see e.g. \cite{Miura}), we construct
        $q \in C^2(\T)$ such that $p = q' + q^2$, and $q(x) = - q(1-x)$, $x \in \T$.
        If follows from the remark after Corollary \ref{p0t8}, that if $H$ is the
        Dirac operator with the potentials $q_1 = 0$ and $q_2 = q$, then $\s(H^2) = \s(h)$.
        Thus, for the operator $H$ we have $\g_n \neq \es$, $1 \leq |n| \leq N$.
        By Theorem \ref{p0t4}, there exists $\ve > 0$ such that each $H_t$, $t \in (0,\ve)$,
        has exactly one eigenvalue in $\g_n$, $1 \leq |n| \leq N$, and $\g_0 = \es$.
    \end{proof}

    \begin{proof}[\bf{Proof of Theorem \ref{p0t7}}]
        Let $V \in \cP_{e}$ such that
        $q_1(x) = c \chi_{(0,\d) \cup (1-\d,1)}(x)$, $x \in [0,1)$, for some $c, \d > 0$,
        where $\chi_I(x)$ is a characteristic function of the set $I$. Since
        $V \in \cP_{e}$, it follows that $\m_n(0) = \a_n^{j}$, $\n_n(0) = \a_n^{-j}$ for any
        $n \in \Z$ and for some $j = \pm$.
        As in proof of Theorem \ref{p0t4} we get $q_1(t) + \a_n^{j} = c+\a_n^{j}$ and
        $g_n^{-j}(m_+(\a_n^{-j},0),t) = -c+\a_n^{-j}$ for any $n \in \Z$, $j = \pm$ and for almost
        all $t \in \R$. Using this identities and asymptotics (\ref{p1e7}-4), we get for $t \to 0$
        $$
            \begin{aligned}
                \sign z_n^{j}(t) &= \sign (\kappa_n^{j}(0) (c+\a_n^{j}) t) =
                -j\sign t \sign (c+\a_n^{j}),\\
                \sign z_n^{-j}(t) &= \sign (-\varkappa_n^{-j} (c-\a_n^{-j}) t) =
                -j\sign t \sign(c-\a_n^{-j}).
            \end{aligned}
        $$
        Thus if $\left| \a_n^{\pm} \right| < c$ for some $n \in \Z$, then we get
        $\sign z_n^{j}(t) = \sign z_n^{-j}(t) = -j$ for sufficiently small $t > 0$, i.e.
        if the gap $\g_n \neq \es$, then we have two eigenvalues if $\m_n(0) < \n_n(0)$ ($j=-$)
        and two resonances if $\m_n(0) < \n_n(0)$ ($j=+$).

        Now we show that for any $N > 0$ there exists $c,\d > 0$ such that
        there exist at least $2N+1$ gaps on the interval $(-c,c)$. We introduce the Hilbert
        space $\el2$ of a sequence $y = (y_n)_{n \in \Z}$ equipped with the norm
        $\| y \|_{\el2}^2 = \sum_{n \in \Z} \left| y_n \right|^2 < \iy$. Now we introduce the gap
        length mapping $\cG: \cP \to \el2 \os \el2$ by $V \mapsto \cG(V)$, where
        $\cG_n = (\cG_{n1}, \cG_{n2}) \in \R^2$, with the length $|\cG_n|^2 =
        \cG_{n1}^2 + \cG_{n2}^2 = |\g_n|/2$ and the components are given by
        $$
            \cG_{n1} = \frac{\a_n^- + \a_n^+}{2}-\m_n,\qq \cG_{n2} = -(-1)^j\left|
            |\cG_n|^2 - \cG_{n1}^2 \right|^{1/2},
            \qq \m_n \in \L_j,\ \ j=1,2.
        $$
        Now we need the following theorem.
        \begin{theorem}[Theorem 1.1 from \cite{Kor05b}] \label{p6t2}
            The mapping $\cG: \cP \to \el2 \os \el2$ is a real analytic isomorphism between $H$ and
            $\el2 \os \el2$ and the following estimates are fulfilled:
            \[ \label{p6e4}
                \frac{1}{\sqrt{2}}\| \cG(V) \|_{\el2} \leq \| V \|_{\cP} \leq
                2\| \cG(V) \|_{\el2} \left( 1 + \| \cG(V) \|_{\el2} \right).
            \]
        \end{theorem}

        Using (\ref{p6e4}), we get for any $i_o \in \Z$
        \[ \label{p6e13}
            \sum_{i = i_o}^{i_o + N} \left| \g_i \right| \leq \sqrt{N} \lt|
            \sum_{i = i_o}^{i_o + N} \left| \g_i \right|^2 \rt|^{1/2} \leq
            2\sqrt{N}\| \cG(V) \|_{\el2} \leq
            2\sqrt{2N(2\d c^2 + \left\| q_2 \right\|^2)}.
        \]
        It follows from Theorem 3.2 in \cite{KK} that $\left| \s_n \right| \leq \pi$ for any $n
        \in \Z$. Using this estimate and (\ref{p6e13}), we get for any $i_o \in \Z$
        $$
            \sum_{i = i_o}^{i_o + N} \left( \left| \g_i \right| + \left| \s_i \right| \right) \leq
            A = 2\sqrt{2N(2\d c^2 + \left\| q_2 \right\|^2)} + \pi N.
        $$
        Now we consider inequality $A < c$, which yields
        $$
            c > \pi N,\quad 2\d c^2 + \left\| q_2 \right\|^2 <
            \frac{(c-\pi N)^2}{8N}.
        $$
        It is easy to see that for sufficiently large $c$ and sufficiently small $\d > 0$ these
        inequalities are holds true. Thus, we have $(-A, A) \ss (-c,c)$, which yields that
        $\g_n \ss (-c,c)$ for any $n \in \Z$ such that $|n+i_o| \leq N$, where $i_o$ is a number
        of the gap closest to zero.
    \end{proof}

    \begin{proof}[\bf{Proof of Theorem \ref{p0t9}}]
        Let $N > 0$ and let $(y_n)_{n \in \Z} \in \el2$ be a sequence such that
        \[ \label{p6e12}
            y_n > 0,\qq n \in \Z;\qq \sum_{n \in \Z} y_n^2 n^2 < \iy;\qq
            2\pi N + 2 < \sum_{n \in \Z} y_n < \iy.
        \]
        Using Theorem \ref{p6t2}, we construct a potential
        $\widetilde{V} \in \cP_e$ such that for $\widetilde{H} = H^0 + \widetilde{V}$ we have
        \[ \label{p6e14}
            |\g_n(\widetilde{H})| = y_n,\qq \m_n(\widetilde{H}) = \a_n^-(\widetilde{H}),\qq
            \n_n(\widetilde{H}) = \a_n^+(\widetilde{H}),\qq n \in \Z.
        \]
        Moreover, $\widetilde{V} \in \cH^1(\T,M_2(\R))$ since $|\g_n|$ decrease fast enough as
        $n \to \iy$ (see e.g. Theorem~11 in \cite{DjMi}). Thus, we can use the trace formula
        to calculate $\tilde{q}_1(0)$ (see e.g. Theorem 3.3 in \cite{GreGui}).
        Substituting (\ref{p6e14}) in the trace formula, and using (\ref{p6e12}), we obtain
        \[ \label{p6e9}
            \tilde{q}_1(0) = \frac{1}{2}\sum_{n \in \Z} (\a_n^+(\widetilde{H}) +
            \a_n^-(\widetilde{H}) - 2\m_n(\widetilde{H})) =
            \frac{1}{2}\sum_{n \in \Z} y_n > \pi N+1.
        \]
        It is well known that $\a_n^{\pm} = \pi n + o(1)$ as $n \to \iy$ (see e.g. Lemma 3.2
        in \cite{Kor05b}). So that there exists $M > 0$ such that
        \[ \label{p6e15}
            \g_n(\widetilde{H}) \ss (\pi n -1,\pi n+1),\qq |n| \geq M.
        \]
        We introduce a unitary operator $U: L^2(\R,\C^2) \to L^2(\R,\C^2)$ given by
        $(Uf)(x) = e^{\pi x J}f(x)$, $x \in \R$. We need the following properties of $U$
        (see Theorem 1.1 in \cite{Kor01}) :

            \no  i) $U \cP_e \ss \cP_e,\qq U(0) = U^*(0) = I_2$;\\
            \no ii)  Let $H^1 = H^0 + V$, $H^2 = H^0 + U V U^*$ for some $V \in \cP$.

            Then for each
            $n \in \Z$:
            \[ \label{p6e17}
                \m_n(H^2) = \m_{n-1}(H^1) + \pi,\qq \n_n(H^2) = \n_{n-1}(H^1) + \pi,
                \qq \g_n(H^2) = \g_{n-1}(H^1) + \pi.
            \]

        Let $V = U^{M+N+1} \widetilde{V} (U^*)^{M+N+1} \in \cP_e$ and $H = H^0 + V$.
        Then it follows from (\ref{p6e15}-27) that
        \[ \label{p6e16}
            \g_n(H) = \g_{n-N-M-1}(\widetilde{H}) + \pi(N+M+1) \ss (\pi n -1, \pi n + 1),\qq
            |n| \leq N.
        \]
        Due to (\ref{p6e9}), and $V(0) = \widetilde{V}(0)$, we have $q_1(0) > \pi N+1$, which
        yields from (\ref{p6e16}) that $\g_n(H) \ss (-q_1(0),q_1(0))$, $|n| \leq N$. It follows
        from (\ref{p6e14}), and (\ref{p6e17}) that $\m_n(H) < \n_n(H)$. Now using
        arguments from the proof of Theorem \ref{p0t7}, we get $\ve > 0$ such that each $H_t$,
        $t \in (0,\ve)$, has exactly two eigenvalues in $\g_n(H)$, $|n| \leq N$.
    \end{proof}

    \begin{proof}[\bf{Proof of Theorem \ref{p0t10}}]
        Let $q_1 = m > 0$. Then the solution
        $\psi(x,\l,t)$ of equation (\ref{p2e14}) satisfies:
        $$
            \psi(x,-\l,t) = M(\l) \psi(x,\l,t) M^{-1}(\l),\qq M(\l) = \ma m-\l & 0 \\ 0 & m+\l \am
        $$
        for each $(x,\l,t) \in \R \ts \C \ts \R$. Thus we have for each $(\l,t) \in \C \ts \R$
        $$
            \D(-\l) = \D(\l),\qq \vp_1(1,-\l,t) = \frac{m-\l}{m+\l} \vp_1(1,\l,t),
            \qq \vt_2(1,-\l,t) = \frac{m+\l}{m-\l} \vt_2(1,\l,t),
        $$
        which yields $\a_n^{\pm} = -\a_{-n}^{\mp}$, $n \in \Z$, and if $\m_n(t) \neq \pm m$
        ($\n_n(t) \neq \pm m$) for some $n \in \Z$ and $t \in \R$, then $\m_n(t) = -\m_{-n}(t)$
        ($\n_n(t) = -\n_{-n}(t)$). If $q_1 = m$ and $\l = m$, then equation (\ref{p2e14}) have
        the form
        $$
            \begin{cases}
                y_2'(x,\l,t) + q_2(x+t) y_2(x,\l,t) = 0 \\
                -y_1'(x,\l,t) + q_2(x+t) y_1(x,\l,t) = 2 m y_2(x,\l,t).
            \end{cases}
        $$
        If $\l=-m$, then the equation has the similar form. Thus, we can explicitly integrate the
        equation, and then the solution $\psi(x,\l,t)$ of equation (\ref{p2e14}) for $\l = \pm m$
        has the following form
        \[ \label{p6e18}
            \psi(\cdot,m,t) = \ma  F^+ & G^+ \\
            0 & F^+ \am,\qq
            \psi(\cdot,-m,t) = \ma F^- & 0 \\
            G^- & F^- \am,
        \]
        where $F^{\pm}(x) = e^{\mp \int_0^x q_2(t+\t) d \t}$ and
        $G^{\pm}(x) = -2 m F^{\mp}(x) \int_0^x F^{\pm2}(\t)d\t$, $x \in \R$.
        Due to $V \in \cP_e$, we get
        $\int_0^1 q_2(t+\t) d\t = 0$, which yields that $F(\cdot)$ is 1-periodic. Thus, it follows
        from (\ref{p6e18}) that $\pm m = \a_{\pm 2n}^{\pm}$ or $\pm m = \a_{\pm 2n}^{\mp}$
        for some $n \geq 0$.
        Moreover, by (\ref{p6e18}), $\vp_1(1,-m,t) = 0$ and $\vt_2(1,m,t) = 0$, so that
        $-m = \m_{-n}(t)$ and $m = \n_n(t)$ for any $t \in \R$. At last, using (\ref{p6e18}), we see
        that the winding number of the curve $[0,1] \ni x \mapsto \vp(x,-m,\t) \in \R^2 \sm \{ 0 \}$
        equals zero, which yields that $n = 0$ (see e.g. Proposition 7.3 in \cite{GreKap09}).
        Due to $-\a_0^- = \a_0^+$ it follows that $j = +$. Thus, we obtain $-m = \m_0(t)$,
        $m = \n_0(t)$ for any $t \in \R$, and $\pm m = \a_{0}^{\pm}$.
    \end{proof}

\section{Appendix} \label{p5}
    We need the following simple lemmas. Recall that $\mH(I_1,I_2)$ has been defined in Section \ref{p2}.
    \begin{lemma} \label{a1l1}
        Let $f$ be real-analytic on $I = [-1,1]$ and a function $G$ be defined on $I \ts I$ by
        $G(x,y) = f(y)$, $(x,y) \in I \ts I$. Then we get $G, G'_y \in \mH(I, I)$.
    \end{lemma}
    \begin{lemma} \label{a1l2}
        Let $G,G'_y,H,H'_y \in \mH(I,I)$, where $I=[-1,1]$, and let $c \in \R$. Then
        $G+H$, $GH$, $c G$, and their partial derivatives with respect to $y$ belongs to $\mH(I,I)$.
        Moreover, if $H(0,0) \neq 0$, then there exists $J \ts J$ an open bounded neighborhood of
        $(0,0)$ such that $G/H, \partial_y \left( G/H \right) \in \mH(J,J)$.
    \end{lemma}
    \begin{lemma} \label{p7l3}
        Let $F,F'_y \in \mH(I,I)$, where $I = [-1,1]$, and let $f \in \cH^1(I)$ such that
        $f(I) \ss I$. Then $F(\cdot,f(\cdot)) \in \cH^1(I)$.
    \end{lemma}
    % \begin{lemma}\label{p7l2}
        % Let $f: I \to \R$ be measurable, where $I = [-1,1]$. Let for each $x \in I$ there exists an open interval $I_x \ss I$ (in the subspace topology on $I$) such that $x \in I_x$, and $\left. f \right|_{I_x} \in \cH^1(I_x)$. Then we have $f \in \cH^1(I)$.
    % \end{lemma}
    Now we give the specific implicit function theorem.
    \begin{theorem}[Theorem 5.5 in \cite{MokKor}] \label{p7t2}
        Let $F, F'_y \in \mH(I,I)$, where $I = [-1,1]$, and let
        \begin{enumerate}[i)]
            \item $F(0,0) = 0$,
            \item $F'_{y}(0,0) \neq 0$.
        \end{enumerate}
        Then there exist an open neighborhood $U_1 \times U_2 \ss I \times I$ of $(0, 0)$ and a
        unique $f \in \cH^1(U_1)$ such that $f(U_1) \ss U_2$ and for any point
        $(x,y) \in U_1 \times U_2$
        $$
            F(x,y) = 0 \text{ if and only if $y = f(x)$}.
        $$
    \end{theorem}

    In order to apply the implicit function theorem, we need specific properties of solutions
    of the Dirac equation. Recall that $\p(x,\l,t)$ is the fundamental solution of the shifted
    Dirac equation (\ref{p2e14}) and $\p(x,\l) = \p(x,\l,0)$ is the fundamental solution of the
    Dirac equation~(\ref{p2e1}).

    We say that a matrix-valued function is from $\cH^1(I)$ or $\mH(I_1,I_2)$ if each of its
    components belongs to $\cH^1(I)$ or $\mH(I_1,I_2)$. Recall that the dot denotes the derivative
    with respect to $t$, i.e. $\dot{u} = du/dt$. Let $[A,B] = AB - BA$ be the commutator of two
    matrices.

    \begin{lemma} \label{p6l1}
        Let $I_1,I_2 \ss \R$ be open bounded intervals and let $F:I_1 \ts I_2 \to \R$ be defined
        by $F(t,\l) = \p(1,\l, t)$, $(t,\l) \in I_1 \ts I_2$. Then we have
        $F,F'_{\l} \in \mH(I_1,I_2)$. Moreover, for almost all $t \in \R$ and for each
        $\l \in \C$ we get
        \begin{align}
            \dot{\p}(1, \l, t) &= \left[ J(V(t)-\l), \p(1, \l, t) \right], \label{p7e5} \\
            \partial_{t} \p'_{\l}(1, \l, t) &= \left[ J(V(t)-\l), \p'_{\l}(1, \l, t) \right] -
            \left[ J, \p(1, \l, t) \right]. \label{p6e11}
        \end{align}
    \end{lemma}
    Substituting the potential $V$ in (\ref{p7e5}), we obtain the following explicit formula
    \[ \label{p6e1}
        \dot{\p}(1,\l, t) =
        \ma
            \dot{\vt}_1 & \dot{\vp}_1 \\
            \dot{\vt}_2 & \dot{\vp}_2 \\
        \am =
        \ma
            \vp_1(q_1 - \l) - \vt_2(q_1+\l) & 2\vp_1q_2 - 2a\left(q_1 + \l\right) \\
            - 2\vt_2q_2 + 2a \left( q_1 - \l \right) & -\vp_1(q_1 - \l) + \vt_2(q_1 + \l)
        \am,
    \]
    where $a = a(\l,t)$, $q_i = q_i(t)$, $\vt_i = \vt_i(1,\l,t)$, and $\vp_i = \vp_i(1,\l,t)$,
    $i =1,2$.

    \footnotesize
    \no {\bf Acknowledgments.} Our study was supported by the RFBR grant No. 19-01-00094.

\end{document}